\PassOptionsToPackage{prologue,dvipsnames}{xcolor}
\documentclass[acmsmall,nonacm]{acmart}

\newif\ifshort
\shortfalse

\usepackage{quiver}
\usepackage{amsmath}
\usepackage{bussproofs} 
    \EnableBpAbbreviations
\usepackage{mathpartir} 
\usepackage{mathtools}
\usepackage{stmaryrd}
    \EnableBpAbbreviations
\usepackage{multirow}

\newcommand{\R}{\mathbb{R}}
\newcommand{\F}{\mathbb{F}}

\newcommand{\Set}{\textbf{Set}}
\newcommand{\denot}[1]{\llbracket {#1} \rrbracket}
\newcommand{\pdenot}[1]{\llparenthesis {#1} \rrparenthesis}
\newcommand{\pdenotid}[1]{{\llparenthesis {#1} \rrparenthesis}_{id}}
\newcommand{\pdenotap}[1]{{\llparenthesis {#1} \rrparenthesis}_{ap}}

\newcommand{\Bel}{\textbf{Bel}}
\newcommand{\Del}{\textbf{Del}}
\newcommand{\sgn}[1]{sgn({#1})}
\newcommand{\stepid}{\Downarrow_{id}}
\newcommand{\stepap}{\Downarrow_{ap}}

\newcommand{\Uid}{U_{id}}
\newcommand{\Uap}{U_{ap}}

\newcommand{\tensor}{\otimes}
\DeclareMathOperator{\dom}{dom}

\usepackage{hyperref}
\usepackage{amsthm}
\usepackage[noabbrev,capitalize,nameinlink]{cleveref}

\theoremstyle{definition}
\newtheorem{theorem}{Theorem}[section]
\newtheorem{lemma}[theorem]{Lemma}

\newtheorem{definition}{Definition}[section]

\newtheorem{remark}{Remark}

\newtheoremstyle{named}{}{}{\itshape}{}{\bfseries}{.}{.5em}{\thmnote{#3}#1}
\theoremstyle{named}
\newtheorem*{namedtheorem}{}

\newcommand{\bea}{\textsc{\textbf{Bean}}}
\newcommand{\LangS}{$\Lambda_S$}

\newcommand{\inl}{\textbf{inl }}
\newcommand{\inr}{\textbf{inr }}
\newcommand{\op}{\textbf{op}}
\newcommand{\slet}[3]{\mathbf{let} \ {#1} = {#2} \ \mathbf{in} \ {#3}}
\newcommand{\dlet}[3]{\mathbf{dlet} \ {#1} = {#2} \ \mathbf{in} \ {#3}}
\newcommand{\case}[3]{\mathbf{case} \ {#1} \ \mathbf{of}\ ({#2} \ | \ {#3})}
\newcommand{\add}[2]{\mathbf{add} \ {#1} \ {#2}}
\newcommand{\sub}[2]{\mathbf{sub} \ {#1} \ {#2}}
\newcommand{\mul}[2]{\mathbf{mul} \ {#1} \ {#2}}
\newcommand{\dmul}[2]{\mathbf{dmul} \ {#1} \ {#2}}
\newcommand{\fdiv}[2]{\mathbf{div} \ {#1} \ {#2}}

\newcommand{\stexttt}[1]{\texttt{\small{#1}}}
\newcommand{\iso}{\xrightarrow{\,\smash{\raisebox{-0.65ex}{\ensuremath{\scriptstyle\sim}}}\,}}

\newcommand{\unit}{\textbf{unit}}
\newcommand{\num}{\textbf{num}}

\usepackage{enumitem}
\usepackage{pdflscape}
\usepackage{lscape}
\usepackage{geometry}
\usepackage{subcaption}
\usepackage{colortbl}
\usepackage{tabularray}

\usepackage[T1]{fontenc}
\usepackage[scaled]{beramono}
\usepackage{adjustbox}
\usepackage{verbatim}
\definecolor{verbgray}{gray}{0.9}
   {\par\noindent\adjustbox{margin=1ex,
    bgcolor=verbgray,margin=0ex \medskipamount}
    \bgroup\minipage\linewidth\verbatim\small}%
   {\endverbatim\endminipage\egroup}

\usepackage{listings}
\lstdefinelanguage{bean}{
    mathescape=true,
    morekeywords=[2]{let, in, case, of, inl, inr, dlet},
    morekeywords=[3]{add, dmul, mul, sub, div},
    morekeywords=[4]{InnerProduct2, PolyVal, Horner, LinSolve, 
        ScaleVec, SVecAdd, InnerProduct, GMatVecMul, MatVecMul, SMatVecMul, SVecAdd,
        DotProd2, MatVecEx,}
    columns=fullflexible,
    sensitive,
    basicstyle =\linespread{2},
    keepspaces=true,
    showstringspaces=false,
    breaklines=true,
    breakatwhitespace=false,
    basicstyle=\ttfamily\footnotesize,
    keywordstyle=[2]{\color{Mahogany}},
    keywordstyle=[3]{\color{Plum}},
    keywordstyle=[1]{\color{MidnightBlue}},
    morecomment=[l][\color{Gray}]{\/},
    xleftmargin=.05\textwidth, 
    xrightmargin=.05\textwidth
   }[keywords,comments,strings]%
\DeclareTextFontCommand{\linec}{\sffamily\selectfont}

\begin{document}

\title{\bea{}: A Language for Backward Error Analysis}

\author{Ariel E. Kellison}
\orcid{0000-0003-3177-7958}
\affiliation{%
  \institution{Cornell University}
  \city{Ithaca}
  \country{USA}
}
\email{ak2485@cornell.edu}

\author{Laura Zielinski}
\orcid{0009-0006-6516-1392}
\affiliation{%
  \institution{Cornell University}
  \city{Ithaca}
  \country{USA}
}
\email{lcz8@cornell.edu}

\author{David Bindel}
\orcid{0000-0002-8733-5799}
\affiliation{%
  \institution{Cornell University}
  \city{Ithaca}
  \country{USA}
}
\email{bindel@cornell.edu}

\author{Justin Hsu}
\orcid{0000-0002-8953-7060}
\affiliation{%
  \institution{Cornell University}
  \city{Ithaca}
  \country{USA}
}
\affiliation{%
  \institution{Imperial College London}
  \city{London}
  \country{United Kingdom}
}
\email{justin@cs.cornell.edu}

\begin{abstract} 
\emph{Backward error analysis} offers a method for assessing the quality of
numerical programs in the presence of floating-point rounding errors.  However,
techniques from the numerical analysis literature for quantifying backward
error require substantial human effort, and there are currently no tools or automated
methods for statically deriving sound backward error bounds. To address this
gap, we propose \bea{}, a typed first-order programming language designed to
express quantitative bounds on backward error.  \bea{}'s type system combines a
graded coeffect system with strict linearity to soundly track the flow of
backward error through programs.  We prove the soundness of our system using a
novel categorical semantics, where every \bea{} program denotes a triple of
related transformations that together satisfy a backward error guarantee. 

To illustrate \bea{}'s potential as a practical tool for automated backward
error analysis, we implement a variety of standard algorithms from numerical
linear algebra in \bea, establishing fine-grained backward error bounds via
typing in a compositional style. We also develop a prototype implementation of
\bea{} that infers backward error bounds automatically. Our evaluation shows
that these inferred bounds match worst-case theoretical relative backward error
bounds from the literature, underscoring \bea{}'s utility in validating a key
property of numerical programs: \emph{numerical stability}.
\end{abstract}

\begin{CCSXML}
<ccs2012>
   <concept>
       <concept_id>10003752.10003790.10002990</concept_id>
       <concept_desc>Theory of computation~Logic and verification</concept_desc>
       <concept_significance>500</concept_significance>
       </concept>
   <concept>
       <concept_id>10002950.10003714.10003715</concept_id>
       <concept_desc>Mathematics of computing~Numerical analysis</concept_desc>
       <concept_significance>500</concept_significance>
       </concept>
 </ccs2012>
\end{CCSXML}

\ccsdesc[500]{Theory of computation~Logic and verification}
\ccsdesc[500]{Mathematics of computing~Numerical analysis}

\keywords{Floating point, Backward error, Linear type systems}

\maketitle

% Proposed outline:
% 1. intro 
% 2. technical overview (background/system overview)
% 3. lens language, type system 
% 4. lens cat and semantics
% 5. soundness
% 6. examples 
% 7. related work 
% 8. conclusion 

\section{Introduction}
It is well known that the floating-point implementations of mathematically
equivalent algorithms can produce wildly different results. This discrepancy
arises due to floating-point rounding errors, where small inaccuracies are
introduced during computation because real numbers are represented with limited
precision. In numerical analysis, the property of \emph{numerical stability} is
used to distinguish implementations that produce reliable results in the
presence of rounding errors from those that do not.  While several notions of
numerical stability exist, \emph{backward stability} is a particularly crucial
property in fields that rely on numerical linear algebra
\cite{Higham:2002:Accuracy,Demmel:1992:BLAS3}, such as scientific computing,
engineering, computer graphics, and machine learning.  

Essentially, a floating-point program is considered backward stable if its
results match those that would be obtained from an exact arithmetic
computation---free of rounding errors---on slightly perturbed inputs.  This
property is valuable because it means that any errors introduced during
execution have the same effect as small changes in the input data.
Consequently, analyzing the accuracy of the computed result reduces to
understanding how small perturbations in the inputs affect the output,
independent of the implementation details.

\emph{Backward error analysis} determines whether a program is backward stable
by quantifying how much the input to an exact version must be perturbed to
match the floating-point result.  The size of this perturbation is the
\emph{backward error}, and in \bea{}, we aim to bound the \emph{relative}
backward error. If this error is small, then the program is considered backward
stable. In a slogan \cite{Trefethen:1997:book}: 
\begin{center}
A backward stable program gives exactly the right solution to nearly the right problem.
\end{center}
Many fields rely on backward stability to ensure accurate, reliable results,
but no tools currently exist to statically derive sound backward error bounds.
In contrast, several tools support the static analysis of \emph{forward error
bounds}, providing direct measures of program accuracy
\cite{Rosa1,Rosa2,REAL2FLOAT,VCFLOAT2,FPTaylor,DAISY,Fluctuat,PRECISA}.
\ifshort
\else
In situations where backward error bounds are required, programmers often compute
the backward error dynamically by implementing rudimentary heuristics. These
methods provide empirical estimates of the backward error in place of sound
rigorous bounds, and often have a higher computational cost than running the
original program \cite{Corless:2013:Analysis}. 
\fi

\paragraph{\textbf{Challenges}} 
The challenges associated with designing static analysis tools that derive
sound backward error bounds are twofold. First, \emph{many programs are not backward
stable}: for these programs, the rounding errors produced during execution
cannot be interpreted as small perturbations to the input of an exact
arithmetic version of the program. Surprisingly, even seemingly straightforward 
computations can lack backward stability. Second, when they do exist, 
\emph{backward error bounds are
generally not preserved under composition}, and the conditions under which
composing backward stable programs yields another backward stable program are
poorly understood \cite{beltran:2023:forward,Bornemann:2007:backwardcompose}.

\paragraph{\textbf{Solution}}
Our work demonstrates that the challenges in designing a static analysis tool
for backward error analysis can be effectively addressed by leveraging concepts
from \emph{bidirectional programming languages}~\citep{Bohannon:2008:lenses}
along with a linear type system and a graded coeffect
system~\citep{Brunel:2014:coeffects}. In this approach, every expression in the
language denotes two forward transformations---a floating-point program and its
associated exact arithmetic function where all operations are performed in
infinite precision without any rounding errors---together with a backward
transformation, which captures how rounding errors appearing in the solution
space can be propagated back as perturbations to the input space. Backward error
bounds characterizing the size of these perturbations are tracked using a
\emph{coeffect graded comonad}. Finally, the \emph{linear} type system ensures
that program variables carrying backward error are never duplicated, providing a
sufficient condition for preserving backward stability under composition.
Structuring the language around these core ideas allows us to capture standard
numerical primitives and their backward error bounds, and to bound the backward
error of large programs in a compositional way. 

Concretely, we propose \bea{}, a programming
language for {\sc{\textbf{b}}}ackward {\sc{\textbf{e}}}rror
{\sc{\textbf{a}}}{\sc{\textbf{n}}}alysis, which features a linear coeffect type
system that tracks how backward error flows through programs, and ensures that
well-typed programs have bounded backward error. In developing \bea{} we 
tackled several foundational problems, resulting in the following technical
contributions:

\paragraph{\textbf{Contribution 1: The \bea{} language} 
(\Cref{sec:language}, \Cref{sec:examples}).} 
We introduce \bea{}, a first-order language with numerical primitives and
a graded coeffect type system designed for tracking backward error. The type
system in \bea{} tracks per-variable backward error bounds, corresponding to
so-called \emph{componentwise} backward error bounds described in the numerical
analysis literature. We demonstrate how \bea{} can establish backward error
bounds for various algorithms from the numerical analysis literature, through
typing.

\paragraph{\textbf{Contribution 2: A \bea{} implementation} 
(\Cref{sec:implementation})} 
We provide a prototype implementation of \bea{} that can automatically infer relative
backward error bounds, marking the first tool to statically derive such bounds.
We translate several large benchmarks into \bea{} and demonstrate that our
implementation infers useful error bounds and scales to large numerical
programs.

\paragraph{\textbf{Contribution 3: Semantics and backward error soundness}
(\Cref{sec:metatheory}).} Building on the literature on bidirectional
programming languages and lenses, we introduce a general semantics describing
computations amenable to backward error analysis based on \emph{backward error
lenses}. We propose the category of backward error lenses 
(\Bel), and give \bea{} a semantics in this category. We conclude our main 
\emph{backward error soundness} theorem for \bea{}---which says that the execution 
of every well-typed program satisfies the backward 
error bound that the type system assigns it---from the properties of the 
morphisms in this category.

\section{Background and Overview}\label{sec:overview}
Our point of departure from other static analysis tools for floating-point
rounding error analysis is our focus on deriving backward error bounds rather
than forward error bounds. We adopt the following notation: floating-point
approximations and perturbed data will be denoted by a tilde. For instance,
$\tilde{f}$ denotes the floating-point approximation of a real-valued function
$f$, and $\tilde{x}$ denotes a slight perturbations of the data $x$.

\subsection{Backward Error and Backward Stability}
Let $\tilde{f}:\F^n\to\F^m$ be a program computing a function $f:\R^n\to\R^m$,
where $\F\subset\R$ is a finite set of floating-point numbers. Given
$x\in\F^n$ and an output $\tilde{y}=\tilde{f}(x)$, a forward error analysis
bounds the distance between $\tilde{y}$ and the exact result $y = f(x)$. In
contrast, a backward error analysis finds $\tilde{x}\in\R^n$ such that
$\tilde{y} = f(\tilde{x})$, and bounds the distance between $x$ and
$\tilde{x}$. 
\begin{figure}
\begin{center}
\begin{tikzpicture}
  \draw[draw=blue!20, thick,fill=blue!20, opacity = 0.6] (0,0) 
  ellipse (1.25cm and 1.5cm);
  \draw[draw=green!20, thick,fill=green!20, opacity = 0.6] (5,0) 
  ellipse (1.25cm and 1.5cm);
  \node[label=below:{Input Space $\R^n$}] (A) at (0,2.2) {};
  \node[label=below:{Output Space $\R^m$}] (B) at (5,2.2) {};
  % Draw and label points inside the carrier set
  \node[circle,fill=black,inner sep=0pt,minimum size=3pt,
  label=above:{$x$}] (A) at (0,0.5) {};
  \node[circle,fill=black,inner sep=0pt,minimum size=3pt,
  label=below:{$\tilde{x}$}] (B) at (0,-0.5) {};
  \node[circle,fill=black,inner sep=0pt,minimum size=3pt] (C) at (5,0) {};
  \node[align=center,anchor=north] (lab) at (5.4,0.75) 
    {$\tilde{{f}}(x)$ \\ = \\ $f(\tilde{x})$};
  \draw (0.1,0.5)  edge[bend left, ->, thick] node[above] 
  {$\tilde{{f}}$} (4.9,0.1);
  \draw (0.1,-0.5) edge[bend right, ->, thick] node[above] {${{f}}$} 
  (4.9,-0.1);
  \draw[dashed,thick,draw = red] node[left] {${\delta x}$} (A) -- (B);
\end{tikzpicture}
\end{center}
\caption{An illustration of backward error.
  The function $\tilde{f}$ represents a floating-point implementation of the
  function $f$. Given the points $\tilde{x} \in \R^n$ and $x \in \F^n \subset
  \R^n$ such that $\tilde{f}(x) = f(\tilde{x})$, the backward error is the
  distance $\delta x$ between $x$ and $\tilde{x}$. }
\label{fig:BEA}
\end{figure}
An illustration of the backward error is given in \Cref{fig:BEA}. If the
backward error is small for every possible input, then an implementation is
said to be \emph{backward stable}:

\begin{definition}(Backward Stability)\label{def:bstab}
  A floating-point implementation $\tilde{f} : \F^n \rightarrow \F^m$ of a
  real-valued function ${f} : \R^n \rightarrow \R^m$ is \emph{backward stable}
  if for every input $x \in \F^n$, there exists an input 
  $\tilde{x} \in \R^n$ such that 
  \begin{equation} 
    f(\tilde{x}) = \tilde{f}(x) \text{ and } d(x,\tilde{x}) \le \alpha u
  \end{equation} 
  where $d :\R^n \times \R^n \rightarrow \R_{\geq 0}\cup\{\infty\}$ is an
  extended metric, $\alpha$ is a small constant, and $u$ is the \emph{unit
  roundoff}---the maximum relative error when rounding to the nearest
  $\F$-representable number. For example, for IEEE binary64 with
  round-to-nearest and any tie-breaking rule, we have $u= 2^{-53}$.
\end{definition}

Forward error can arise from two sources: the problem's conditioning or the
program's stability. An \emph{ill-conditioned} problem is highly sensitive to
changes in its input, which can amplify the effect of rounding errors. A
\emph{well-conditioned} problem is not sensitive to its input, but if it is
implemented with an \emph{unstable} program, the results can still be inaccurate
due to rounding error introduced by the computation. Forward error does not
distinguish between these two sources of error, but backward error provides a
way to separate them. Their relationship is governed by the \emph{condition
number}:
\begin{align}
  \text{forward error $\le$ condition number $\times$ backward error}.
\end{align}
\ifshort
\else
A formal definition of the condition number is given in \Cref{def:cnum}. 
\fi

By automatically deriving sound backward error bounds that indicate the
backward stability of programs, \bea{} addresses a significant gap in current
tools for automated error analysis. To quote Dianne P. O'Leary
\cite{Oleary:2009:book}: ``Life may toss us some ill-conditioned problems, but
there is no good reason to settle for an unstable algorithm.''

\subsubsection{Backward Error of the Basic Arithmetic Operations}
Analyzing a program's forward and backward error requires assumptions about the
accuracy of basic floating-point operations. A common approach models
arithmetic over a finite set $\F \subset \R$ by defining each basic operation
$\op \in \{+,-,\cdot, / \}$ as computing $fl(x ~\op~ y)$, where $x, y \in \F$
and $fl: \R \rightarrow \F$ rounds real numbers to nearby floating-point
values. The \emph{standard rounding error model} \cite[Theorem
2.2]{Higham:2002:Accuracy} asserts that, assuming no underflow or overflow,
there exists a real number $\delta$ such that
\begin{equation}\label{eq:std_mdl}
    fl(x ~\op~ y) = (x ~\op~ y)(1 + \delta), \qquad |\delta| \le u, 
    \qquad  \op \in \{+,-,\cdot, / \}.
\end{equation}
We instead adopt an model due to Olver \cite{Olver:78:error}, which replaces
the linear $1 + \delta$ error term with an exponential one: 
\begin{equation}\label{eq:olver_mdl}
    fl(x ~\op~ y) = (x ~\op~ y)e^\delta, \qquad |\delta| \le u/(1-u), 
    \qquad  \op \in \{+,-,\cdot, / \}.
\end{equation}
This formulation soundly overapproximates the standard model, as shown by the
Taylor series expansion $e^\delta = 1 + \delta + \mathcal{O}(\delta^2)$ around
$\delta =0$. The exponential form preserves the leading-order behavior of the
standard model while offering better compositional properties for analyzing
error propagation through sequences of operations~\cite{NUMFUZZ}. 

The following extended metric $RP : \R \times \R \to \R^{\geq 0}\cup\{\infty\}$,
called the \emph{relative precision metric}, captures a relative notion of error 
induced by Olver's model:
\begin{equation}\label{eq:olver}
 RP(x,y) =
 \begin{cases}
   |\ln(x/y)| & \text{if $\sgn{x} = \sgn{y}$ and $x, y \neq 0$} \\
   0 &  \text{if $x = y = 0$} \\
   \infty & \text{otherwise.}
 \end{cases}
\end{equation}
Using relative precision, the \emph{forward error} of each basic floating-point
operation is bounded by $RP(fl(x ~\op~ y),x ~\op~ y) \le u/(1-u)$. To derive a
corresponding \emph{backward error} bound, we interpret each floating-point
operation as exact on slightly perturbed inputs. Let $a_i, b_i, c_i, d_i$ be
floating-point numbers for $i = 1, 2$. Then:
\begin{center}
\begin{minipage}{0.50\textwidth}
\begin{align}
fl(a_1 + a_2)  &= a_1e^\alpha + a_2e^\alpha = \tilde{a}_1 + \tilde{a}_2
\label{eq:be_add} \\
fl(b_1 - b_2)  &= b_1e^\beta - b_2e^\beta = \tilde{b}_1 - \tilde{b}_2
\label{eq:be_sub}
\end{align}
\end{minipage}%
\begin{minipage}{0.50\textwidth}
\begin{align}
fl(c_1 \cdot c_2)  &= c_1e^{\gamma/2} \cdot c_2e^{\gamma/2}
= \tilde{c}_1 \cdot \tilde{c}_2 
\label{eq:be_mul} \\
fl(d_1/d_2)  &= (d_1e^{\delta/2}) / (d_2e^{-\delta/2})
= \tilde{d}_1 / \tilde{d}_2,
\label{eq:be_div}
\end{align}
\end{minipage}
\end{center}
\vspace{0.5em}
where $|\alpha|,|\beta|,|\gamma|,|\delta| \le \varepsilon = u/(1-u)$. The
intuition behind a perturbed input like $\tilde{a}_1 = a_1e^\alpha$ in
\Cref{eq:be_add} is that $\tilde{a}_1$ is approximately equal to $a_1 (1 +
\alpha)$ when the magnitude of $\alpha$ is small; the backward
error with respect to $a_1$ is then $RP(a_1,\tilde{a}_1) = |\alpha| \le
\varepsilon$. Thus, the result of basic arithmetic operations is exact
for perturbed inputs differing from the true inputs by at most $\varepsilon$ in
relative terms. 

\subsubsection{Backward Error Analysis by Example} 
A classic example illustrating the value of backward error analysis is the dot
product of two vectors. Although it can be computed in a backward stable way,
if the vectors are orthogonal (i.e., the dot product is zero), the relative
forward error can be arbitrarily large. In such cases, a forward error analysis
yields trivial bounds, while backward error analysis provides
meaningful bounds describing the quality of an implementation for all possible
inputs. 

To illustrate, consider vectors $x = (x_0,x_1) \in \R^2$ and $y = (y_0,y_1) \in
\R^2$ with floating-point entries. The exact dot product is $s = x_0 \cdot y_0
+ x_1 \cdot y_1$, and the floating-point version is $\tilde{s} = fl(fl(x_0
\cdot y_0) + fl(x_1 \cdot y_1))$. A backward error bound on $\tilde{s}$ follows
from \Cref{eq:be_add,eq:be_mul}. In particular, we use the bounds
for floating-point addition and multiplication to define the perturbed vectors
$\tilde{x} =(\tilde{x}_0 , \tilde{x}_1)$ and $\tilde{y} =
(\tilde{y}_0,\tilde{y}_1)$ such that their exact dot product equals
$\tilde{s}$: 
\begin{align}
\tilde{s} = fl(fl(x_0 \cdot y_0) + fl(x_1 \cdot y_1))
&= fl(x_0 e^{\delta_0/2} \cdot y_0 e^{\delta_0/2} +
  x_1 e^{\delta_1/2} \cdot y_1 e^{\delta_1/2})
  \nonumber \\
&= (x_0 e^{\delta_0/2} \cdot y_0 e^{\delta_0/2})e^{\delta_2} +
    (x_1 e^{\delta_1/2} \cdot y_1 e^{\delta_1/2}) e^{\delta_2}
  \nonumber \\
&= (x_0 e^{\delta_0/2+\delta_2} \cdot y_0 e^{\delta_0/2+\delta_2}) +
  (x_1 e^{\delta_1/2+\delta_2} \cdot y_1 e^{\delta_1/2+\delta_2}) 
  \nonumber \\
&= (x_0 e^{\delta} \cdot y_0 e^{\delta}) +
    (x_1 e^{\delta'} \cdot y_1 e^{\delta'}) 
= (\tilde{x}_0  \cdot \tilde{y}_0) + (\tilde{x}_1 \cdot \tilde{y}_1)
  \label{eq:dpbe}
\end{align}
where we set $\delta=\delta_0/2+\delta_2$ and $\delta'=\delta_1/2+\delta_2$. By
\Cref{eq:be_add,eq:be_mul}, we have
$|\delta_0|,|\delta_1|,|\delta_2|\leq \varepsilon=u/(1-u)$. Since $\delta$ is
the sum of two terms, one scaled by $1/2$, we obtain $|\delta| \le
|\delta_0/2| + |\delta_2| \le \varepsilon/2 + \varepsilon = 3\varepsilon/2$,
and similarly for $|\delta'|$. Thus, the above analysis shows that
the floating-point dot product of the vectors $x$ and $y$ is equal to an exact
dot product of the slightly perturbed inputs $\tilde{x}$ and $\tilde{y}$, where
each component satisfies $RP(x_i,\tilde{x}_i) \le 3\varepsilon/2$ and
$RP(y_i,\tilde{y}_i) \le 3\varepsilon/2$. We say that the dot product operation
is backward stable with componentwise backward error bounded by $3\varepsilon/2$
in relative precision, as formalized by \Cref{def:bstab}.

A subtle point is that the backward error for multiplication and division can
be described in a slightly different way, while still maintaining the same
backward error bound given in \Cref{eq:be_mul}. In particular, floating-point
multiplication behaves like multiplication in exact arithmetic with a
\emph{single} input subject to small perturbations: let $c_i, d_i$ be
floating-point numbers for $i = 1, 2$, then
\begin{align}
fl(c_1 \cdot c_2) = c_1 \cdot c_2e^\gamma = c_1 \cdot \tilde{c}_2 \label{eq:be_mul2}\\
fl(d_1 / d_2) = d_1e^\delta / d_2 = \tilde{d}_1 / d_2 \label{eq:be_div2}
\end{align} 
with $|\gamma|, |\delta| \le \varepsilon$. There are many other ways to assign
backward error to multiplication and division, provided the exponents sum to
$\gamma$ and $\delta$, respectively; in general, a given program may satisfy a
variety of different backward error bounds depending on how the backward error
is allocated between the program inputs. However, assigning different
per-variable backward error bounds for individual operations does not
fundamentally alter the stability analysis, as total error across all variables 
remains within the same bound. We will see in \Cref{sec:linearity} that, in some cases, 
using \Cref{eq:be_mul2} instead of \Cref{eq:be_mul} enables a backward
error analysis of computations that share variables across subexpressions.

\subsection{Backward Error Analysis in \bea{}: Motivating Examples}
In order to reason about backward error as it has been described so far, the
type system of \bea{} combines three ingredients: coeffects, distances, and
linearity. To get a sense of the critical role each of these components plays
in the type system, we first consider the following \bea{} program for
computing the dot product of 2D-vectors \stexttt{x} and \stexttt{y}:

\vspace{-1em}
\begin{center}
\begin{minipage}[t]{.5\textwidth}
\begin{lstlisting}
DotProd2 x y :=
let (x0, x1) = x in
let (y0, y1) = y in
let v = mul x0 y0 in
let w = mul x1 y1 in
add v w
\end{lstlisting}
\end{minipage}
\end{center}

\subsubsection{Coeffects}
The type system of \bea{} allows us to prove the following typing judgment: 
\begin{equation}\label{eq:dpbound} 
\emptyset \mid \stexttt{x} :_{3\varepsilon/2} \R^2, \stexttt{y}:_{3\varepsilon/2}
\R^2 \vdash \stexttt{DotProd2} : \R
\end{equation}
The $\emptyset$ represents the set of input variables which may not be perturbed; 
in this program, there are none. 
To the right of the vertical bar, \stexttt{x} and \stexttt{y}
are the inputs which may be perturbed, i.e., they may have nonzero backward error. 
The \emph{coeffect} annotations ${3\varepsilon/2}$ in the context bindings
$\stexttt{x}:_{3\varepsilon/2} \R^2$ and $\stexttt{y} :_{3\varepsilon/2} \R^2$ express
per-variable relative backward error bounds for \stexttt{DotProd2}.
Thus, the typing judgment for \stexttt{DotProd2} captures the desired
backward error bound in \Cref{eq:dpbe}.

{Coeffect} systems \cite{Ghica:2014:coeffects, Petricek:2014:coeffects,
Brunel:2014:coeffects, tate:2013:effects} have traditionally been used in the
design of programming languages that perform resource management, and provide a
formalism for precisely tracking the usage of variables in programs. In
\emph{graded coeffect systems} \cite{Gaboardi:2016:combining}, bindings in a
typing context $\Gamma$ are of the form $x :_r \sigma$, where the annotation
$r$ is some quantity controlling how $x$ can be used by the program. In
\bea{}, these annotations describe the amount of backward error that can be
assigned to the variable: a typing judgment $\emptyset \mid x
:_r \sigma \vdash e : \tau$ ensures that the term $e$ has at most $r$ backward
error with respect to the variable $x$. 

In \bea{}, the coeffect system allows us to derive backward error bounds for
larger programs from the known language primitives; the typing rules are used
to track the backward error of increasingly large programs in a compositional
way. For instance, the typing judgment given in \Cref{eq:dpbound} for the
program \stexttt{DotProd2} is derived using primitive typing rules for
addition and subtraction. These rules capture the backward error bounds
described in \Cref{eq:be_add} and \Cref{eq:be_mul}: 
\vspace{0.5em}
\begin{center}
\AXC{}
\RightLabel{(Add)}
\UIC{$\emptyset \mid  \Gamma, x:_\varepsilon \R, 
            y:_\varepsilon\R \vdash
   \add x y : \R$}
\bottomAlignProof
\DisplayProof
\hspace{2em}
\AXC{}
\RightLabel{({Mul})}
\UIC{$\emptyset \mid  \Gamma, x:_{\varepsilon/2} \R, 
            y:_{\varepsilon/2}\R \vdash \mul  x  y : \R$}
\bottomAlignProof
\DisplayProof
\end{center}
We can also derive similar rules for subtraction and division, where division
may return a $\unit$ type to represent a division-by-zero error:
\vspace{0.5em}
\begin{center}
\AXC{}
\RightLabel{(Sub)}
\UIC{$\emptyset \mid  \Gamma, x:_\varepsilon \R, 
            y:_\varepsilon\R \vdash
   \sub x y : \R$}
\bottomAlignProof
\DisplayProof
\hspace{2em}
\AXC{}
\RightLabel{({Div})}
\UIC{$\emptyset \mid  \Gamma, x:_{\varepsilon/2} \R, 
            y:_{\varepsilon/2}\R \vdash \fdiv  x  y : \R+\unit$}
\bottomAlignProof
\DisplayProof
\end{center}

The following rule captures the backward error bound described in 
\Cref{eq:be_mul2}:  

\vspace{0.5em}
\begin{center}
\AXC{}
\RightLabel{({DMul})}
\UIC{$ x: \R \mid 
    \Gamma, y:_{\varepsilon}\R \vdash \dmul  x  y : \R$}
\bottomAlignProof
\DisplayProof
\end{center}

In \bea{}, backward error originates from the arithmetic operations captured 
by these rules, each reflecting standard floating-point analyses. To make rounding 
behavior explicit, \bea{} could also be extended with a rule modeling a unary
rounding operation. 

\subsubsection{Distances}
In order to derive concrete backward error bounds, we require a notion of
\emph{distance} between points in an input space. To this end, each type
$\sigma$ in \bea{} is equipped with a distance function $d_\sigma :
\sigma \times \sigma \to \R_{\ge 0}\cup\{\infty\}$ describing how close pairs of
values of type $\sigma$ are to one another. For our numeric type
$\R$, choosing the relative precision metric given in \Cref{eq:olver} allows
us to prove backward error bounds for a relative notion of error. This idea is
reminiscent of type systems capturing function sensitivity
\citep{FUZZ,DFUZZ,NUMFUZZ}; however, the \bea{} type system does not capture
function sensitivity since this concept does not play a central role in
backward error analysis.

\subsubsection{Linearity}\label{sec:linearity}
The conditions under which composing backward stable programs yields another
backward stable program are poorly understood. Our development of a static
analysis framework for backward error analysis led to the following insight:
the composition of two backward stable programs remains backward stable
\emph{as long as they do not assign backward error to shared variables}. To
enforce this, \bea{} uses a \emph{linear} typing discipline to control variable 
duplication. While most coeffect type systems allow using a variable $x$ in two
subexpressions as long as the grades $x :_r \sigma$ and $x :_s \sigma$ are
combined in the overall program, \bea{} requires a stricter condition: linear
variables cannot be duplicated at all.

To understand why a type system for backward error analysis must control
duplication, consider the floating-point computation $h(x, a, b) =
ax^2 + bx$. The variable $x$ is used in each of the subexpressions $f(x, a) =
ax^2$ and $g(x, b) = bx$. Using the backward error bound given in
\Cref{eq:be_mul} for multiplication, the backward stability of $f$ is
guaranteed by the existence of the perturbed coefficient $\tilde{a} =
ae^{\delta_1}$ and the perturbed variable $\tilde{x}_f = xe^{\delta_2}$:
\begin{align}
  \tilde{f}(x, a) = ae^{\delta_1} \cdot (xe^{\delta_2})^2  = \tilde{a} \cdot 
  \tilde{x}_f^2=f(\tilde{x_f},\tilde{a}).
\end{align}
Similarly, the backward stability of $g$ is guaranteed by the existence of the
coefficient $\tilde{b} = be^{\delta_3/2}$ and the variable $\tilde{x}_g =
xe^{\delta_3/2}$. However, since $\delta_2$ might not be equal to $\delta_3$,
there might not be a common input $\tilde{x}$ such that
${f}(\tilde{x}, \tilde{a}) + g(\tilde{x},\tilde{b}) = \tilde{f}(x, a) +
\tilde{g}(x, b)$. 

Linearity in \bea{} ensures that we never need to reconcile multiple, possibly
incompatible backward error requirements for the same variable. However, it
excludes some programs that are backward stable---for instance,
$h(x, a, b) = ax^2 + bx$ actually \emph{is} backward stable! To regain
flexibility, we note that a variable \emph{can} be duplicated safely when it
does not need to be perturbed to provide a backward error guarantee.
For $h(x, a, b)$, using \Cref{eq:be_mul2} allows assigning zero backward error
to $x$ and non-zero backward error to $a$ and $b$. Since $x$ need not be
perturbed to provide a backward error guarantee, it can be duplicated without
violating backward stability. 

To realize this idea in \bea, the type system distinguishes between linear
(restricted-use) and non-linear (reusable) data; backward error is only
assigned to linear variables, never to non-linear ones. Technically, \bea{}
uses a dual context judgment, reminiscent of work on linear/non-linear logic
\cite{DBLP:conf/csl/Benton94}, to track both kinds of variables. In more
detail, a judgment of the form $y: \alpha \mid x :_r \sigma \vdash e: \tau$,
ensures that $e$ has at most $r$ backward error with respect to the
\emph{linear} variable $x$, and has \emph{no backward error} with respect to
the \emph{non-linear} variable $y$. (Note that the bindings in the nonlinear
context do not carry an index, because no amount of backward error can be
assigned to these variables.) The soundness theorem for \bea{} introduced in
the next section formalizes this guarantee.

\begin{remark}[Linearity and Completeness] 
  Like most type systems, \bea{} is \emph{sound but not complete}: it may
  reject backward stable programs. For instance, $f(x, y) = x \cdot y + y$ is
  backward stable but rejected by \bea{} because $y$ is reused, and an overall
  backward error guarantee cannot be derived by assigning error to $x$ alone.
  While the conditions under which forward error bounds compose has been
  studied (e.g., \citep{Bornemann:2007:backwardcompose}), less is known about
  composing backward error bounds. \bea{} enforces a sufficient but not
  necessary condition: shared variables must not accumulate incompatible
  perturbations. Whether a more permissive system can remain sound is an open
  question.
\end{remark}

\section{\bea{}: A Language for Backward Error Analysis} \label{sec:language}

\begin{figure}[tbp]
  \begin{alignat*}{3}
         &\sigma,\tau &::=~ &\unit
         \mid \num
         \mid \sigma \otimes  \sigma
         \mid \sigma + \sigma \mid \alpha
         \tag*{(types)} \\
         &\alpha &::=~ &m(\sigma) \tag*{(discrete types)} \\
	 &\Gamma &::=~ &\emptyset
	 \mid \Gamma, x:_r \sigma
	 \tag*{(linear typing contexts)} \\
	 &\Phi &::=~ &\emptyset 
	 \mid \Phi, z:\alpha
         \tag*{(discrete typing contexts)} \\
         &e,f \ &::=~ &x \mid z
         \mid ()
         \mid ~!e
         \mid  (e, f) 
         \mid \inl e
         \mid \inr e \\
         & & & \hspace{-0.25em} 
         \mid \slet x e f 
         \mid \slet {(x,y)} e f 
         \mid \dlet z e f 
         \mid \dlet {(x,y)} e f \\
         & & & \hspace{-0.25em} 
         \mid \case {e'} {\inl x.e} {\inr y.f} 
         \mid \add e f
         \mid \sub e f 
         \mid \mul e f 
         \mid \dmul e f 
         \mid \fdiv e f
      %   \mid \fle e f 
        \tag*{(expressions)}
  \end{alignat*}
  \caption{Grammar for $\bea{}$ types and terms.}
  \label{fig:grammar}
\end{figure}

\bea{} is a simple first-order programming language, extended with a few
constructs that are unique to a language for backward error analysis. The
grammar of the language is presented in \Cref{fig:grammar}, and the typing
relation is presented in \Cref{fig:typing_rules}.

\subsection{{Types}} 
The base types in \bea{} are the numeric type $\num$ and the unit type $\unit$,
with a single inhabitant. Types combine using tensor product $\otimes$
and sum $ + $. The type constructor $m$ forms \emph{discrete}
versions of types; types not wrapped in $m$ are called \emph{linear}.
This distinguishes unrestricted-use (discrete) data, which cannot have 
backward error, from restricted-use (linear) data, which can.

\subsection{Typing Judgments}
Terms are typed with judgments of the form ${\Phi, z: \alpha \mid \Gamma, x:_r
\sigma \vdash e : \tau}$ where $\Phi$ is a discrete typing context and $\alpha$
is a discrete type, $\Gamma$ is a linear typing context and $\sigma$ is a
linear type, and $e$ is an expression. For linear typing contexts, variable
assignments have the form $x :_r \sigma$, where the grade $r$ is a member of a
preordered monoid $\mathcal{M} = (\mathbb{R}^{\ge0},+,0)$. Typing contexts,
both linear and discrete, are defined inductively as shown in
\Cref{fig:grammar}.

Although linear typing contexts cannot be joined together with discrete typing
contexts, linear typing contexts can be joined with other linear typing
contexts as long as their domains are disjoint. We write $\Gamma, \Delta$ to
denote the disjoint union of the linear contexts $\Gamma$ and $\Delta$.

While most graded coeffect systems support the composition of linear typing
contexts $\Gamma$ and $\Delta$ via a \emph{sum}, or \emph{contraction} operation
$\Gamma + \Delta$ ~\cite{DalLago:2022:relational, Gaboardi:2016:combining},
where the grades of shared variables in the contexts are combined together, this
operation is not supported in \bea{} because the sum operation allows
duplication of variables, but \bea{}'s strict linearity requirement does not
allow variables to be duplicated in general. However, \bea{}'s type system does
support a sum operation that adds a given grade $q \in \mathcal{M}$ to the
grades in a linear typing context: 
\begin{gather*} 
  q + \emptyset \triangleq \emptyset \qquad
  q + (\Gamma, x:_r \sigma) \triangleq q + \Gamma, x:_{q+r} \sigma.
\end{gather*}
Intuitively, this operation is used to push $q$ backward error through a typing
judgment.

In \bea{}, a well-typed expression ${ \Phi \mid x:_r \sigma \vdash e : \tau}$
is a program that has at most $r$ backward error with respect to the
\emph{linear} variable $x$, and has \emph{no backward error} with respect to
the \emph{discrete} variables in the context $\Phi$. For more general programs
of the form 
\[
	\Phi \mid x_1:_{r_1} \sigma_1, \dots, x_i:_{r_i} \sigma_i \vdash e : \tau,
\] 
\bea{} guarantees that the program has at most $r_i$ backward error with
respect to each variable $x_i$, and has \emph{no backward error} with respect
to the {discrete} variables in the context $\Phi$. This idea is formally
expressed in our soundness theorem (\Cref{thm:main}).

\subsection{{Expressions}} 
\bea{} expressions include linear variables $x$ and discrete variables $z$, as
well as a unit () value. Linear variables are bound in let-bindings of the form
{$\slet x {e} {f}$}, while discrete variables are bound in let-bindings of the
form {$\dlet z {e} {f}$}. The !-constructor
declares that an expression can be duplicated, preventing backward error
from being pushed onto the expression. The pair constructor $(e,f)$
corresponds to a tensor product, and can be composed of expressions of both
linear and discrete type. Discrete pairs are eliminated by pattern matching
using the construct {$\dlet {(x,y)} {e} {f}$}, whereas linear pairs are
eliminated using the construct $\slet {(x,y)} {e} {f}$.
The injections $\inl e$ and $\inr e$ construct the coproduct, and are
eliminated by case analysis using the construct $\case {e} {\inl x.f} {\inr y.f}$. 
The primitive arithmetic operations of the language ($\mathbf{add}$,
$\mathbf{mul}$, $\mathbf{sub}$, $\mathbf{div}$, $\mathbf{dmul}$) were introduced in
\Cref{sec:overview}. 

\subsection{{Typing Relation}} 
\begin{figure}
%% ROW1
\begin{center}
%% var1
\AXC{ }
\RightLabel{(Var)}
\UIC{$\Phi\mid \Gamma, x:_r \sigma \vdash x : \sigma$}
\bottomAlignProof
\DisplayProof
\hskip 0.5em
%% var2
\AXC{}
\RightLabel{(DVar)}
\UIC{$\Phi, z:\alpha\mid \Gamma \vdash z : \alpha$}
\bottomAlignProof
\DisplayProof
\vskip 1em
%%
%% ROW2
%% cprod intro
\AXC{$\Phi\mid\Gamma \vdash e : \sigma$}
\AXC{$\Phi\mid\Delta \vdash f : \tau$}
\RightLabel{($\otimes $ I)}
\BinaryInfC{$\Phi\mid\Gamma, \Delta \vdash ( e, f ): \sigma \otimes   \tau $}
\bottomAlignProof
\DisplayProof
\hskip 0.5em
%% unit
\AXC{}
\RightLabel{(Unit)}
\UIC{$\Phi\mid \Gamma \vdash () : \mathbf{unit}$}
\bottomAlignProof
\DisplayProof
\vskip 1em
%%
%% ROW3
%%
%% prod elim
\AXC{$\Phi\mid\Gamma \vdash e : \tau_1 \otimes  \tau_2$}
\AXC{$\Phi\mid\Delta, x:_r \tau_1, 
            y:_r \tau_2 \vdash f : \sigma$}
\RightLabel{($\otimes $ E$_\sigma$)}
\BIC{$\Phi\mid r + \Gamma, \Delta \vdash \slet {(x,y)} e f : \sigma$}
\bottomAlignProof
\DisplayProof
\vskip 1em

%%
%% NEW ROW
%%
%% prod elim discrete
\AXC{$\Phi\mid\Gamma \vdash e : \alpha_1 \otimes  \alpha_2$}
\AXC{$\Phi, z_1: \alpha_1, 
            z_2: \alpha_2\mid\Delta \vdash f : \sigma$}
\RightLabel{($\otimes $ E$_\alpha$)}
\BIC{$\Phi\mid \Gamma, \Delta \vdash 
        \dlet {(z_1,z_2)} e f : \sigma$}
\bottomAlignProof
\DisplayProof
\vskip 1em

%%
%% ROW4
%%
% sum elim
\AXC{$\Phi\mid\Gamma \vdash e' : \sigma+\tau$}
\AXC{$\Phi\mid\Delta, x:_q \sigma \vdash e : \rho$}
\AXC{$\Phi\mid\Delta, y:_q \tau \vdash f: \rho$}
\RightLabel{($+$ E)}
\TIC{$\Phi\mid q+\Gamma,\Delta \vdash \case{e'}{\inl x.e}{\inr y.f} : \rho$}
\bottomAlignProof
\DisplayProof
\vskip 1em
%%
%% ROW 5
%% ind sum intro
\AXC{$\Phi\mid\Gamma \vdash e : \sigma$ }
\RightLabel{($+$ $\text{I}_L$)}
\UIC{$\Phi\mid\Gamma \vdash \inl \ e : \sigma + \tau$}
\bottomAlignProof
\DisplayProof
\hskip 0.5em
%% ind sum intro
\AXC{$\Phi\mid\Gamma \vdash e : \tau$ }
\RightLabel{($+$ $\text{I}_R$)}
\UIC{$\Phi\mid\Gamma \vdash \inr \ e : \sigma + \tau$}
\bottomAlignProof
\DisplayProof
\vskip 1em
%%
%%% ROW 6
% let 
\AXC{$\Phi\mid\Gamma \vdash e :  \tau$}
\AXC{$\Phi\mid\Delta, x :_r \tau \vdash f : \sigma$}
\RightLabel{(Let)}
\BIC{$\Phi\mid r + \Gamma,\Delta \vdash \slet x e f : \sigma$}
\bottomAlignProof
\DisplayProof
\vskip 1.3em
%%
%%% ROW 7
% disc
\AXC{$\Phi\mid\Gamma \vdash e :  \sigma$}
\RightLabel{(Disc)}
\UIC{$\Phi \mid \Gamma  \vdash ~ !e : m(\sigma)$}
\bottomAlignProof
\DisplayProof
\hskip 0.5em
% dlet 
\AXC{$\Phi\mid\Gamma \vdash e :  \alpha$}
\AXC{$\Phi, z: \alpha \mid\Delta \vdash f : \sigma$}
\RightLabel{(DLet)}
\BIC{$\Phi \mid \Gamma,\Delta \vdash \dlet z e f : \sigma$}
\bottomAlignProof
\DisplayProof
\vskip 1.3em
%%
%%% ROW 8
% add, sub
\AXC{}
\RightLabel{(Add, Sub)}
\UIC{$\Phi\mid \Gamma, x:_{\varepsilon + r_1} \num, 
            y:_{\varepsilon + r_2}\num \vdash
   \{\mathbf{add}, 
    \mathbf{sub}\} \ x \ y : \num$}
\bottomAlignProof
\DisplayProof
\vskip 1.3em
%%
%%% ROW 9
% mul 
\AXC{ }
\RightLabel{(Mul)}
\UIC{$\Phi\mid \Gamma, 
    x:_{\varepsilon/2 + r_1} \num, 
    y:_{\varepsilon/2 + r_2}\num \vdash \mul  x  y : \num$}
\bottomAlignProof
\DisplayProof
\vskip 1.3em
%%% ROW 10
%  div 
\AXC{ }
\RightLabel{(Div)}
\UIC{$\Phi\mid
    \Gamma, x:_{\varepsilon/2 + r_1} \num, 
    y:_{\varepsilon/2 + r_2}\num
     \vdash \fdiv x  y : \num + \unit$}
\bottomAlignProof
\DisplayProof
\vskip 1.3em
%%% ROW 11
% mul 
\AXC{ }
\RightLabel{(DMul)}
\UIC{$\Phi,z: m(\num)\mid 
        \Gamma, 
    x:_{\varepsilon + r}\num \vdash \dmul  z  x : \num$}
\bottomAlignProof
\DisplayProof
\vskip 1em
\end{center}
    \caption{Typing rules  with 
    $q,r,r_1,r_2 \in \mathbb{R}^{\ge0}$ and fixed parameter $\varepsilon \in \mathbb{R}^{>0}$. }
    \label{fig:typing_rules}
\end{figure}

The full type system for \bea{} is given in \Cref{fig:typing_rules}. It is
parametric with respect to the constant $\varepsilon = u/(1-u)$, where $u$
represents the unit roundoff. 

Let us now describe the rules in \Cref{fig:typing_rules}, starting with those
that employ the sum operation between grades and linear typing contexts: the
linear let-binding rule (Let) and the elimination rules for sums ($+$ E) and
linear pairs ($\otimes$ E$_\sigma$). Using the intuition that a grade describes
the backward error bound of a variable with respect to an expression, we see
that whenever we have an expression $e$ that is well-typed in a context
$\Gamma$ and we want to use $e$ in place of a variable that has a backward
error bound of $r$ with respect to another expression, then we must assign $r$
backward error onto the variables in $\Gamma$ using the sum operation $r +
\Gamma$. That is, if an expression $e$ has a backward error bound of
$q$ with respect to a variable $x$ and the expression $f$ has backward error
bound of $r$ with respect to a variable $y$, then $f[e/y]$ will have backward
error bound of $r + q$ with respect to the variable $x$.

The action of the !-constructor is illustrated in the Disc rule, which promotes
an expression of linear type to discrete type. The
!-constructor allows an expression to be freely duplicated, while preventing
backward error to be pushed onto the expression.

The only remaining rules in
\Cref{fig:typing_rules} that are not mostly standard are the rules for
primitive arithmetic operations: addition (Add), subtraction (Sub),
multiplication between two linear variables (Mul), division (Div), and
multiplication between a discrete and non-linear variable (DMul). While these
rules are designed to mimic the relative backward error bounds for
floating-point operations following analyses described in the numerical
analysis literature
\cite{Olver:78:error,Higham:2002:Accuracy,Corless:2013:Analysis} and as briefly
introduced in \Cref{sec:overview}, they also allow weakening, or relaxing, the
backward error guarantee. Intuitively, if the backward error of an expression
with respect to a variable is \emph{bounded} by $\varepsilon$, then it is also
\emph{bounded} by $\varepsilon + r$ for some grade $r \in \mathcal{M}$. We also
note that division is a partial operation, where coproduct result is used to
model a possible division by zero.

\subsection{Backward Error Soundness}\label{ssec:soundness}
With the basic syntax of the language in place, we can state the following 
backward error soundness theorem for the type system, which is the central 
guarantee for \bea{}. 

\begin{theorem}[Backward Error Soundness]\label{thm:main}
Let $e$ be a well-typed \bea{} program
\[ 
z_1:\alpha_1,\dots,z_n:\alpha_n \mid  x_1:_{r_1}\sigma_1,\dots,x_m:_{r_m}\sigma_m
\vdash e : \tau ,
\] 
and let $(p_i)_{1 \le i \le n}$ and $(k_j)_{1
\le j \le m}$ be sequences of values such that $ \vdash p_i : \alpha_i$ and $ \vdash k_j
:  \sigma_j$ for all $i \in 1, \dots, n$ and $j \in 1, \dots, m$. If the
program $e[p_1/z_1]\cdots[{k_m}/x_m]$ evaluates to a value $v$ under an
approximate floating-point semantics, then (1) there exist well-typed values
$(\tilde{k}_j)_{1 \le i \le m}$ such that the program
$e[p_1/z_1]\cdots[{\tilde{k}_m}/x_m]$ also evaluates to $v$ under an ideal,
infinite-precision semantics, and (2) the distance between the values ${k}_i$
and $\tilde{k}_i$ is at most $r_i$ for every $j \in 1, \dots, m$. 
\end{theorem}

We will return to the precise statement and proof of this theorem in
\Cref{sec:metatheory}, but we first walk through some examples of \bea{}
programs.

\section{Example \bea{} Programs}\label{sec:examples}
We present case studies showing how algorithms with known backward error bounds from
the literature can be implemented in \bea{}. We begin with matrix-vector
multiplication, starting with the simplest case and then composing programs to
implement a generalized version. Next, we explore polynomial evaluation,
comparing a naive strategy with Horner's scheme. Finally, we write a triangular
linear solver to illustrate some of \bea{}'s more complex language features.

To improve the readability of our examples, we adopt several conventions.
First, matrices are assumed to be stored in row-major order. Second, following
the convention used in the grammar for \bea{} in \Cref{sec:language}, we use
\stexttt{x} and \stexttt{y} for linear variables and \stexttt{z} for discrete
variables. Finally, for types, we denote both discrete and linear numeric
types by $\R$, and use a shorthand for type assignments of vectors and
matrices. For instance: $\R^2 \equiv (\R \otimes \R)$ and $\R^{3 \times 2}
\equiv (\R \otimes \R) \otimes (\R \otimes \R) \otimes (\R \otimes \R)$.

Since \bea{} is a simple first-order language and currently does not support
higher-order functions or variable-length tuples, programs can become verbose.
To reduce code repetition, we use basic user-defined abbreviations in our
examples. 

\subsection{Matrix-Vector Product}
\subsubsection{A Simple Matrix-Vector Product}
To illustrate how \bea{} performs backward error analysis, consider the
matrix-vector product of a matrix $A \in \R^{2\times 2}$ and a vector $z \in
\R^2$. We present the computation in both standard mathematical notation and
\bea{} syntax, and explain how the type system verifies a backward error bound.
Rather than showing the full typing derivation, we highlight key typing rules
and their connection to standard floating-point error analysis. To begin, let
\begin{align*}
A = 
\begin{pmatrix}
a_{00} & a_{01}\\
a_{10} & a_{11} 
\end{pmatrix} \in \R^{2\times 2}
\quad
\text{and}
\quad
z = 
\begin{pmatrix}
z_0\\
z_1
\end{pmatrix} \in \R^2
\end{align*}
with exact product
\begin{align*}
A z = 
\begin{pmatrix}
a_{00}z_0 +  a_{01}z_1\\
a_{10}z_0 +  a_{11}z_1
\end{pmatrix}. 
\end{align*}
The corresponding floating-point matrix-vector product $fl(Az)\in \R^2$ is
\begin{align*}
fl(A z) = 
\begin{pmatrix}
fl(fl(a_{00} \cdot z_0) +  fl(a_{01} \cdot z_1))\\
fl(fl(a_{10}\cdot z_0) +  fl(a_{11} \cdot z_1))
\end{pmatrix}.
\end{align*}
Using \Cref{eq:be_add} and \Cref{eq:be_mul2}, we have $fl(Az) = \tilde{A}z$ where  
\begin{align*}
\tilde{A} = 
\begin{pmatrix}
\tilde{a}_{00}  &  \tilde{a}_{01} \\
\tilde{a}_{10} &  \tilde{a}_{11}
\end{pmatrix},
\end{align*}
and each entry $\tilde{a}_{ij}$ of $\tilde{A}$ satisfies $\tilde{a}_{ij} =
a_{ij}e^\delta$ with $|\delta| \le 2 \varepsilon$. We encode this analysis of
the matrix-vector product in \bea{} as follows:
\vspace{-1em}
\begin{center}
\begin{minipage}[t]{.6\textwidth}
\begin{lstlisting} 
MatVecEx ((a00, a01), (a10, a11)) (z0, z1) :=
let s0 = dmul z0 a00 in
let s1 = dmul z1 a01 in 
let s2 = dmul z0 a10 in 
let s3 = dmul z1 a11 in 
let u0 = add s0 s1 in 
let u1 = add s2 s3 in 
(u0, u1)
\end{lstlisting}
\end{minipage}
\end{center}
Each \stexttt{dmul} operation represents a floating-point multiplication of a
matrix entry and a vector component. According to the DMul rule, the backward
error is attributed to the second operand: \stexttt{dmul z0 a00} is interpreted
as being precisely equal to the exact result of multiplying the exact value of
\stexttt{z0} with a perturbed version of \stexttt{a00}. Using the standard
backward error notation (\Cref{eq:be_mul2}), the assignments in
\stexttt{MatVecEx} correspond to
\begin{equation*}
\begin{aligned}
s_0 &= fl(z_0 \cdot a_{00}) = z_0 (a_{00} e^{\delta_0}) \quad &   
s_1 &= fl(z_1 \cdot a_{01}) = z_1 (a_{01} e^{\delta_1})  \\
s_2 &= fl(z_0 \cdot a_{10}) = z_0 (a_{10} e^{\delta_2}) \quad &    
s_3 &= fl(z_1 \cdot a_{11}) = z_1 (a_{11} e^{\delta_3}),
\end{aligned}
\end{equation*}
where each $e^{\delta_i}$ satisfies $|\delta_i| \le \varepsilon=u/(1-u)$. 

The backward error from a single \stexttt{dmul} is composed with errors from
later operations via the let-binding rule (Let). To illustrate, consider the
error attributed to \stexttt{a00}. Let $e$ denote the remainder of
\stexttt{MatVecEx} after the assignment of \stexttt{s0}, so that $\stexttt{let
s0 = dmul z0 a00 in }e \triangleq \stexttt{MatVecEx}$.
By the DMul rule, the expression \stexttt{dmul z0 a00} is well-typed in the
context $\Phi \mid \stexttt{a00} :_{\varepsilon} \R$, where $\Phi \triangleq
\stexttt{z0}:\R,\stexttt{z1}:\R$, indicating a backward error of at most
$\varepsilon$ attributed to \stexttt{a00}. When the result is bound to
\stexttt{s0} and used in $e$, the Let rule propagates any error assigned to
\stexttt{s0} in $e$ to \stexttt{a00}. Since \stexttt{s0} is later used in a
floating-point addition, \stexttt{a00} accumulates an additional $\varepsilon$
of error. This composition of errors is captured in the following valid typing
derivation:
\[ 
  \AXC{}\RightLabel{(DMul)}
  \UIC{$\Phi\mid \stexttt{a00} :_{\varepsilon} \R \vdash \stexttt{dmul z0 a00} : \R$} 
  \AXC{\strut \vdots}\RightLabel{(Let)}
  \UIC{$\Phi\mid \stexttt{s0}:_\varepsilon \R, \Gamma \vdash e : \R^2$} \RightLabel{(Let)}
  \BinaryInfC{$\Phi \mid (\varepsilon + \stexttt{a00} :_{\varepsilon} \R), \Gamma
  \vdash \stexttt{let s0 = dmul z0 a00 in }e : \R^2 $}
  \bottomAlignProof \DisplayProof 
\]
where $\Gamma \triangleq  \stexttt{a01}: _{2\varepsilon} \R,
\stexttt{a10}:_{2\varepsilon} \R, \stexttt{a11}: _{2\varepsilon} \R.$ The sum
$\varepsilon + \stexttt{a00} :_{\varepsilon} \R$ reduces to $\stexttt{a00}
:_{2\varepsilon} \R$, meaning that \stexttt{a00} accumulates a total backward
error of $2\varepsilon$ in \stexttt{MatVecEx}.

We now generalize from \stexttt{a00} to the full set of inputs. Each
multiplication \stexttt{dmul} introduces a backward error of $\varepsilon$ to
the corresponding matrix input (\stexttt{a00}, \stexttt{a01}, \stexttt{a10},
\stexttt{a11}), and each resulting product (\stexttt{s0}, \stexttt{s1},
\stexttt{s2}, \stexttt{s3}) is later used in a floating-point addition.
According to the rounding error analysis for addition (\Cref{eq:be_add}), the
assignments to \stexttt{u0} and \stexttt{u1} in \stexttt{MatVecEx} are
interpreted as 
\begin{align*} 
u_0 = fl(s_0 + s_1) = (s_0 + s_1) e^{\delta_4} \quad \text{ and } \quad
u_1 = fl(s_2 + s_3) = (s_2 + s_3) e^{\delta_5},
\end{align*}
where the new error terms $e^{\delta_4}$ and $e^{\delta_5}$ account for
rounding errors introduced by the floating-point additions and satisfy
$|\delta_4|,|\delta_5| \le \varepsilon$. The Add typing rule reflects this
analysis, assigning a backward error of $\varepsilon$ to each operand of the
addition. 

As a result, the local variables \stexttt{s0}, \stexttt{s1}, \stexttt{s2} and
\stexttt{s3} contribute $\varepsilon$ of backward error to the (linear) 
variables in the expressions they are bound to. This error is added to the 
$\varepsilon$ error already attributed to the \stexttt{dmul} operations. Thus, 
each matrix element \stexttt{a00}, \stexttt{a01}, \stexttt{a10}, and \stexttt{a11} 
accumulates a total backward error of $2 \varepsilon$, ensuring that the 
computed matrix-vector product corresponds to the exact result of a perturbed 
input within the expected bound.

\subsubsection{Generalized Matrix-Vector Product}
A key feature of \bea{}'s type system is its ability to compose backward error
guarantees for smaller programs to derive guarantees for larger programs.
We illustrate this feature with a
series of examples that build up to a scaled matrix-vector multiplication of
the form $a \cdot (M \cdot v) + b \cdot u$, where $M \in \R^{m \times n}$, $v
\in \R^{n}$, $u \in \R^m$, and $a,b \in \R$. Since \bea{} does not currently
support variable-length tuples, we focus on the case where $M \in \R^{2 \times
2}$.

We begin with \stexttt{SVecAdd}, which scales a vector using \stexttt{ScaleVec}
and adds the result to another vector. Given a discrete variable $\stexttt{a}
: \R$ and linear variables $\stexttt{x}: \R^2$ and $\stexttt{y}: \R^2$, the
programs are implemented as follows:
\vspace{-1em}
\begin{center}
\begin{minipage}[t]{.45\textwidth}
\begin{lstlisting}
ScaleVec a x :=
let (x0, x1) = x  in 
let u = dmul a x0 in 
let v = dmul a x1 in 
(u, v)
\end{lstlisting}
\end{minipage}
\begin{minipage}[t]{.45\textwidth}
\begin{lstlisting}
SVecAdd a x y :=
let (x0, x1) = ScaleVec a x in 
let (y0, y1) = y  in 
let u = add x0 y0 in
let v = add x1 y1 in
(u, v)
\end{lstlisting}
\end{minipage}
\end{center}
\noindent
These programs have the following valid typing judgments:
\[ 
  \stexttt{a}: \R \mid \stexttt{x} :_{\varepsilon}
  \R^2 \vdash \stexttt{ScaleVec a x} : \R
  \quad \text{and} \quad
  \stexttt{a}:  \R \mid \stexttt{x} :_{2\varepsilon} \R^2,
  \stexttt{y}:_{\varepsilon} \R^2 \vdash \stexttt{SVecAdd a x y} : \R.
\]
In the typing judgment for \stexttt{SVecAdd}, we observe that the linear
variable \stexttt{x} has a backward error bound of $2 \varepsilon$, while the
linear variable \stexttt{y} has backward error bound of only $\varepsilon$.
This difference arises because \stexttt{x} accumulates $\varepsilon$ backward
error from \stexttt{ScaleVec} and an additional $\varepsilon$ backward error
from the vector addition with the linear variable \stexttt{y}.

Now, given discrete variables $\stexttt{a} : \R$
and $\stexttt{b} : \R$, and $\stexttt{v}: \R^{2}$,
along with the linear variables 
$\stexttt{M}: \R^{2 \times 2}$ 
and $\stexttt{u}: \R^{2}$, 
we can compute a matrix-vector product of 
\stexttt{M} and \stexttt{v}
with
\stexttt{MatVecMul},
and use the result in the scaled matrix-vector 
product, \stexttt{SMatVecMul}:
\vspace{-1em}
\begin{center}
\begin{minipage}[t]{.45\textwidth}
\begin{lstlisting}
MatVecMul M v :=
let (m0, m1) = M in 
let u0 = InnerProduct m0 v in 
let u1 = InnerProduct m1 v in 
(u0, u1)
\end{lstlisting}
\end{minipage}
\begin{minipage}[t]{.45\textwidth}
\begin{lstlisting}
SMatVecMul M v u a b :=
let x = MatVecMul M v in
let y = ScaleVec b u  in
SVecAdd a x y 
\end{lstlisting}
\end{minipage}
\end{center}
For \stexttt{MatVecMul}, we rely on a program \stexttt{InnerProduct},
which computes the dot product of two $2 \times 2$ vectors. Notably,
\stexttt{InnerProduct} differs from the \stexttt{DotProd2} program 
described in \Cref{sec:overview} because it assigns backward error only onto 
the first vector. The type of this program is:
\[
  \stexttt{v}: \R^2 \mid \stexttt{u} :_{2\varepsilon}
  \R^2 \vdash \stexttt{InnerProduct u v}: \R
\]
The \bea{} programs \stexttt{MatVecMul} and 
\stexttt{SMatVecMul} have the following valid typing judgments: 
\begin{align*}
  \stexttt{v} :  \R^2 \mid \stexttt{M} :_{2\varepsilon}  \R^{2\times 2} &\vdash
  \stexttt{MatVecMul M v} : \R^2
  \\
  \stexttt{a}:\R, \stexttt{b}:\R, \stexttt{v}:\R^{2} 
  \mid \stexttt{M} :_{4\varepsilon} \R^{2\times 2},
  \stexttt{u}:_{2\varepsilon}\R^{2} &\vdash \stexttt{SMatVecMul M v u a b}:
  \R^2
\end{align*}

By error soundness, these judgments imply that the vector \stexttt{u} has at
most $2 \varepsilon$ backward error and the matrix \stexttt{M} has at most $4
\varepsilon$ backward error. The bound for \stexttt{M} has two sources:
\stexttt{MatVecMul} contributes at most $2 \varepsilon$ backward error, and
\stexttt{SVecAdd} adds at most an additional $2 \varepsilon$ backward error, 
resulting in a backward error bound of $4 \varepsilon$. Likewise, \stexttt{u}
accumulates at most $2 \varepsilon$ backward error from \stexttt{ScaleVec}
and \stexttt{SVecAdd}, each contributing at most $\varepsilon$. As shown in
\Cref{sec:evaluation}, \bea{}'s inferred bounds match the worst-case
theoretical backward error bounds from the literature.

These examples highlight the compositional nature of \bea's analysis:
like all type systems, the type of a \bea{} program is derived from the types
of its subprograms. While the numerical analysis literature is unclear on
whether (and when) backward error analysis can be performed compositionally
(e.g., ~\citep{Bornemann:2007:backwardcompose}), \bea{} demonstrates that this
is in fact possible.

\subsection{Polynomial Evaluation} 
To demonstrate \bea{}'s fine-grained backward error analysis, we start with two
simple programs for polynomial evaluation. The first, \stexttt{PolyVal},
follows a naive scheme, computing each term by multiplying the corresponding
coefficient with successive powers of the variable and summing the results.
The second, \stexttt{Horner}, uses Horner's method, which rewrites the
polynomial to minimize operations \cite[p.94]{Higham:2002:Accuracy}. We focus
on second-degree polynomials; \Cref{sec:implementation} describes a prototype
implementation of \bea{} and evaluates its inferred bounds for higher-degree
cases.

Given a tuple $\stexttt{a} : \R^3$ of coefficients and a discrete variable
$\stexttt{z} : \R$, the \bea{} programs for evaluating a second-order
polynomial $p(z) = a_0 + a_1z + a_2z^2$ using naive polynomial evaluation and
Horner's method are shown below. 
\vspace{-1em}
\begin{center}
\begin{minipage}[t]{.5\textwidth}
\begin{lstlisting}
PolyVal a z :=
let (a0, a') = a  in 
let (a1, a2) = a' in
let y1  = dmul z a1  in 
let y2' = dmul z a2  in
let y2  = dmul z y2' in
let x   = add a0 y1  in
add x y2
\end{lstlisting}
\end{minipage}
\begin{minipage}[t]{.5\textwidth}
\begin{lstlisting}
Horner a z :=
let (a0, a') = a   in 
let (a1, a2) = a'  in
let y1 = dmul z a2 in
let y2 = add a1 y1 in
let y3 = dmul z y2 in
add a0 y3
\end{lstlisting}
\end{minipage}
\end{center}
\noindent
Recall from \Cref{sec:overview} that the $\textbf{dmul}$ operation assigns
backward error onto its second argument; in the programs above, the operation
indicates that backward error should not be assigned to the discrete variable
\stexttt{z}. Using \bea{}'s type system, the following typing judgments are
valid:
\[ 
  \stexttt{z} : \R \mid \stexttt{a} :_{3\varepsilon} \R^3 \vdash \stexttt{PolyVal a z}: \R 
  \qquad  \qquad
  \stexttt{z} : \R \mid \stexttt{a} :_{4\varepsilon} \R^3 \vdash \stexttt{Horner a z} : \R
\] 
Surprisingly, although Horner's scheme is often considered more numerically
stable due to fewer floating-point operations, it can incur greater backward
error with respect to the coefficient vector. Examining each coefficient
individually provides more insight. By rewriting the implementations to take
each coefficient as a separate input, we can derive individual backward error
bounds for each coefficient. With this change, \bea{}'s type system yields the
following valid judgments:
\begin{align*}
  \stexttt{z} : \R \mid \stexttt{a0} :_{2\varepsilon}\R, \stexttt{a1}:_{3\varepsilon}\R, 
  \stexttt{a2}:_{3\varepsilon}\R &\vdash
   \stexttt{PolyValAlt z a0 a1 a2}: \R
  \\
  \stexttt{z} : \R \mid \stexttt{a0} :_{\varepsilon}\R, \stexttt{a1}:_{3\varepsilon}\R, 
  \stexttt{a2}:_{4\varepsilon}\R &\vdash
  \stexttt{HornerAlt z a0 a1 a2}: \R
\end{align*}
We see that Horner's scheme assigns more backward error onto the coefficients
of higher-order terms than lower-order terms, while naive polynomial evaluation 
assigns the same error onto all but the lowest-order coefficient. In this way, 
\bea{} can be used to investigate the numerical stability of different polynomial 
evaluation schemes by providing a fine-grained error analysis.

\subsection{Triangular Linear Solver}
One of the benefits of integrating error analysis with a type system is the
ability to weave common programming language features, such as conditionals
(if-statements) and error-trapping, into the analysis. We demonstrate these
features in our final, and most complex example: a linear solver for triangular
matrices. Given a lower triangular matrix $A\in\R^{2\times 2}$ and a vector
$b\in\R^2$, the linear solver should compute return a vector $x$ satisfying $Ax
= b$ if there is a unique solution.

We comment briefly on the program \stexttt{LinSolve}, shown below. The matrix
and vector are given as $\stexttt{((a00, a01), (a10, a11))}:\R^{2\times 2}$,
assuming $\stexttt{a01}=0$ and $\stexttt{(b0, b1)}:\R^2$. The program returns
either the solution vector $x$ or an error if the system does not have a unique
solution. The $\mathbf{div}$ operator has return type $\R + \unit$, where
$\unit$ represents division by zero. Ensuing computations can check if the
division succeeded using $\mathbf{case}$ expressions. This example also uses
the $!$-constructor to convert a linear variable into a discrete one; this is
required since the later entries in the vector $x$ depend on---i.e., require
duplicating---earlier entries in the vector.
\begin{center}
\begin{lstlisting}
    LinSolve ((a00, a01), (a10, a11)) (b0, b1) :=
    let x0_or_err = div b0 a00 in // solve for x0 = b0 / a00
    case x0_or_err of 
      inl (x0) => // if div succeeded
        dlet d_x0 = !x0 in // make x0 discrete for reuse
        let s0 = dmul d_x0 a10 in // s0 = x0 * a10
        let s1 = sub b1 s0 in // s1 = b1 - x0 * a10
        let x1_or_err = div s1 a11 in // solve for x1 = (b1 - x0 * a10) / a11
        case x1_or_err of 
          inl (x1) => inl (d_x0, x1) // return (x0, x1)
        | inr (err) => inr err // division by 0
    | inr (err) => inr err // division by 0
\end{lstlisting}
\end{center}
\noindent 
The type of \stexttt{LinSolve} is
$
  \stexttt{A}:_{5\varepsilon/2}\R^{2\times 2}, \stexttt{b}:_{3\varepsilon/2}\R^2
  \vdash \stexttt{LinSolve A b} : \R^2 + \unit.
$
Hence, \stexttt{LinSolve} has a guaranteed backward error bound of at most
$5\varepsilon/2$ with respect to the matrix \stexttt{M} and at most
$3\varepsilon/2$ with respect to the vector \stexttt{b}. If either of the
division operations fail, the program returns $\unit$. This example demonstrates
how various features in \bea{} combine to establish backward error guarantees
for programs involving control flow and duplication, via careful control of how
to assign and accumulate backward error through the program.

\begin{remark}[Conditionals and Backward Error]
The reader may wonder how programs with conditionals can have bounded backward error. 
While roundoff error in conditionals can cause large \emph{forward
error}---since small perturbations may change the control flow---the \emph{backward
error} can still be small. This is not a contradiction: forward error is
bounded by the product of the condition number and backward error
(\Cref{def:cnum}). Thus, even with poor forward stability, strong backward error
guarantees may still hold.
\end{remark}

\section{Implementation and Evaluation}\label{sec:implementation}
\subsection{{Implementation}}\label{sec:algorithm}
We implemented a type checking and coeffect inference algorithm for \bea{} in
OCaml. It is based on the sensitivity inference algorithm introduced by
\citet{ADA:2014:typecheck}, which is used in implementations of
\emph{Fuzz}-like languages \cite{DFUZZ, NUMFUZZ}. Given a \bea{} program
without any error bound annotations in the context, the type checker ensures
the program is well-formed, outputs its type, and infers the tightest possible
backward error bound on each input variable. Using the type checker, users can
write large \bea{} programs and automatically infer backward error with respect
to each variable.

More precisely, let $\Gamma^\bullet$ denote a context \emph{skeleton}, a linear
typing context with no coeffect annotations. If $\Gamma$ is a linear context,
let $\overline{\Gamma}$ denote its skeleton. Next, we say $\Gamma_1$ is a
\emph{subcontext} of $\Gamma_2$, $\Gamma_1\sqsubseteq\Gamma_2$, if
$\dom{\Gamma_1}\subseteq\dom{\Gamma_2}$ and for all $x:_r\sigma\in\Gamma_1$, we
have $x:_q\sigma\in\Gamma_2$ where $r \leq q$. In other words, $x$ has a
tighter backward error bound in the subcontext. Now, the input to the type
checking algorithm is a typing context skeleton $\Phi\mid \Gamma^\bullet$ and a
\bea{} program, $e$. The output is the type of the program $\sigma$ and a
linear context $\Gamma$ such that $\Phi\mid\Gamma\vdash e:\sigma$ and
$\overline{\Gamma}\sqsubseteq\Gamma^\bullet$. Calls to the algorithm are
written as $\Phi\mid \Gamma^\bullet;e\Rightarrow \Gamma;\sigma$. The algorithm
uses a recursive, bottom-up approach to build the final context.

For example, to type the \bea{} program $(e,f)$, where $e$ and $f$ are
themselves programs, we use the algorithmic rule
\[
  \AXC{$\Phi\mid\Gamma^\bullet;e\Rightarrow \Gamma_1;\sigma$} \AXC{$\Phi\mid
  \Gamma^\bullet;f\Rightarrow\Gamma_2;\tau$}
  \AXC{$\dom\Gamma_1\cap\dom\Gamma_2=\emptyset$} \RightLabel{($\tensor$ I)}
  \TrinaryInfC{$\Phi\mid\Gamma^\bullet;(e,f)\Rightarrow\Gamma_1,\Gamma_2;\sigma\otimes\tau$}
  \bottomAlignProof \DisplayProof
\]
Concretely, we recursively calling the algorithm on $e$ and $f$, and then
combine their resulting contexts. The output contexts discard unused variables
from the input skeletons; thus, the requirement
$\dom\Gamma_1\cap\dom\Gamma_2=\emptyset$ ensures the strict linearity
requirement is met. 

The type checking algorithm is sound and complete, meaning that it agrees
exactly with \bea{}'s typing rules. Precisely: 
\begin{theorem}[Algorithmic Soundness]\label{thm:algo_sound} If
  $\Phi\mid\Gamma^\bullet;e\Rightarrow \Gamma;\sigma$, then
  $\overline{\Gamma}\sqsubseteq\Gamma^\bullet$ and the derivation
$\Phi\mid\Gamma\vdash e:\sigma$ exists.
\end{theorem}
\begin{theorem}[Algorithmic Completeness]\label{thm:algo_complete} If
  $\Phi\mid\Gamma\vdash e:\sigma$ is a valid derivation in \bea{}, then there
  exists a context $\Delta\sqsubseteq\Gamma$ such that
  $\Phi\mid\overline{\Gamma}; e\Rightarrow \Delta;\sigma$. 
\end{theorem} 
\ifshort
\else
The full algorithm and proofs of its correctness are given in
\Cref{app:algorithm}.
\fi
The \bea{} implementation is parametrized only by unit roundoff, which is
dependent on the floating-point format and rounding mode and is fixed for a
given analysis. The implementation computes rounding error bounds using 
OCaml's Float module, which follows the IEEE 754 standard with double-precision
(64-bit) numbers; the default rounding mode is round-to-nearest, ties-to-even.

\subsection{{Evaluation}}\label{sec:evaluation}
\begin{table}
\caption{ 
    The performance of \bea{} benchmarks with known backward error bounds
    from the literature. The Input Size column gives 
    the length of the input vector or dimensions of the input matrix; the 
    Ops column gives the total number of floating-point operations.  
    The Backward Bound column reports the bounds
    inferred by \bea{} and well as the standard bounds (Std.) from
    the literature. The Timing column reports the time in seconds for 
    \bea{} to infer the backward error bound. 
} 
\label{tab:benchStd} 
\begin{tabular}{l l l c c r} 
\hline
{Benchmark}
& {Input Size}
& {Ops}
& \multicolumn{2}{c}{Backward Bound} 
& {Timing (s)} 
\\ \cline{4-5}
& & & {{\bea}} & {{Std.}} & 
\\ \hline 
\multirow{4}{1em}{{DotProd}} 
& 20 & {39} & {2.22e-15} & {2.22e-15} & 0.004
\\ \arrayrulecolor{gray}\cline{2-6}
& 50 & {99} & {5.55e-15} & {5.55e-15} & 0.04
\\ \cline{2-6}
& 100 & 199 & {1.11e-14} & {1.11e-14} & 0.3 
\\ \cline{2-6}
& 500 & 999 & {5.55e-14} & {5.55e-14} & 30
\\ \arrayrulecolor{black}\hline
\multirow{4}{1em}{{Horner}} 
& 20 & {40} & {4.44e-15} & {4.44e-15} & 0.002
\\ \arrayrulecolor{gray}\cline{2-6}
& 50 & {100} & {1.11e-14} & {1.11e-14} & 0.02
\\ \cline{2-6}
& 100 & 200 & {2.22e-14} & {2.22e-14} & 0.1
\\ \cline{2-6}
& 500 & 1000 & {1.11e-13} & {1.11e-13} & 10
\\ \arrayrulecolor{black}\hline
\multirow{4}{1em}{{PolyVal}} 
& 10 & {65} & {1.22e-15} &{1.22e-15} & 0.004
\\ \arrayrulecolor{gray}\cline{2-6}
& 20 & {230} & {2.33e-15} & {2.33e-15} & 0.06
\\ \cline{2-6}
& 50 & 1325 & {5.66e-15} & {5.66e-15} & 5
\\ \cline{2-6}
& 100 & 5150 & {1.12e-14} & {1.12e-14} & 200
\\ \arrayrulecolor{black}\hline
\multirow{4}{1em}{{MatVecMul}} 
& 5 $\times$ 5 & 45 & {5.55e-16} & {5.55e-16} & 0.003
\\ \arrayrulecolor{gray}\cline{2-6}
& 10 $\times$ 10 & 190 & {1.11e-15} & {1.11e-15} & 0.1
\\ \cline{2-6}
& 20 $\times$ 20 & 780 & {2.22e-15} & {2.22e-15} & 6
\\ \cline{2-6}
& 50 $\times$ 50 & 4950 & {5.55e-15} & {5.55e-15} & 1000
\\ \arrayrulecolor{black}\hline
\multirow{4}{1em}{{Sum}} 
& 50 & {49} & {5.44e-15} & {5.44e-15} & 0.008
\\ \arrayrulecolor{gray}\cline{2-6}
& 100 & {100} & {1.10e-14} & {1.10e-14} & 0.04
\\ \cline{2-6}
& 500 & 499 & {5.54e-14} & {5.54e-14} & 4
\\ \cline{2-6}
& 1000 & 999 & {1.11e-13} & {1.11e-13} & 30
\\ \arrayrulecolor{black}\hline
\end{tabular}
\end{table} 

\begin{table}
\caption{ 
    A comparison of \bea{} to \citet{Fu:2015:BEA} on polynomial approximations
    of $\sin$ and $\cos$. The \bea{} implementation matches the programs
    evaluated by \citet{Fu:2015:BEA} for the given range of input values.
    $^*$The timing values 1310 and 1280 were taken from the \citet{Fu:2015:BEA}
    paper.
} 
\label{tab:benchFu}
\begin{tabular}{l l c c c c} 
\hline
{Benchmark} & {Range} & \multicolumn{2}{c}{Backward Bound} 
& \multicolumn{2}{c}{Timing (ms)}  \\
\cline{3-6} 
{} & {} & {{\bea}}
& {{\citet{Fu:2015:BEA}}}
& {{\bea}}
& {{\citet{Fu:2015:BEA}}} \\
\hline 
\stexttt{cos} & [0.0001, 0.01] & {1.33e-15} & {5.43e-09} & 1 & 1310* \\
\arrayrulecolor{gray}\hline
\stexttt{sin} & [0.0001, 0.01] & {1.44e-15} & {1.10e-16} & 1 & 1280*
\\ \arrayrulecolor{black}\hline
\end{tabular}
\end{table}

\begin{table}
\caption{ 
  A comparison of forward bounds. Forward error bounds for \bea{} are derived
  from backward error bounds using the relation \emph{forward error} $\le$
  \emph{condition number} $\times$ \emph{backward error} (see \Cref{def:cnum}),
  and the fact that all listed benchmarks have a relative condition number of
  exactly $\kappa_{rel}=1$. \bea{} computes the backward error statically; forward
  error is not directly computed. We compare the resulting bounds to those of
  NumFuzz and Gappa. For Gappa, we assume all variables are in the interval
  $[0.1,1000]$.
} 
\label{tab:benchFz} 
\begin{tabular}{l c c c c c c} 
\hline
{Benchmark}
& {Input Size}
& {Ops}
& \multicolumn{3}{c}{Forward Bound} 
\\ \cline{4-6}
& & & {{\bea}} & {NumFuzz } & Gappa \\
\hline 
Sum
& 500 & 499 & {1.11e-13} & {1.11e-13} & {1.11e-13} \\
\arrayrulecolor{gray}\hline
DotProd
& 500 & 999 & {1.11e-13} & {1.11e-13} & {1.11e-13} \\
\hline
Horner
& 500 & 1000 & {2.22e-13} & {2.22e-13} & {2.22e-13} \\
\hline
PolyVal
& 100 & {5150} & {2.24e-14} & {2.24e-14} & {2.24e-14} 
\\ \arrayrulecolor{black}\hline
\end{tabular}
\end{table}

In this section, we report results from an empirical evaluation of our \bea{}
implementation, focusing primarily on the quality of the inferred bounds. Since
\bea{} is the first tool to statically derive \emph{sound} backward error
bounds, a direct comparison with existing tools is challenging. We therefore
evaluate the inferred bounds using three complementary methods.

We first evaluate \bea{} against theoretical worst-case backward error bounds
from the literature, providing a benchmark for how closely \bea{}'s inferred
bounds approach these worst-case values. Next, we compare our results to those
from an analysis tool developed by \citet{Fu:2015:BEA}, which, to our
knowledge, provides the only automatically derived quantitative bounds on
backward error. While useful as a baseline, their results focus on
transcendental functions, whereas \bea{} targets larger programs involving
linear algebra primitives. Finally, we evaluate the quality of the backward
error bounds derived by \bea{} using forward error as a proxy. Specifically,
using known values of the relative componentwise condition number
(\Cref{def:cnum}), we compute forward error bounds from our backward error
bounds. This approach enables a comparison to existing tools focused on forward
error analysis. We compare our derived forward error bounds to those produced
by two tools that soundly and automatically bound relative forward error:
NumFuzz~\cite{NUMFUZZ} and Gappa~\cite{Gappa}. Both tools are capable of
scaling to larger benchmarks involving over 100 floating-point operations,
making them suitable tools for comparison with \bea{}. All of our experiments
were performed on a MacBook Pro with an Apple M3 processor and 16 GB of memory.

\subsubsection{Evaluation Against Theoretical Worst-Case Bounds}
\Cref{tab:benchStd} shows results for benchmark problems with known backward
error bounds from the literature. Each benchmark was run on inputs of
increasing size (given in Input Size), with the total number of floating-point
operations listed in the Ops column. The Std. column gives the worst-case
theoretical backward error bound assuming double-precision and
round-to-nearest, obtained from \citet[p.63, p.94,
p.82]{Higham:2002:Accuracy}. For simplicity, \bea{} programs use a single
linear input variable, with remaining inputs treated as discrete variables. The
maximum elementwise backward bound is computed with respect to the linear
input. The \bea{} programs emulate the following analyses for input size $N$: 
\begin{itemize}
\item \stexttt{DotProd} computes the dot product
of two vectors in $\R^N$, assigning backward error to a single vector.
\item \stexttt{Horner} evaluates an $N$-degree polynomial using Horner's
scheme, assigning backward error onto the vector of coefficients.
\item \stexttt{PolyVal} naively evaluates an $N$-degree polynomial, 
assigning backward error onto the vector of coefficients.
\item \stexttt{MatVecMul} computes the product of a matrix in $\R^{N\times
N}$ and a vector in $\R^N$, assigning backward error onto the matrix.
\item \stexttt{Sum} sums the elements of a vector
in $\R^N$, assigning backward error onto the vector.
\end{itemize}
Since we report the backward error bounds from the literature under the
assumption of double-precision and round-to-nearest, we instantiated \bea{}
with a unit roundoff of $u=2^{-53}$.

\subsubsection{Comparison to Dynamic Analysis}
\Cref{tab:benchFu} shows the comparison between \bea{} and the
optimization-based tool for automated backward error analysis by
\citet{Fu:2015:BEA}. The benchmarks are polynomial approximations of $\sin$
and $\cos$ implemented using Taylor series expansions, following the GNU C
Library (glibc) version 2.21. Our \bea{} implementations match the glibc
implementations on the input range $[0.0001,0.01]$. Specifically, the Taylor
series expansions implemented in \bea{} only match the glibc implementations
for inputs in this range. Since the glibc implementations analyzed by
\citet{Fu:2015:BEA} use double-precision and round-to-nearest, we instantiated
\bea{} with a unit roundoff of $u=2^{-53}$. While we include timing data for
reference, the implementation from \citeauthor{Fu:2015:BEA} is not publicly 
available or maintained, preventing direct runtime comparisons; thus, all values 
are taken from Table 6 of \citet{Fu:2015:BEA}. 

\subsubsection{Using Forward Error as a Proxy}
We can compare the quality of the backward error bounds derived by \bea{} to
existing tools by using forward error as a proxy. Specifically, for benchmarks
with a known relative componentwise condition number $\kappa_{rel}$, we can
convert relative backward error bounds to relative forward error bounds
\cite[Definition 2.12]{hohmann:2003:numerical}: 
\begin{definition}[Relative Componentwise Condition Number]
\label{def:cnum}
The \emph{relative componentwise condition number} of a scalar function $f:
\R^n \rightarrow \R$ is the smallest number $\kappa_{rel} \ge 0$ such that,
for all $x\in\R^n$, 
\begin{align}\label{eq:krel}
d_{\R}(f(x),\tilde{f}(x)) \le \kappa_{rel} \max_{i} d_\R(x_i,\tilde{x}_i)
\end{align}
\end{definition}
\noindent where $\tilde{f}$ is the approximating program and $\tilde{x}$ is the 
perturbed input such that $f(\tilde{x})=\tilde{f}(x)$. Here,
$d_{\R}(f(x),\tilde{f}(x))$ is the \emph{relative forward error}, and $\max_{i}
d_\R(x_i,\tilde{x}_i)$ is the maximum \emph{relative backward error}. Thus, when
$\kappa_{rel}$ is known, we can compute relative
forward error bounds from the relative backward error bounds inferred by
\bea{}. 

\Cref{tab:benchFz} presents the results for several benchmark problems with
$\kappa_{rel}=1$, so the maximum relative backward error directly bounds the
relative forward error. For example, the relative condition number for summing
$n$ values $(a_i)_{1\le i\le n}$ is $\kappa_{rel} = \sum_{i=1}^n |a_i|/
|\sum_{i=1}^n a_i|$ \cite{MullerBook}, which reduces to $1$ when all $a_i >0$.
In each benchmark in \Cref{tab:benchFz}, this condition holds only under
the assumption of strictly positive inputs. NumFuzz requires this assumption
for soundness, and we enforced it in Gappa by specifying inputs intervals of
$[0.1, 1000]$. To ensure consistency, all tools were instantiated with a unit
roundoff of $u=2^{-52}$. Each benchmark was implemented separately in each tool
to compute the reported error. The results reported in \Cref{tab:benchFz} are
identical to the second digit because each tool computes bounds in double
precision, and the bounds are optimally tight (i.e., they match the theoretical
worst-case error bounds), so agreement to the second digit is expected.

\subsubsection{Evaluation Summary} 
The main conclusions from our evaluation results are as follows.

\emph{\textbf{\bea{}'s backward error bounds are useful}}: In all of our
experiments, \bea{} produced competitive error bounds. Compared to the
backward error bounds reported by \citet{Fu:2015:BEA} for their dynamic
backward error analysis tool, \bea{} was able to derive \emph{sound} backward
error bounds that were close to or better than those produced by the dynamic
tool. Furthermore, \bea{}'s sound bounds precisely match the worst-case
theoretical backward error bounds from the literature, demonstrating that our
approach guarantees soundness without being overly conservative. Finally,
when using forward error as a proxy to assess the quality of \bea{}'s backward
error bounds, we find that \bea{}'s bounds again precisely match the bounds
produced by NumFuzz and Gappa.

\emph{\textbf{\bea{} performs well on large programs}}: In our comparison to
worst-case theoretical error bounds, we find that \bea{} takes under a minute to
infer backward error bounds on benchmarks with fewer than 1000 floating-point
operations. Overall, \bea{}'s performance scales linearly with the number of
floating-point operations in a benchmark.

\section{Semantics and Metatheory}\label{sec:metatheory}
Recall the intuition behind the guarantee for \bea{}'s type system: a
well-typed term of the form 
$\Phi \mid x_1:_{r_1} \sigma_1, \dots, x_i:_{r_i} \sigma_i \vdash e : \tau$
is a program that has at most $r_i$ backward error with respect to each
variable $x_i$, and has {no backward error} with respect to the {discrete}
variables in the context $\Phi$. To prove soundness for the type system, we
provide a categorical semantics for \bea{}: we interpret every typing judgment
as a morphism in a suitable category, and conclude soundness from properties of
the morphisms in our category. While this approach is standard, the category 
we use is not standard and is a novel contribution of our work. 

\subsection{\Bel: The Category of Backward Error Lenses}\label{sec:meta_errcat}
The key semantic structure for \bea{} is the category \Bel{} of \emph{backward
error lenses}.  Each morphism in \Bel{} corresponds to a \emph{backward error
lens}, which consists of a \emph{triple} of transformations describing an an
ideal computation, its floating-point approximation, and a \emph{backward map}
that serves as a constructive mechanism for witnessing the existence of a
backward error result: 

\begin{definition}[Backward Error Lenses]\label{def:error_lens} 
First, we define a \emph{slack distance space} to be a triple $(X,d_X,r_X)$
where $X$ is a set, $d_X:X\times X\to \R_{\geq 0}\cup\{\infty\}$ is a function
providing distance, and $r_X\in \R_{\geq 0}\cup\{\infty\}$ is a constant called the slack.
It is not required that $d_X$ satisfies the metric laws.
A \emph{backward error lens} between the slack distance spaces $(X,d_X,r_X)$ and $(Y,d_Y,r_Y)$
has three components: a \emph{forward map} $f: X \rightarrow Y$,
an \emph{approximate map} $\tilde{f}: X \rightarrow Y$, and 
a \emph{backward map} $b:  X\times Y \to X$ defined 
for every $x \in X$ and every $y \in Y$ such that 
$ d_Y(\tilde{f}(x),y)<\infty$.  
These three components satisfy two properties for every $ x \in
X$ and $y \in Y$, as long as ${d_Y(\tilde{f}(x),y)}<\infty$: 
\begin{align*}
d_X(x,b(x,y))-r_X &\le d_Y(\tilde{f}(x),y)-r_Y \label{eq:prop1} \tag*{(Property 1)} \\
f(b(x,y))&=y \label{eq:prop2} \tag*{(Property 2)} 
\end{align*}
\end{definition}  
\noindent Note that $\infty+a=\infty-a=\infty$ for any $a\in\R\cup\{\infty\}$, 
but $a-\infty=-\infty$ for any $a\in\R$. 

Backward error lenses generalize the notion of backward stability from
\Cref{def:bstab}. A lens $(f,\tilde{f},b) \in X \rightarrow Y$ ensures two
properties: (1) for any input $x \in X$, any output ${y} \in Y$ at finite
distance from the floating-point result $\tilde{y} = \tilde{f}(x) \in Y$ is the
image of some $\tilde{x} \in X$ under $f$; and (2) the distance between $x$ and
$\tilde{x}$ increases smoothly with the distance between $\tilde{y}$ and $y$.
In contrast, \Cref{def:bstab} only guarantees that $\tilde{y}$ can be matched
by an exact computation $f(\tilde{x})$ with $\tilde{x}$ close to $x$. 
\ifshort
Our generalization enables compositional reasoning about backward error, and this
compositional property allows us to form the category \Bel{} with slack
distance spaces as objects and backward error lenses as morphisms.
\else
Our generalization enables compositional reasoning about backward error, and this
compositional property allows us to form the category \Bel{} with slack
distance spaces as objects and backward error lenses as morphisms; a precise
definition of the category is given in \Cref{def:lensC}.
\fi

\subsubsection{Basic Constructions in \Bel{}}
In order to interpret typing judgments in \bea{} as morphisms in \Bel{}, we
must define lenses for the tensor product, coproduct, and a \emph{graded
comonad}. 
\ifshort
\else
These lenses, along with other basic constructions in \Bel{}, are defined in
\Cref{app:check_bel_cons}. 
\fi

The key construction in \Bel{} that enables our semantics to capture the
standard backward stability guarantee in \Cref{def:bstab} is a graded comonad.
To summarize, the graded comonad on \Bel{} 
\ifshort
\else
(see \Cref{app:comonad}) 
\fi
is defined by the family of functors 
\begin{align*}
\{D_r :
  \Bel\rightarrow \Bel \ | \ r \in \mathcal{R}\}
\end{align*}
where the pre-ordered monoid $\mathcal{R}$ is the non-negative real numbers
${\R^{\ge 0}}$ with the usual order and addition. On objects, $D_r:
\Bel{} \rightarrow \Bel{}$ takes a slack distance space $(X,d_X, r_X)$ to a 
slack distance space $(X,d_X, r_X+r)$. On morphisms, $D_r$ takes an error lens 
$(f,\tilde{f},b) : A \rightarrow X$ to an error lens $(D_rf,D_r\tilde{f},D_rb)
: D_rA \rightarrow D_rX$ where $(D_rg)x \triangleq g(x)$. 

Crucially, using the graded comonad, we can show that the standard backward
stability guarantee is a special case of the general guarantee provided by
backward error lenses: given a lens $(f,\tilde{f},b) : D_{\alpha\varepsilon} X
\rightarrow Y$, for every input $x \in X$, the input $\tilde{x} =
b(x,\tilde{f}(x)) \in X$ exists such that $f(\tilde{x}) = \tilde{f}(x)$
(Property 2) and $d_X(x,\tilde{x}) \le \alpha\varepsilon$ (Property 1).

\subsubsection{Interpreting \bea{}}
With the basic structure of \Bel{} in place, we can now interpret the types and
typing judgments of \bea{} as objects in \Bel{} and morphisms in \Bel{},
respectively. 
\ifshort 
First, every type $\tau$ is interpreted as a distance space $\denot{\tau} \in$
\Bel{} with reflexive distance and slack $0$.
\else
First, every type $\tau$ is interpreted as a distance space $\denot{\tau} \in$
\Bel{} with reflexive distance and slack $0$, defined inductively in
\Cref{app:interp_bea}.
\fi
For instance, the numeric type $\num$ is interpreted as the real numbers with
the RP metric (\Cref{eq:olver}). Typing contexts $\Phi~\mid~\Gamma$ are also
interpreted as objects $\denot{\Phi~\mid~\Gamma} \in$ \Bel{}. The graded
comonad $D_r$ is used to interpret linear variable bindings: $\denot{\Phi \mid
x:_r \sigma} \triangleq \denot{\Phi} \otimes D_r \denot{\sigma}$.

Finally, we interpret discrete types $\alpha$ and contexts $\Phi$ as
\emph{discrete distance spaces} where the distance between any two distinct
points is infinite. It turns out that discrete distance spaces belong to a
natural subcategory $\Del$ of \emph{discrete} objects and error lenses. Any
lens from a discrete object $A$ is guaranteed not to push backward error onto
$A$, and discrete objects can be duplicated.

Given these ingredients, we can define the interpretation of well-typed terms in \bea{}:
\begin{definition}(Interpretation of \bea{} terms.)\label{def:interpL}
	We can interpret each well-typed \bea{} term
  $\Phi \mid \Gamma \vdash e : \tau$ as an
	error lens $\denot{e} : \denot{\Phi} \otimes \denot{\Gamma} \rightarrow
	\denot{\tau}$ in \Bel, by induction on the typing derivation. 
\end{definition}
\ifshort
\else
\noindent The details of each construction for \Cref{def:interpL} are provided
in \Cref{app:interp_bea}.
\fi

\subsection{Metatheory}
The proof of backward error soundness (\Cref{thm:main}) relates the lens
semantics described above to an exact and a floating-point operational
semantics.  In the lens semantics, a \bea{} program corresponds to a triple of
set-maps $(f,\tilde{f},b)$ describing the ideal computation, its floating-point
approximation, and a map that constructs the backward error between them.
Intuitively, the \emph{set semantics} of the first component of the lens should
model the ideal operational semantics, while the \emph{set semantics} of the
second component of the lens should model the floating-point operational
semantics. We achieve this distinction for \bea{} programs by defining an
intermediate language, which we call \LangS, where programs denote morphisms in
$\Set$. We then define an ideal ($\stepid$) and approximate ($\stepap$)
big-step operational semantics for \LangS, and relate these semantics to the
backward error lens semantics of $\bea{}$ via the $\Set$ semantics of \LangS.
Importantly, \LangS{} is semantically sound and computationally adequate: a
\LangS{} program computes to a value if and only if their interpretations in
\Set{} are equal.
\ifshort
\else
The ideal and approximate operational semantics along with the details of
\LangS{} are deferred to \Cref{app:LangS}.
\fi

\paragraph{Backward Error Soundness}
With the interpretation of \bea{} terms in place, the path towards a proof of
\emph{backward error soundness} is clear.  

\begin{theorem}[Backward Error Soundness]\label{thm:main2}
Let the well-typed \bea{} program
\[ 
z_1:\alpha_1,\dots,z_n:\alpha_n \mid  x_1:_{r_1}\sigma_1,\dots,x_m:_{r_m}\sigma_m
\vdash e : \tau
\] 
be given, and let the sequences of values $(p_i)_{1 \le i \le n}$ and $(k_j)_{1
\le j \le m}$ be given also. Suppose $ \vdash p_i : \alpha_i$ and $ \vdash k_j
:  \sigma_j$ for all $i \in 1, \dots, n$ and $j \in 1, \dots, m$. If
$e[p_1/z_1]\cdots[{k_m}/x_m] \stepap v$, then the sequence of values
$(\tilde{k}_j)_{1 \le i \le m}$ exists such that
$e[p_1/z_1]\cdots[{\tilde{k}_m}/x_m] \stepid v$, and 
$d_{\denot{\sigma_i}}({k}_i,\tilde{k}_i) \le r_i$ for every $j \in 1, \dots, m$. 
\end{theorem}

\begin{proof}
Given that $e$ is a well-typed \bea{} term, from \Cref{def:interpL} we have 
\begin{align*}
  \denot{z_1: \alpha_1, \dots, z_n: \alpha_n \mid x_1:_{r_1} \sigma_1, 
  \dots, x_m:_{r_m} \sigma_m \vdash e : \tau} 
  = (f,\tilde{f},b) : \denot{\alpha_1} \otimes
  \dots \otimes  D_{r_m}\denot{\sigma_m} \rightarrow \denot{\tau}.
\end{align*}
If $e[p_1/z_1]\dots[k_m/x_m] \stepap v$ for the well-typed substitutions
$(p_i)_{1 \le i \le n}$ and $(k_j)_{1 \le j \le m}$, then we can guarantee our
desired backward error result if we can witness the existence of a well-typed
substitution $(\tilde{k})_{1 \le i \le m}$ such that
$e[p_1/z_1]\dots[\tilde{k}_m/x_m] \stepid  v$, and
$d_{\denot{\sigma_j}}(k_j,\tilde{k}_j) \le r_j$ for every $j \in 1, \dots, m$.  
The key idea is to use the backward
map $b$ to construct the well-typed substitutions $(\tilde{p}_i)_{1 \le i \le
n}$ and $(\tilde{k}_j)_{1 \le j \le m}$ such that 
\[
((\tilde{p}_1,\dots,\tilde{p}_n),(\tilde{k}_1,\dots,\tilde{k}_m)) =
b(((p_1,\dots,p_n),(k_1,\dots,k_m)),v).
\]
Then, from Property 2 of the lens, we have $f
((\tilde{p}_1,\dots,\tilde{p}_n),(\tilde{k}_1,\dots,\tilde{k}_m)) =
v$.
\ifshort 
Projecting into \LangS{}, we can use this result along with computational
adequacy to show $e[\tilde{p}_1/z_1]\dots[\tilde{k}_m/x_m] \stepid v$. From
the semantic soundness of \LangS{} we then have
\else
We can use this result along with pairing (\Cref{lem:pairing}) and
computational adequacy (\Cref{thm:adequacy}) to show
$e[\tilde{p}_1/z_1]\dots[\tilde{k}_m/x_m] \stepid v$. From semantic soundness
(\Cref{thm:soundid}) and Property 2, we then have
\fi
\[
\tilde{f} (({p}_1,\dots,{p}_n),({k}_1,\dots,{k}_m)) = f
{((\tilde{p}_1,\dots,\tilde{p}_n),(\tilde{k}_1,\dots,\tilde{k}_m))}.
\]

Two things remain to be shown. First, we must show the values of discrete type
carry no backward error, i.e., $\tilde{p}_i = p_i$ for every $i \in 1, \dots,
n$.  Second, we  must show the values of linear type have bounded backward
error. Both follow from Property 1 of the error lens: from the inequality
\[ 
\max\left\{d_{\denot{\alpha_1}}(\tilde{p}_1,{p}_1),\dots,
d_{\denot{\sigma_m}}(\tilde{k}_m,{k}_m)-r_m\right\} \le
d_{\denot{\tau}}\left(v,v\right) = 0
\] 
we can conclude $\tilde{p}_i = p_i$ for every 
$i \in 1, \dots, n$ and $d_{\denot{\sigma_j}}(\tilde{k}_j,{k}_j) \le r_j$ 
for every $j \in 1, \dots, m$.
\end{proof}

\ifshort
\else
The full details of the proof are provided in \Cref{sec:app_soundness}.
\fi

\section{Related Work}\label{sec:rw}

\subsubsection*{Automated Backward Error Analysis} Existing automated methods
for backward error analysis are based on automatic differentiation and
optimization techniques, but unlike \bea{}, they lack soundness guarantee. For
instance, Miller's algorithm \cite{Miller:78:BEA} uses automatic
differentiation to estimate backward error via partial derivatives, and can
support rich program features \cite{Gati:2012:BEA}. 

The first optimization-based tool, introduced by \citet{Fu:2015:BEA}, uses a
generic minimizer to compute a backward error function mapping an input-output
pair to an input that produces the same output under a higher-precision version
of the program. This function is used to estimate the maximal backward error
for a range of inputs using Markov Chain Monte Carlo techniques. 

\bea{} and this optimization-based approach both require the direct
construction of a backward function: both require an ideal function, an
approximating function, and an explicit backward function. However, unlike
\bea{}, the existing optimization method must perform a sometimes costly
analysis to construct the ideal (higher-precision) and backward functions. In
\bea{}, it is not necessary to explicitly construct these functions since they
are built into our semantic model.

\subsubsection*{Residual-Based Methods}
\emph{A posteriori} estimates of backward error can be computed dynamically
using \emph{residual-based methods} \cite[Section
1.4.3]{Corless:2013:Analysis}, which use the computed result to estimate
backward error. These methods are widely used in high-performance solvers for
linear systems of equations such as SuperLU \cite{SUPERLU}, STRUMPACK
\cite{STRUMPACK1, STRUMPACK2}, cuSolver \cite{CUSOLVER}, SPRAL/SSIDS
\cite{SPRAL} and PaStiX \cite{PASTIX}. These solvers compute a solution
$\tilde{x} \in \R^n$ to systems of the form $Ax =b$, where $A \in \R^{n \times
n}$ is a matrix, $b \in \R^n$ is a vector, and $x\in\R^n$ denotes the exact
(ideal) solution. The backward error is then estimated using the residual $r =
b - A \tilde{x}$. This approach is valuable for assessing solver quality on
ill-conditioned problems, as demonstrated in recent work on power grid
optimization \cite{swirydowicz:2022}.

While \bea{} does not yet scale to the full complexity of solvers
like STRUMPACK, it could be integrated into more fully-fledged programming
languages, such as Granule \cite{granule}. Although it remains an open
question whether linearity in such a language would be too restrictive for
implementing a production scale solver, this challenge could be addressed by
selectively incorporating dynamic residual-based methods in parts of the
program where static guarantees are infeasible. 

\subsubsection*{Type Systems and Formal Methods}
Many formal verification tools have been developed for \emph{forward} rounding
error analysis, including both static approaches
\cite{Rosa1,Rosa2,VCFLOAT2,FPTaylor,Gappa,DAISY,ASTREE,Fluctuat,REAL2FLOAT} and
dynamic approaches \cite{SATIRE}. In contrast, the formal methods literature on
\emph{backward error} analysis is sparse. The only effort we are aware of is
the LAProof library developed by \citet{Kellison:Arith:2023}, which uses Coq to
verify backward error bounds for linear algebra programs. However, proof
assistants like Coq are not widely used by numerical analysts. \bea{} lowers
the barrier to entry by maintaining formal guarantees without requiring manual
proofs.  

Other type-based approaches to \emph{forward} rounding error analysis include
NumFuzz \cite{NUMFUZZ} and a system by \citet{Martel:types:2018}. Like \bea{},
NumFuzz uses a linear type system and coeffects---but for \emph{forward error}.
While \bea{} share syntactic similarities, semantically they are entirely
different. First, errors in NumFuzz are tracked through an error monad, which
captures forward error; \bea{} tracks error in the context, through a comonad.
Second, while both NumFuzz and \bea{} use a graded comonad, the graded comonad
in NumFuzz scales the metric to track function sensitivity, while the graded
comonad in \bea{} shifts the metric to track backward error. Finally, similar to
other \emph{Fuzz}-like languages, NumFuzz interprets programs in the category of
metric spaces, which lacks the necessary structure for reasoning about backward
error.
% To address this, we introduced the novel category of backward error lenses,
% offering a completely new semantic foundation that distinguishes \bea{} from all
% languages in the \emph{Fuzz} family.

\subsubsection*{Linear type systems and coeffects}
Our type system is most closely related to coeffect-based type systems. We
cannot comprehensively survey this active area of research; the thesis by
\citet{Petricek:2014:coeffects} provides an overview. Type systems in this
area include the Fuzz family of programming languages \citep{FUZZ,DFUZZ} and
the Granule language \citep{DBLP:journals/pacmpl/OrchardLE19}. A notable
difference in our approach is that we enforce strict linearity for graded
variables, with a separate context for reusable variables. 

\subsubsection*{Lenses and bidirectional programming}
Our semantic model is inspired by lenses, introduced by
\citet{Foster:2007:trees} to address the view-update problem in databases.
Lenses have been rediscovered in many contexts, ranging from categorical proof
theory and G\"odel's Dialectica translation \citep{dialectica} to open games
\citep{Ghani:2018:gamelenses} and supervised learning
\citep{DBLP:conf/lics/FongST19}. While our backward error lenses are formally
similar to prior notions, we are not aware of any existing lens framework that
subsumes ours.

\section{Conclusion and Future Directions}\label{sec:conc}
\bea{} is a typed first-order programming language that guarantees backward
error bounds. Its type system combines three elements: a notion of
distances for types, a coeffect system tracking backward error, and a
linear type system controlling the flow of backward error though programs. 
\ifshort
Although the backward error analysis modeled by \bea{} is 
more general than the standard approach, we can capture the standard definition
as a special case, as shown by our main theorem of backward error soundness
(\Cref{thm:main}). 
\fi
A benefit of \bea{} is its compositional structure: when backward error
bounds exist, they are derived compositionally from the backward error bounds
of subprograms. The linear type system correctly rejects programs that are not
backward stable, yet remains flexible enough to express classic backward-stable
algorithms. As the first static analysis framework for backward error, \bea{}
opens several paths for future work. We conclude with two promising directions
for future work. 

\subsubsection*{Higher-order functions.}
\bea{} does not currently support higher-order functions, which limits code reuse.
It is unclear whether \Bel{} supports linear exponentials, which are needed to
interpret function types. While most lens categories do not support higher-order
functions, there are some notable exceptions where the lens category is symmetric
monoidal closed (e.g., \citet{dialectica}). Connecting \bea{} to these lens
categories could suggest ways to support higher-order functions in the future. 
For now, we view a first-order language for backward error analysis as a meaningful 
step, particularly since most tools for forward error analysis are also first-order.

\subsubsection*{Richer backward error.}
This work focuses on a basic form of backward error guarantees, but richer
notions have been studied in the numerical analysis literature. For example, in 
\emph{probabilistic backward error} analysis, forward computations are randomized~(e.g.,
\citep{doi:10.1137/20M1334796}). In \emph{structured backward error}
analysis~(e.g., \citep{doi:10.1137/0613014}), the approximate input must satisfy 
additional structural constraints beyond mere proximity to the exact input---for 
instance, preserving matrix symmetry or rank.
Ideas from effectful and dependent lenses may be useful in extending our
approach to these settings.

\section*{Artifact}
The artifact for the implementation of \bea{} described in \Cref{sec:implementation} 
is available online \cite{Bean:artifact:2025}.

\begin{acks}
  We thank Andrew Appel, Dominic Orchard, Jessica Richards, and the anonymous
  reviewers for their close reading and valuable suggestions. This material is
  based upon work supported by the U.S. Department of Energy, Office of
  Science, Office of Advanced Scientific Computing Research, Department of
  Energy Computational Science Graduate Fellowship under award
  \#DE-SC0021110. This work was also partially supported by the National Science
  Foundation under awards \#1943130 and \#2219758, the Office of Naval
  Research under award \#N00014-23-1-2134, the Wolfson Foundation, and the Royal
  Society.
\end{acks}

\bibliographystyle{ACM-Reference-Format}
\bibliography{bib}

\ifshort
\else
\appendix
\section{The Category of Error Lenses}\label{app:check_bel}
This appendix provides the precise definition of the category \Bel{} of
backward error lenses. 

\begin{definition}[The Category \Bel{} of Backward Error Lenses]
\label{def:lensC} 
The category \Bel{} of \emph{backward error lenses} is the category with the
following data.
\begin{itemize}
  \item Its objects are slack distance spaces: 
    $(M, d : M \times M \to \R_{\geq 0}\cup\{\infty\}, r\in \R_{\geq 0}\cup\{\infty\})$, 
    where $d$ is a function giving distance and $r$ is known as the slack.
    \begin{itemize}
      \item We require that $d(x,x)\leq r$ for all $x\in X$
        (in other words, self-distance is bounded by slack).
    \end{itemize}
\item Its morphisms from $X$ to $Y$ are backward error lenses from $X$ to $Y$:
  triples of maps $(f, \tilde{f}, b)$, satisfying the two properties in
  \Cref{def:error_lens}.
\item The identity morphism on objects $X$ is given by the triple 
$(id_X,id_X,\pi_2)$.
\item The composition 
\[
(f_2,\tilde{f}_2,b_2) \circ (f_1,\tilde{f}_1,b_1)
\]
of error lenses ${(f_1,\tilde{f}_1,b_1) : X \rightarrow Y}$
and ${(f_2,\tilde{f}_2,b_2) : Y \rightarrow Z}$ is the error lens
${(f,\tilde{f},b): X \rightarrow Z}$ defined by
\begin{itemize}
\item the forward map 
  \begin{equation} \label{eq:fcomp}
    f: x \mapsto (f_1;f_2) \ x 
  \end{equation}
\item the approximation map
  \begin{equation}\label{eq:acomp}
    \tilde{f} : x \mapsto (\tilde{f}_1;\tilde{f}_2) \ x 
  \end{equation}
\item the backward map
  \begin{equation} \label{eq:bcomp}
    b: (x,z) \mapsto b_1(x, b_2(\tilde{f}_1(x), z)).
  \end{equation}
\end{itemize}
\end{itemize}
\end{definition}

Now, we verify that 
the composition is well-defined. 

\begin{figure}[h]
\[\begin{tikzcd}
	{X\times Y \times Z} && {X \times Y} \\
	\\
	{X \times Z} && X
	\arrow["{\langle id_X, \tilde{f}_{1} \rangle ~ \times ~ id_Z}", 
    from=3-1, to=1-1]
	\arrow["{id_X \times  ~b_2}", from=1-1, to=1-3]
	\arrow["{b}"', from=3-1, to=3-3]
	\arrow["{b_{1}}", from=1-3, to=3-3]
\end{tikzcd}\]
\caption{The backward map $b$ for the composition 
${{(f_2,\tilde{f}_2,b_2) \circ (f_1,\tilde{f}_1,b_1)}}$.}
\label{fig:compose}
\end{figure}
\noindent The diagram for the backward map for the composition of error lenses
is given in \Cref{fig:compose}.  

Let $L_1 = (f_1,\tilde{f}_1,b_1)$ and let $L_2 = (f_2,\tilde{f}_2,b_2)$.  We
first check the domain of $b$: for all $x \in X$ and $z \in Z$, and assuming
$d_Z(\tilde{f}(x),z)<\infty$, we must show  
\[
  {d_Y(\tilde{f}_1(x),b_2(\tilde{f}_1(x), z))}<\infty.
\] 
This follows from Property 1 for $L_2$:
\begin{align}
  d_Y(\tilde{f}_1(x),b_2(\tilde{f}_1(x), z))-r_Y &\le 
  d_Z(\tilde{f}_2(\tilde{f}_1(x)),z)-r_Z \\ 
  &= d_Z( \tilde{f}(x),z)-r_Z<\infty.
\end{align}
Now, given $x \in X$ and $z \in Z$ where
$d_Z(\tilde{f}(x),z)<\infty$, we can freely use 
Properties 1 and 2 of the lens $L_1$ to show that the 
lens properties hold for the composition:
\begin{enumerate}[align=left]
\item[{Property 1.}] 
  \begin{align*}
  d_X(x,b(x,z))-r_X &= 
      d_X(x,b_1(x, b_2(\tilde{f}_1(x), z)))-r_X 
      &&  \text{\cref{eq:bcomp}} \\
  &\le d_Y(\tilde{f}_1(x),b_2(\tilde{f}_1(x), z))-r_Y 
      && \text{Property 1 for  $L_1$} \\
  &\le d_Z(\tilde{f}_2(\tilde{f}_1(x)),z)-r_Z
      && \text{Property 1 for  $L_2$} \\
  &= d_Z(\tilde{f}(x),z)-r_Z. 
      && \text{\cref{eq:acomp}}
  \end{align*}
\item[{Property 2.}] 
\begin{align*}
f(b(x,z)) &= 
  f_2(f_1(b_1(x, b_2(\tilde{f}_1(x), z)) 
  )) && \text{\cref{eq:bcomp} \&  \cref{eq:fcomp}} \\
&= 
  f_2 ( b_2(\tilde{f}_1(x), z)) &&  
  \text{Property 2 for  $L_1$} \\
&=z. &&   \text{Property 2 for  $L_2$}
\end{align*}
\end{enumerate}

\section{Basic Constructions in \Bel{}}\label{app:check_bel_cons}
This appendix defines basic constructions in \Bel{} and 
verifies that they are well-defined.
Following the description of
\bea{} given in \Cref{sec:language}, we give the constructions below for lenses
corresponding to a tensor product, coproducts, and a \emph{graded comonad} for
interpreting linear typing contexts.

\subsection{Initial and Weak Terminal Objects}\label{app:terminal}
To warm up, let's consider the initial and weak terminal objects of our category.  Let
$0 \in \Bel$ be the empty distance space $(\emptyset, d_\emptyset, 0)$, and $1 \in
\Bel$ be the singleton distance space $(\{\star\}, \underbar{0}, 0)$ with a single
element and a constant distance function $d_1({\star,\star})=0$. Then for any
object $X \in \Bel$, there is a unique morphism $0_X : 0 \to X$ where the
forward, approximate, and backward maps are all the empty map, so $0$ is an
\emph{initial object} for \Bel.

Similarly, for every object $X \in \Bel$ there is a morphism $!_X : X \to 1$ given by
$f_{!} = \tilde{f}_!:= x \mapsto \star$ and $b_{!} := (x,\star) \mapsto x$. To
check that this is indeed a morphism in \Bel, we must check the two backward
error lens conditions in \Cref{def:error_lens}. The first condition boils down
to checking $d_X(x, x)-r_X \leq d_1(\star, \star)= 0$, but this holds by
our condition on objects in \Bel{} that $d_X(x,x)\leq r_X$ for all $x\in X$. 
The second condition is
clear, since there is only one element in $1$. However, this morphism is not 
necessarily unique, so $1$ is a \emph{weak terminal object} for \Bel.

\subsection{Tensor Product}
Next, we turn to products in \Bel. Like most lens categories, \Bel{} does not
support a Cartesian product. In particular, it is not possible to define a
diagonal morphism $\Delta_A : A \to A \times A$, where the space $A \times A$
consists of pairs of elements of $A$. The problem is the second lens condition
in \Cref{def:error_lens}: given an approximate map $\tilde{f} : A \to A \times
A$ and a backward map $b : A \times (A \times A) \to A$, we need to satisfy
\[
  f(b(a_0, (a_1, a_2))) = (a_1, a_2)
\]
for all $(a_0, a_1, a_2) \in A$. But it is not possible to satisfy this
condition when $a_1 \neq a_2$: that backward map can only return $a_0, a_1$, or
$a_2$, and in any case there is not enough information for the ideal map
$f$ to recover $(a_1, a_2)$. More conceptually, this is the technical
realization of the problem we described in \Cref{sec:overview}: if we think of
$a_1$ and $a_2$ as backward error witnesses for two subcomputations that both
use $A$, we may not be able to reconcile these two witnesses into a single
backward error witness.

Although a Cartesian product does not exist, \Bel{} does support a weaker,
monoidal product making it into a symmetric monoidal category.  Given two
objects $X$ and $Y$ in \Bel{} we have the object 
$X\otimes Y \triangleq (X \times Y,d_{X\otimes Y}, r_{X\otimes Y})$ where 
\vspace{0.5em}
\begin{equation}\label{eq:d_tensor_def}
d_{X\otimes Y}((x_1, y_1), (x_2, y_2))\triangleq
\begin{cases}
  \infty & \text{if }d_X(x_1,x_2)=\infty \\
   & \text{or }d_Y(y_1,y_2)=\infty \\
  d_X(x_1,x_2) & \text{if }r_Y=\infty \\
  d_Y(y_1,y_2) & \text{if }r_X=\infty \\
  \max\{d_X(x_1,x_2)+r_Y, d_Y(y_1,y_2) + r_X\} & \text{otherwise}
\end{cases}
\end{equation}
and 
\[
r_{X\otimes Y}\triangleq
\begin{cases}
  r_X & \text{if }r_Y=\infty \\
  r_Y & \text{if }r_X=\infty \\
  r_X+r_Y & \text{otherwise}.
\end{cases}
\] \\
We can see that $d_{X\otimes Y}$ is nonnegative 
and $d_{X\otimes Y}((x,y), (x,y))\leq r_{X\otimes Y}$ 
for all $(x,y)\in X\times Y$ by casework. 

We prove a useful property: for all $(x_1,y_1),(x_2,y_2)\in X\times Y$, 
\begin{equation}\label{eq:d_tensor}
  d_{X\otimes Y}((x_1,y_1), (x_2,y_2))-r_{X\otimes Y} = 
  \max\{d_X(x_1,x_2)-r_X, d_Y(y_1,y_2)-r_Y\}.
\end{equation}
\begin{proof}
  If either $d_X(x_1,x_2)=\infty$ or $d_Y(y_1,y_2)=\infty$, 
  the equality holds. Else, if $r_Y=\infty$, then 
  \begin{align*}
    d_{X\otimes Y}((x_1,y_1), (x_2,y_2))-r_{X\otimes Y}&=d_X(x_1,x_2)-r_X \\
    &=\max\{d_X(x_1,x_2)-r_X, d_Y(y_1,y_2)-r_Y\}.
  \end{align*}
  The equality similarly holds if $r_X=\infty$.
  If both $r_X$ and $r_Y$ are finite, then 
  \begin{align*}
    d_{X\otimes Y}((x_1,y_1), (x_2,y_2))-r_{X\otimes Y}&=  
    \max\{d_X(x_1,x_2)+r_Y, d_Y(y_1,y_2)+r_X\}-(r_X+r_Y) \\
    &=\max\{d_X(x_1,x_2)-r_X,d_Y(y_1,y_2)-r_Y\}.\qedhere
  \end{align*}
\end{proof}

Given any two morphisms ${(f,\tilde{f},b) : A \rightarrow X}$ and
${(g,\tilde{g},b'): B \rightarrow Y}$, we have the morphism
\[
  {(f,\tilde{f},b) \otimes  (g,\tilde{g},b') : 
  A \otimes  B \rightarrow X \otimes  Y}
\]
defined by
\begin{itemize}
\item the forward map
  \begin{equation}\label{eq:tensor_lens1}
     (a_1, a_2) \mapsto (f(a_1),g(a_2)) 
   \end{equation}
\item the approximation map
  \begin{equation}\label{eq:tensor_lens2}  
    (a_1,a_2) \mapsto (\tilde{f}(a_1),\tilde{g}(a_2)) 
  \end{equation}
\item the backward map
  \begin{equation} \label{eq:tensor_lens3}
    ((a_1, a_2),(x_1,x_2)) \mapsto (b(a_1,x_1),b'(a_2,x_2)) 
  \end{equation}
\end{itemize}

The tensor product given in
\Cref{eq:tensor_lens1,eq:tensor_lens2,eq:tensor_lens3} is only well-defined if
the domain of the backward map is well-defined, and if the error lens
properties hold. We check these properties below.

We first check the domain of $b_{\otimes}$: for all ${(a_1,a_2) \in A \times  B}$ and 
${(x_1,x_2) \in X \times  Y}$, we assume 
\begin{equation}\label{eq:prod_assum}
  d_{X\otimes Y}(\tilde{f}_{ \otimes }(a_1,a_2),(x_1,x_2))<\infty
\end{equation} 
and we are required to show  
\begin{equation}\label{eq:domten}
  d_{X}(\tilde{f}(a_1),x_1)<\infty \text{ and } 
  d_{Y}(\tilde{g}(a_2),x_2) <\infty.
\end{equation}
This follows immediately by \Cref{eq:tensor_lens2} and \Cref{eq:d_tensor_def}.

Given that \Cref{eq:domten} holds for all ${(a_1,a_2) \in A \times  B}$ 
and ${(x_1,x_2) \in X \times  Y}$ under the assumption given in 
\Cref{eq:prod_assum}, we can freely use Properties 1 and 2 of the lenses 
$(f,\tilde{f},b)$ and $(g,\tilde{g},b')$ to show that the lens properties 
hold for the product.
\begin{enumerate}[align=left]
  \item[Property 1.]
  \begin{align*}
  d_{A\otimes B}((a_1,a_2), b_{ \otimes } 
    ((a_1,a_2),(x_1,x_2)))-r_{A\otimes B}
  &=\max\{d_A(a_1,b(a_1,x_1))-r_A,d_B(a_2,b'(a_2,x_2))-r_B\}
  \tag*{\cref{eq:d_tensor} and \cref{eq:tensor_lens3}} \\
  &\le \max\{d_X(\tilde{f}(a_1),x_1)-r_X, d_Y(\tilde{g}(a_2),x_2)-r_Y\}
  \tag*{Property 1 of $(f,\tilde{f},b)$ and $(g,\tilde{g},b')$} \\
  &=d_{X\otimes Y}(\tilde{f}_\otimes(a_1,a_2), (x_1,x_2))-r_{X\otimes Y}.
  \tag*{\cref{eq:d_tensor} and \cref{eq:tensor_lens2}}
  \end{align*}
  \item[{Property 2.}] As above, the property follows directly from 
    Property 2 of the components $(f,\tilde{f},b)$ and $(g,\tilde{g},b')$.
\end{enumerate}

\begin{lemma}\label{lem:bifun}
	The tensor product operation on lenses induces a bifunctor on  \Bel.
\end{lemma}
\noindent The proof of \Cref{lem:bifun} requires checking conditions expressing
preservation of composition and identities:

\begin{proof}
The functoriality of the triple given in
\Cref{eq:tensor_lens1,eq:tensor_lens2,eq:tensor_lens3} follows by checking
conditions expressing preservation of composition and identities.
Specifically, for any error lenses 
$h : A \rightarrow B$,
$h': A' \rightarrow B'$,
$g : B \rightarrow C $, and 
$g': B' \rightarrow C'$, 
we must show  
\[
 (g \otimes g') \circ (h \otimes h') = (g\circ h) \otimes  (g' \circ h').
\]
We check the backward map: given any $(a_1,a_2) \in A \otimes A'$ and 
$(c_1,c_2) \in C \otimes C'$ we have 
\begin{align}
  b_{(g \otimes g') \circ (h \otimes h')}((a_1,a_2),(c_1,c_2)) &=
  b_{h \otimes h'}((a_1,a_2),b_{g \otimes g'}
  (\tilde{f}_{h \otimes h'}(a_1,a_2),(c_1,c_2)))\\
  &= (b_h(a_1,b_g(\tilde{f}_h(a_1),c_1)),b_{h'}(a_2,b_{g'}(\tilde{f}_{h'}(a_2),c_2)))\\
  &= b_{(g\circ h) \otimes  (g' \circ h')}((a_1,a_2),(c_1,c_2)).
\end{align}

Moreover, for any objects $X$ and $Y$ in $\Bel$, the identity lenses 
$id_X$ and $id_Y$ clearly satisfy 
\[ 
id_X \otimes  id_Y = id_{X \otimes  Y}. \qedhere
\]
\end{proof}

The bifunctor $\otimes $ on the category \Bel{} gives rise to a symmetric
monoidal category of error lenses. The unit object $I$ is defined to be the
object $(\{\star\}, \underbar{0}, \infty)$ with a single element and a
constant distance function $d_I({\star,\star})=0$ along with natural
isomorphisms for the associator ($\alpha_{X,Y,Z} : X \otimes  (Y \otimes Z)\to (X\otimes Y)\otimes Z)$,
and we can define the usual left-unitor ($\lambda_X : I \otimes  X \rightarrow
X)$, right-unitor ($\rho_X : X \otimes I \rightarrow X )$, and symmetry
($\gamma_{X,Y} : X \otimes   Y \rightarrow Y \otimes  X$) maps.  

\subsubsection{{Associator}} 
We define the associator 
$\alpha_{X,Y,Z} : X \otimes  (Y  \otimes  Z) 
\rightarrow (X \otimes  Y) \otimes  Z$ as the following triple:
\begin{align}
  f_{\alpha}(x,(y,z)) &\triangleq ((x,y),z)  \\
  \tilde{f}_{\alpha}(x,(y,z)) &\triangleq ((x,y),z) 
      \label{eq:aassoc} \\
  b_{\alpha}((x,(y,z)),((a,b),c)) &\triangleq (a,(b,c)) 
      \label{eq:bassoc}.
\end{align}

It is straightforward to check that $\alpha_{X,Y,Z}$ is an error lens
satisfying Properties 1 and 2. To check that the associator is an isomorphism,
we are required to show the existence of the lens
\[
  \alpha' :  (X \otimes Y) \otimes Z \rightarrow 
  X \otimes  (Y \otimes Z) 
\] 
satisfying 
\[
  \alpha' \circ \alpha_{X,Y,Z} = id_{X\otimes (Y\otimes Z)}
\]
and 
\[
  \alpha_{X,Y,Z} \circ \alpha' = id_{(X\otimes Y)\otimes Z}
\]
where $id$ is the identity lens (see \Cref{def:lensC}).  Defining the forward
and approximation maps for $\alpha'$ is straightforward; for the forward map we
have
\begin{align*} 
  f_{\alpha'}((x,y),z) &\triangleq (x,(y,z))
\end{align*}
and the approximation map is defined identically. For the backward map we have
\begin{align*} 
  b_{\alpha'}(((x,y),z),(a,(b,c))) \triangleq ((a,b),c).
\end{align*}
It is straightforward to check that $\alpha'$ satisfies Properties 1 and 2 of 
an error lens. 

The naturality of the associator follows by checking that the following diagram
commutes. 
\[\begin{tikzcd}
	{X \otimes  (Y \otimes  Z)} && {X \otimes  (Y \otimes  Z)} \\ \\
	{(X \otimes  Y) \otimes  Z} && {(X \otimes  Y) \otimes  Z}
	\arrow["{\alpha_{X,Y,Z}}"', from=1-1, to=3-1]
	\arrow["{g_X \otimes   (g_Y  \otimes   g_Z)}", from=1-1, to=1-3]
	\arrow["{\alpha_{X,Y,Z}}", from=1-3, to=3-3]
	\arrow["{(g_X  \otimes   g_Y)  \otimes   g_Z}"', from=3-1, to=3-3]
\end{tikzcd}\]
That is, we check that 
\[
  ((g_X \otimes  g_Y) \otimes   g_Z) \circ \alpha_{X,Y,Z} = 
  \alpha_{X,Y,Z} \circ (g_X \otimes (g_Y \otimes  g_Z))
\] 
for the error lenses 
\begin{align*}
  g_X&: X \rightarrow X \triangleq (f_X,\tilde{f}_X,b_X) \\ 
  g_Y&: Y \rightarrow Y \triangleq (f_Y,\tilde{f}_Y,b_Y) \\
  g_Z&: Z \rightarrow Z \triangleq (f_Z,\tilde{f}_Z,b_Z).
\end{align*}
This follows from the definitions of lens composition (\Cref{def:error_lens})
and the tensor product on lenses
(\cref{eq:tensor_lens1,eq:tensor_lens2,eq:tensor_lens3}). 
We detail here the case of the backward map. 

Using the notation $b_g$ (resp. $\tilde{f}_g$) to refer to both of the 
backward maps (resp. approximation maps) of the tensor product lenses of 
the lenses $g_X$, $g_Y$, and $g_Z$, we are required to show that 
\begin{equation}
b_\alpha(xyz,b_{g}(\tilde{f}_{\alpha}(xyz),x'y'z')) 
  = b_{g}(xyz,b_\alpha(\tilde{f}_g(xyz),x'y'z')) 
\label{eq:verif_assoc1}
\end{equation}
for any $xyz \in X \ \times \ (Y\times  Z)$ and 
$x'y'z' \in (X\times  Y)\times  Z$:
\begin{align*}
  b_\alpha(xyz,b_{g}(\tilde{f}_{\alpha}(xyz),x'y'z'))
  &= b_{\alpha}(xyz,b_{g}(((x,y),z),x'y'z')) 
    & \text{\Cref{eq:aassoc}} \\
  &= b_{\alpha}(xyz,((b_X(x,x'),b_Y(y,y')),b_Z(z,z'))) 
    & \text{\Cref{eq:tensor_lens3}} \\
  &= (b_X(x,x'),(b_Y(y,y'),b_Z(z,z'))) 
    & \text{\Cref{eq:bassoc}} \\
  &= b_{g}(xyz,(x',(y',z'))) 
    & \text{\Cref{eq:tensor_lens3}} \\
  &= b_{g}(xyz,b_\alpha(\tilde{f}_g(xyz),x'y'z')).
    & \text{\Cref{eq:bassoc}} \\
\end{align*}

\subsubsection{{Unitors}} We define the left-unitor 
$\lambda_X: I \otimes  X \rightarrow X$ as
\begin{align*}
  f_{\lambda}(\star,x) &\triangleq x \\
  \tilde{f}_{\lambda}(\star,x) &\triangleq x \\
  b_{\lambda}((\star,x),x') &\triangleq (\star,x').
\end{align*}
The right-unitor is similarly defined. 
The fact that $r_I = \infty$ is essential in order for 
$\lambda_X$ to satisfy the first property of an error lens.
Assuming $d_X(\tilde{f}(\star, x), x')=d_X(x,x')<\infty$,
\begin{align*}
  d_{I\otimes X}((\star,x),b_\lambda((\star,x), x')) - r_{I\otimes X}
  &= d_{I\otimes X}((\star,x),(\star,x')) - r_{I\otimes X} \\
  &= \max\{d_I(\star, \star)-r_I, d_X(x,x')-r_X\} 
    &\text{\Cref{eq:d_tensor}} \\
  &= \max\{-\infty, d_X(x,x')-r_X\} \\
  &= d_X(x,x') - r_X \\
  &= d_X(\tilde{f}_\lambda(\star, x), x')-r_X.
\end{align*}

Checking the naturality of $\lambda_X$ amounts to checking that 
the following diagram commutes for all error lenses 
$g : X \rightarrow Y$. 
\[\begin{tikzcd}
	{I \otimes  X} && {I \otimes  Y} \\ \\
	{X} && {Y}
	\arrow["{\lambda_X}"', from=1-1, to=3-1]
	\arrow["{ id_I \otimes  g }", from=1-1, to=1-3]
	\arrow["{\lambda_Y}", from=1-3, to=3-3]
	\arrow["{g}"', from=3-1, to=3-3]
\end{tikzcd}\]

\subsubsection{{Symmetry}}
We define the symmetry map $\gamma_{X,Y} : X \otimes   Y \rightarrow Y \otimes
X$ as the following triple:
\begin{align*}
  f_{\gamma}(x,y) &\triangleq (y,x) \\
  \tilde{f}_{\gamma}(x,y) &\triangleq (y,x) \\
  b_{\gamma}((x,y),(y',x')) &\triangleq (x',y').
\end{align*} 

It is straightforward to check that  $\gamma_{X,Y}$ is an error 
lens. Checking the naturality of  $\gamma_{X,Y}$ amounts to 
checking that the following diagram commutes for any error 
lenses $g_1 : X \rightarrow Y$ and $g_2 : Y \rightarrow X$.

\[\begin{tikzcd}
	{X \otimes  Y} && {Y \otimes  X} \\ \\
	{Y \otimes  X} && {X \otimes  Y}
	\arrow["{\gamma_{X,Y}}"', from=1-1, to=3-1]
	\arrow["{ g_1 \otimes  g_2 }", from=1-1, to=1-3]
	\arrow["{\gamma_{Y,X}}", from=1-3, to=3-3]
	\arrow["{g_2 \otimes  g_1}"', from=3-1, to=3-3]
\end{tikzcd}\]

\subsection{Projections}
For any two spaces $X$ and $Y$ with zero self-distance and $r_X=r_Y$, 
we can define a projection map
$\pi_1 : X \otimes Y \rightarrow X$ via:
\begin{itemize}
\item the forward map 
  \[
    f: (x, y) \mapsto x
  \]
\item the approximation map
  \[
    \tilde{f} : (x, y) \mapsto x
  \]
\item the backward map
  \[
    b: ((x, y),z) \mapsto (z, y).
  \]
\end{itemize}
The projection $\pi_2 : X \otimes Y \rightarrow Y$ is defined similarly.
We prove the first lens property for $\pi_1$.
Suppose $d_X(\tilde{f}(x,y),z)=d_X(x,z)<\infty$:
\begin{align*}
  d_{X\otimes Y}((x,y),b((x,y),z)) -r_{X\otimes Y} 
  &= \max\{d_X(x,z)-r_X,d_Y(y,y)-r_Y\} 
    &\text{\Cref{eq:d_tensor}} \\
  &= \max\{d_X(x,z)-r_X, -r_Y\} 
    &\text{since $d_Y(y,y)=0$} \\
  &= d_X(x,z)-r_X. 
    &\text{since $r_X=r_Y$}
\end{align*}

\subsection{Coproducts}\label{app:coproducts}
For any two spaces $X$ and $Y$ such that
$r_X$ and $r_Y$ are finite, we have the object $(X + Y,d_{X+Y},r_{X+Y})$,
where the metric $d_{X+Y}$ is defined as 
\begin{equation} 
  d_{X+Y}(x,y) \triangleq   
   \begin{cases} 
      d_X(x_0,y_0) +r_Y & \text{if } x = inl \ x_0  \text{ and } y = 
      inl \ y_0  \\
      d_Y(x_0,y_0) +r_X & \text{if } x = inr \ x_0  \text{ and } y = 
      inr \ y_0 \\
     \infty & \text{otherwise}
    \end{cases}
\label{eq:prod_met}
\end{equation}
and $r_{X+Y}\triangleq r_X+r_Y$. We can easily show that self-distance is bounded by 
$r_{X+Y}$.

We define the morphism for the first projection $in_1: X \rightarrow X + Y$ as
the triple
\begin{align}
  f_{in_1}(x) = \tilde{f}_{in_1}(x) &\triangleq inl \ x \label{eq:inlf} \\
  b_{in_1}(x,z) &\triangleq 
      \begin{cases}  
        x_0 & \text{if } z = inl \ x_0 \\
        x & \text{otherwise.}
      \end{cases} \label{eq:inlb}
\end{align}

The morphism $in_2 : Y \rightarrow X + Y $ for the second projection can be
defined similarly. We now check that the first projection is well-defined:
\begin{enumerate}[align=left]
\item[Property 1.] For any $x \in X$ and $z \in X+Y$, suppose
\[
  d_{X+Y}(\tilde{f}_{in_1}(x),z) = d_{X+Y}(inl \ x,z)<\infty.
\] 
Thus, we know $z=inl \ x_0$ for some $x_0\in X$, and so 
\begin{align*}
  d_{X}(x,b_{in_1}(x,z))-r_X &= d_X(x,x_0)-r_X \\
  &= d_X(x,x_0) + r_Y-r_{X+Y} \\
  &= d_{X+Y}(\tilde{f}_{in_1} (x), z)-r_{X+Y}. &\text{\Cref{eq:prod_met}}
\end{align*}

\item[Property 2.] For any $x \in X$ and $z \in X+Y$, again suppose 
$d_{X+Y}(\tilde{f}_{in_1}(x),z) = d_{X+Y}(inl \ x,z)<\infty$, and thus
$z=inl\ x_0$ for some $x_0\in X$. 
Thus,
\[
  f_{in_1}(b_{in_1}(x,z))=f_{in_1}(x_0)=inl\ x_0=z
\]
by \Cref{eq:inlf}.
\end{enumerate}

Given any two morphisms 
${g : X \rightarrow C \triangleq (f_g,\tilde{f}_g,b_g)}$ and 
${h : Y \rightarrow C \triangleq (f_h,\tilde{f}_h,b_h)}$ 
where $r_X$ and $r_Y$ are finite,
we can define the unique copairing morphism ${[g,h] : X + Y \rightarrow C}$ 
such that $[g,h] \circ in_1 = g$ and $[g,h]  \circ in_2= h$:
\begin{align}
  f_{[g,h]}(z) &\triangleq       
      \begin{cases}  
        f_g(x) & \text{if } z = inl \ x \\
        f_h(y) & \text{if } z = inr \ y
      \end{cases}  \\
  \tilde{f}_{[g,h]}(z) &\triangleq 
      \begin{cases}  
        \tilde{f}_g(x) & \text{if } z = inl \ x \\
        \tilde{f}_h(y) & \text{if } z = inr \ y
      \end{cases} \label{eq:approx_copairing}\\
  b_{[g,h]}(z,c) &\triangleq 
      \begin{cases}  
        inl \ (b_g(x,c)) & \text{if } z = inl \ x \\
        inr \ (b_h(y,c)) & \text{if } z = inr \ y.
      \end{cases} 
      \label{eq:b_copairing}
\end{align}
We now check that copairing is well-defined:

\begin{enumerate}[align=left]
\item[Property 1.] For some $z \in X  + Y$ and $c \in C$, 
suppose $d_C(\tilde{f}_{[g,h]}(z),c)<\infty$.
If $z=inl\ x$ for some $x\in X$, then $d_C(\tilde{f}_g(x),c)<\infty$ and 
we can use Property 1 for $g$:
\begin{align*}
  d_{X+Y}(z,b_{[g,h]}(z,c)) -r_{X+Y} 
  &= d_X(x, b_g(x,c)) -r_X &\text{\Cref{eq:prod_met} and \Cref{eq:b_copairing}} \\
  &\leq d_C(\tilde{f}_g(x),c)-r_C &\text{Property 1 for $g$} \\
  &= d_C(\tilde{f}_{[g,h]}(z), c)-r_C. &\text{\Cref{eq:approx_copairing}}
\end{align*}
Otherwise, $z=inr \ y$ for some $y \in Y$ then $d_C(\tilde{f}_{h}(y), c)<\infty$
 and we use Property 1 for $h$.

\item[Property 2.] For some $z \in  X+ Y$ and $c \in C$, 
$f_{[g,h]}(b_{[g,h]}(z,c)) = c$
supposing $d_C (\tilde{f}_{[g,h]}(z),c)<\infty$.
If $z=inl \ x$ for some $x \in X$ then 
$d_C(\tilde{f}_{g}(x), c)<\infty$ 
and we use Property 2 for $g$. Otherwise, $z=inr \ y$ for some $y \in Y$ 
then $d_C(\tilde{f}_{h}(y), c)<\infty$ and we use 
Property 2 for $h$.
\end{enumerate}

To show that $[g,h] \circ in_1 = g$ (resp. $[g,h]  \circ in_2= h$), 
we observe that the following diagrams, by definition, commute.

\begin{figure}[!htb]
\minipage{0.32\textwidth}
\[\begin{tikzcd}
	{X + Y} && C \\
	\\
	X
	\arrow["{f_{in_1}}", from=3-1, to=1-1]
	\arrow["{f_{[g,h]}}", from=1-1, to=1-3]
	\arrow["{f_g}"', from=3-1, to=1-3]
\end{tikzcd}\]
\endminipage\hfill
\minipage{0.32\textwidth}
\[\begin{tikzcd}
	{X + Y} && C \\
	\\
	X
	\arrow["{\tilde{f}_{in_1}}", from=3-1, to=1-1]
	\arrow["{\tilde{f}_{[g,h]}}", from=1-1, to=1-3]
	\arrow["{\tilde{f}_g}"', from=3-1, to=1-3]
\end{tikzcd}\]
\endminipage\hfill
\minipage{0.32\textwidth}%
\[\begin{tikzcd}
	{(X+Y)\times C} && {X + Y} \\
	\\
	{X \times C} && X
	\arrow["{\tilde{f}_{in_1}\times id_C}", from=3-1, to=1-1]
	\arrow["{b_{[g,h]}}", from=1-1, to=1-3]
	\arrow["{b_g}"', from=3-1, to=3-3]
	\arrow["{b_{in_1}}", from=1-3, to=3-3]
\end{tikzcd}\]
\endminipage
\end{figure}

\subsubsection{Uniqueness of the copairing}
We check the uniqueness of copairing by showing that for any two morphisms 
$g_1 : X \rightarrow C$ and $g_2 : Y \rightarrow C$, if 
$h \circ in_1 = g_1$ and $h \circ in_2 = g_2$ for any 
$h : X + Y \rightarrow C$, then $h = [g_1,g_2]$. 

We detail the cases for the forward and backward map; the case for the 
approximation map is identical to that of the forward map. 

\begin{enumerate}[align=left]
\item[\textbf{forward map}] We are required to show that 
$f_h(z) = f_{[g_1,g_2]}(z)$ for any ${z \in X + Y}$ assuming that 
$f_{in_1};f_{h} = f_{g_1}$ and $f_{in_2};f_{h} = f_{g_2}$. 
The desired conclusion follows by cases on $z$; i.e., $z = inl \ x$ 
for some $x \in X$ or $z = inr \ y$ for some $y \in Y$.

\item[\textbf{backward map}] We are required to show that 
$b_{h}(z,c) = b_{[g_1,g_2]}(z,c)$ for any $z \in X+Y$ and $c \in C$
such that 
\[d_C(\tilde{f}_h(z),c)=d_C(\tilde{f}_{[g1,g2]}(z),c)<\infty.\]
If $z = inl \ x$ for some $x \in X$ then we are required to show that 
\begin{align*}
  b_{h}(inl \ x,c) &= inl \ (b_{g_1}(x,c)). & \text{\Cref{eq:b_copairing}}
\end{align*}

By definition of lens composition, since $d_C(\tilde{f}_{h}(z),c)=d_C(\tilde{f}_{h\circ in_1}(x),c)<\infty$, 
we have that 
\[d_{X+Y}( inl \ x, b_h(inl \ x, c)) <\infty.\]
Thus, $b_h(inl \ x, c) = inl \ x_0$ for some $x_0 \in X$.
Unfolding definitions in the assumption $b_{h \circ in_1} = b_{g_1}$, 
we have that 
$b_{in_1}(x,b_h( \tilde{f}_{in_1}(x),c)) = b_{g_1}(x,c)$.
Hence, $b_{g_1}(x,c) = b_{in_1}(x,inl \ x_0) = x_0$, from which 
the desired conclusion follows. 
The case of $z = inr \ y$ for some $y \in Y$ is analogous.
\end{enumerate}

\subsection{A Graded Comonad}\label{app:comonad}
Next, we turn to the key construction in \Bel{} that enables our semantics to
capture morphisms with non-zero backward error. The rough idea is to use a
graded comonad to shift the distance by a numeric constant; this change then
introduces slack into the lens conditions in \Cref{def:error_lens} to support
backward error.

More precisely, construct a comonad graded by the real numbers.  Let the
pre-ordered monoid $\mathcal{R}$ be the non-negative real numbers ${\R^{\ge 0}}$
with the usual order and addition.  We define a graded comonad on \Bel{} by the
family  of functors 
\begin{align*}
\{D_r :
  \Bel\rightarrow \Bel \ | \ r \in \mathcal{R}\}
\end{align*}
as follows.
\begin{itemize}
\item The object-map $D_r: \mathbf{Bel} \rightarrow \mathbf{Bel}$ takes
  $(X,d_X,r_X)$ to $(X,d_X, r_X+r)$, where $\infty + r$ is defined to be $\infty$.
\item The arrow-map $D_r$ takes an error lens 
$(f,\tilde{f},b) : A \rightarrow X$ to an error lens 
\begin{align}
	(D_rf,D_r\tilde{f},D_rb) : D_rA \rightarrow D_rX
\end{align}
 where
\begin{align}\label{eq:monad1}
	(D_rg)x \triangleq g(x). 
\end{align}
\item The \emph{counit} map $\varepsilon_X: D_0 X \rightarrow X$ is 
the identity lens.
\item The \emph{comultiplication} map 
$\delta_{q,r,X}: D_{q+r}X \rightarrow D_q(D_rX)$ is the identity lens.
\item The \emph{2-monoidality} map 
$m_{r,X,Y} : D_rX \otimes D_rY \iso D_r(X \otimes Y)$ is the identity
lens. 
\item The map
$m_{q\le r,X} : D_rX \rightarrow D_qX$ is the identity lens.
\end{itemize}
Unlike similar graded comonads considered in the literature on coeffect
systems, our graded comonad does not support a graded contraction map: there is
no lens morphism $c_{r, s, X} : D_{r + s} X \to D_r X \otimes D_s X$. This is
for the same reason that our category does not support diagonal maps: it is not
possible to satisfy the second lens condition in \Cref{def:error_lens}. Thus,
we have a graded comonad, rather than a graded exponential comonad
\citep{Brunel:2014:coeffects}.

\subsection{Discrete Objects}
While there is no morphism $X \to X \otimes X$ in general, graded or not, there
is a special class of objects where we do have contraction and thus a diagonal map: the
\emph{discrete} spaces.

\begin{definition}[Discrete space]
  We say a slack distance space $(X, d_X, r)$ is \emph{discrete} if its
  distance function $d_X(x_1, x_2)=\infty$ for all $x_1 \neq x_2$, and 
  $d_X(x,x)=0$ for all $x$. 
\end{definition}
We write \Del{} for category of the discrete spaces and backward error lenses;
this forms a full subcategory of \Bel.  Discrete objects are closed under the
monoidal product in \Bel, and the unit object $I$ is discrete. We define a monad
which promotes some objects in \Bel{} to discrete objects.

\begin{definition}[Discrete monad]
  We define a monad, $M$, on the subcategory of \Bel{} of objects $(X,d_X,0)$ 
  whose metrics satisfy \emph{reflexivity}, i.e. $d_X(x,y)=0$ if and only
  if $x=y$. The monad acts on objects as follows:
  \begin{equation*}
    M(X,d_X,0)\triangleq (X,d_\alpha,0)
  \end{equation*}
  where $d_\alpha$ is the discrete metric on $X$. On morphisms, 
  $M(f,\tilde{f},b)\triangleq (f,\tilde{f},b)$.
  The unit $(\eta)$ and multiplication $(\mu)$ maps are the identity.
\end{definition}
We verify for all morphisms $(f,\tilde{f},b):(X,d_X,0)\to (Y,d_Y,0)$, where $d_X$ and $d_Y$
satisfy reflexivity, that $M(f,\tilde{f},b):(X,d_\alpha,0)\to (Y,d_\alpha,0)$ is a valid lens.
\begin{proof}
  Take $x\in X,y\in Y$ such that $d_\alpha(\tilde{f}(x),y)<\infty$. As $d_\alpha$ 
  is the discrete metric and thus infinite on distinct objects, we must have 
  $\tilde{f}(x)=y$. As $d_Y$ satisfies reflexivity by assumption, we 
  know $d_Y(\tilde{f}(x),y)=0$. Hence, the first lens condition for 
  $(f,\tilde{f},b)$ asserts:
  \begin{equation*}
    d_X(x,b(x,y))\leq d_Y(\tilde{f}(x),y)=0.
  \end{equation*}
  By reflexivity, as $d_X(x,b(x,y))=0$, we know $x=b(x,y)$. Now, we can prove
  the first lens condition for $M(f,\tilde{f},b)$:
  \begin{equation*}
    d_\alpha(x,b(x,y))=d_\alpha(\tilde{f}(x),y)=0.
  \end{equation*}
  The second lens condition follows immediately from the second lens condition 
  on $(f,\tilde{f},b)$.
\end{proof}

There are two other key facts about discrete objects. 
First, the graded comonad $D_r$ restricts to a graded comonad on $\Del$. In
particular, if $X$ is a discrete object, then $D_r X$ is also a discrete object.
Second, it is possible to define a contraction map.
\begin{lemma}[Discrete contraction]
  For any discrete object $X \in \Del$ and $r,s\in\mathcal{R}$, 
  there is a lens morphism $c_{r,s,X} : D_{r+s}X\to D_rX\otimes D_sX$ defined via
  \begin{align*}
    f_t = \tilde{f}_t &\triangleq x \mapsto (x,x) \\
    b_t &\triangleq (x,(x_1,x_2)) \mapsto x.
  \end{align*}
\end{lemma}
\begin{proof}
  For $x,x_1,x_2\in X$, suppose $d_{X\otimes X}((x,x),(x_1,x_2))<\infty$.
  By \Cref{eq:d_tensor_def}, we have $d_X(x,x_1)<\infty$ and $d_X(x,x_2)<\infty$;
  since $d_X$ is a discrete metric, we know $x=x_1=x_2$. Now, using the fact that
  self-distance is bounded by slack,
  \begin{enumerate}[align=left]
    \item[Property 1.] \begin{align*}
      d_X(x,b(x,(x_1,x_2)))-r_{D_{r+s}X} &=d_X(x,x)-r_X -(r+s) \\
      &\leq -(r+s) \\
      &\leq \max\{-r,-s\} \\
      &\leq\max\{d_X(x,x_1)-(r_X+r),d_X(x,x_2)-(r_X+s)\} \\
      &=d_{X\otimes X}((x,x),(x_1,x_2))-r_{D_rX\otimes D_sX}.
    \end{align*}

    \item[Property 2.] \begin{align*}
      f_t(b_t(x,(x_1,x_2)))= (x,x)=(x_1,x_2).
    \end{align*}\qedhere
  \end{enumerate}
\end{proof}

More conceptually, we can think of a discrete object as a space that can't have any
backward error pushed onto it. Thus, the backward error witnesses $x_1$ and
$x_2$ are always equal to the input, and can always be reconciled.

Using the contraction map, we can define the diagonal lens $t_X$ and another useful map $t_X^r$, 
so long as $X$ is a discrete object:
\begin{align}
  t_X:X\to X\otimes X &\triangleq D_0;\ c_{0,0,X};\ \varepsilon_X\otimes \varepsilon_X \label{eq:discr-diag} \\
  t^r_X:X\to D_rX\otimes X &\triangleq D_r;\ c_{r,0,X};\ id\otimes \varepsilon_X. \label{eq:disc-diag-2}
\end{align}

\section{Interpreting \bea{} Terms}\label{app:interp_bea}
In this section of the appendix, we detail the constructions for 
interpreting \bea{} terms. 
Applications of the associativity map, symmetry map
$\gamma_{X,Y} : X \otimes Y \rightarrow Y \otimes X$,
and 2-monoidality map $m_{r,A,B} : D_rA \otimes D_rB \iso D_r(A \otimes B)$ are omitted for clarity 
of the presentation. 

Given a type $\tau$, we define a metric space $\denot{\tau}$ with the rules
\begin{gather*}
	\denot{m(\sigma)} \triangleq M\denot{\sigma} \qquad
	\denot{\num}  \triangleq (\R,d_\R, 0) \qquad
	\denot{\sigma \otimes  \tau} \triangleq 
	\denot{\sigma} \otimes \denot{\tau} \\
	\denot{\sigma + \tau} \triangleq \denot{\sigma} + \denot{\tau} \qquad
  \denot{\textbf{unit}} \triangleq 1 = (\{\star\}, \underbar{$0$}, 0)
\end{gather*}
where the distance function $d_\alpha$ is the discrete metric on $\R$, and the
distance function $d_\R$ is defined following the error arithmetic proposed by
~\citet{Olver:78:error}, as given in \Cref{{eq:olver}}. By definition, if $d_\R$
is a standard distance function, then $\denot{\tau}$ is a standard metric space. 
Moreover, we can apply the discrete monad $M$ to any type interpretation $\denot\tau$ 
as it will satisfy positive definiteness and slack zero.

The interpretation of typing contexts is defined inductively as follows:
\begin{alignat*}{3}
  \denot{\emptyset \mid \emptyset}  &\triangleq I = (\{ \star \},
  \underbar{$0$}, \infty) & \qquad
  \denot{\emptyset \mid  \Gamma, x:_r\sigma}  &\triangleq 
  \denot{\Gamma} \otimes D_r\denot{\sigma} \\
  \denot{\Phi, z : \alpha \mid \emptyset}  &\triangleq \denot{\Phi} \otimes 
  \denot{\alpha} & \qquad
  \denot{\Phi, z : \alpha \mid  \Gamma, x:_r\sigma}  &\triangleq 
  \denot{\Phi} \otimes \denot{\alpha} \otimes \denot{\Gamma} \otimes
  D_r\denot{\sigma}
\end{alignat*}
where the graded comonad $D_r$ is used to interpret the linear variable
assignment $x:_r \sigma$.

Given the above interpretations of types and typing environments, we can
interpret each well-typed \bea{} term
$\Phi \mid \Gamma \vdash e : \tau$ as an
error lens $\denot{e} : \denot{\Phi} \otimes \denot{\Gamma} \rightarrow
\denot{\tau}$ in \Bel, by structural induction on the typing derivation.

\begin{description}
% VAR LINEAR
\item[Case (Var).] Suppose that $\Gamma = x_{0} :_{q_0} \sigma_0, \dots, x_{i-1}
  :_{q_{i-1}} \sigma_{i-1}$ and $\Phi$ has length $j$. Define the map   
$\denot{\Phi\mid \Gamma, x:_r \sigma \vdash x:\sigma}$ 
as the composition
\[
  \pi_{i+j}\circ 
  (id_{\denot{\Phi}} \otimes \varepsilon_{\denot{\Gamma}} \otimes 
    \varepsilon_{\denot{\sigma}}) \circ
  (id_{\denot{\Phi}} \otimes
    m_{0\leq q_0,\denot{\sigma_0}}\otimes\dots\otimes
    m_{0\leq q_{i-1},\denot{\sigma_{i-1}}} \otimes
    m_{0\leq r,\denot{\sigma}})
\]
where the lens $\pi_{i+j}$ is the $(i+j)$th projection, 
and we can apply it since all types $\sigma$ 
are interpreted as metric spaces, i.e., satisfying self-distance zero
and with zero slack.

% VAR DISCRETE
\item[Case (DVar).] Define the map   
$\denot{\Phi,z : \alpha \mid\Gamma \vdash z:\alpha}$ as 
\[
  \pi_i\circ (id_{\denot\Phi\otimes \denot\alpha}\otimes \varepsilon_{\denot\Gamma})
  \circ (id_{\denot\Phi\otimes \denot\alpha} 
    \otimes \overline{m_{0\leq q_j,\denot{\sigma_j}}})
\]
where the map $\overline{m_{0 \leq q_j, \denot{\sigma_j}}}$ applies the map $m_{0
\leq q, X} : D_q X \to D_0 X$ to each binding $\denot{x :_q \sigma}$ in the
context $\denot{\Gamma}$.

% UNIT
\item[Case (Unit).] Define the map 
$
\denot{\Phi\mid \Gamma \vdash ( ) : \unit }
$
as the lens given in \Cref{app:terminal} from a tuple 
$\bar{x} \in \denot{\Phi\mid \Gamma}$ to the weak terminal object
 $1 = (\{\star\},\underline{0}, 0)$.

% TENS INTRO
\item[Case ($\otimes $ I).] 
Given the maps  
\begin{align*}
  h_1 &= \denot{\Phi\mid \Gamma \vdash e : \sigma}: \denot{\Phi}\otimes \denot{\Gamma}
  \rightarrow \denot{\sigma}\\ 
  h_2 &=\denot{\Phi\mid \Delta \vdash f : \tau} : \denot{\Phi}\otimes \denot{\Delta}
  \rightarrow \denot{\tau}
\end{align*}
define the map 
$\denot{\Phi\mid \Gamma,\Delta \vdash (e,f) : \sigma \otimes \tau}$
as the composition
\[ 
  (h_1 \otimes h_2) \circ (t_{\denot{\Phi}} \otimes id_{\denot{\Gamma,\Delta}}),
\]
where the map $t_{\denot{\Phi}}: \denot{\Phi} 
  \rightarrow \denot{\Phi} \otimes \denot{\Phi}$ is 
  the diagonal lens on discrete metric spaces (\Cref{eq:discr-diag}).
% TENS ELIM LINEAR
\item[Case ($\otimes $ E$_\sigma$).] Given the maps
\begin{align}
h_1 &= \denot{\Phi\mid \Gamma \vdash e : \tau_1 \otimes \tau_2} : 
  \denot{\Phi} \otimes \denot{\Gamma} \rightarrow {\denot{\tau_1}} \otimes\denot{\tau_2} \\
h_2 &= \denot{\Phi\mid \Delta, x :_r \tau_1, y :_r \tau_2 \vdash 
  f : \sigma} : \denot{\Phi} \otimes \denot{\Delta} \otimes 
  D_r\denot{\tau_1} \otimes D_r\denot{\tau_2} \rightarrow {\denot{\sigma}}
\end{align}
we define the map
$\denot{\Phi\mid r + \Gamma, \Delta \vdash \slet {(x,y)} e f : \sigma}$ as the composition
\[
  h_2 \circ (
  (m^{-1}_{r,\denot{\tau_1},\denot{\tau_2}} \circ 
  D_r(h_1)) \otimes id_{\denot{\Phi}\otimes\denot{\Delta}})
 \circ (t^r_{\denot{\Phi}} \otimes id_{D_r\denot{\Gamma} \otimes  \denot{\Delta}})
\]
where $t^r_{\denot{\Phi}}:\denot{\Phi}\rightarrow D_r\denot{\Phi} \otimes \denot{\Phi}$
is a variation of the diagonal lens on discrete metric spaces (\Cref{eq:disc-diag-2}).

% TENS ELIM DISCRETE
\item[Case ($\otimes $ E$_\alpha$).] Given the maps
\begin{align}
h_1 &= \denot{\Phi\mid \Gamma \vdash e : \alpha_1 \otimes \alpha_2} : 
  \denot{\Phi} \otimes \denot{\Gamma} \rightarrow 
  {\denot{\alpha_1}}\otimes\denot{\alpha_2} \\
h_2 &= \denot{\Phi, x : \alpha_1, y : \alpha_2; \Delta \vdash 
  f : \sigma} : \denot{\Phi} \otimes \denot{\alpha_1} \otimes \denot{\alpha_2} 
    \otimes \denot{\Delta} \rightarrow {\denot{\sigma}}
\end{align}
define the map 
$\denot{\Phi\mid \Gamma, \Delta \vdash \dlet {(x,y)} e f : \sigma}$
as the composition
\[
h_2 \circ (h_1 \otimes id_{\denot{\Phi} \otimes  \denot{\Delta}})
 \circ (t_{\denot{\Phi}} \otimes id_{\denot{\Gamma} \otimes  \denot{\Delta}}).
\]
% SUM ELIM
\item[Case ($+$ E).] Given the maps
\begin{align*}
h_1 &= \denot{\Phi\mid  \Gamma \vdash e' : \sigma + \tau} : 
  \denot{\Phi} \otimes \denot{\Gamma} \rightarrow 
  \denot{\sigma + \tau}\\
h_2 &= \denot{\Phi\mid \Delta,x :_q \sigma \vdash e : \rho} : 
  \denot{\Phi} \otimes \denot{\Delta} \otimes D_q\denot{\sigma} \rightarrow 
  \denot{\rho}\\
h_3 &= \denot{\Phi\mid \Delta, y :_q \tau \vdash f : \rho} : 
  \denot{\Phi} \otimes \denot{\Delta} \otimes D_q\denot{\tau}\rightarrow 
  \denot{\rho} 
\end{align*}
we require a lens
$\denot{\Phi\mid q+\Gamma,\Delta \vdash \case {e'} {\inl x.e} {\inr y.f}: \rho}$. 
We define it as the composition
\[
  [h_2,h_3] \circ \Theta \circ
  (id_{\denot{\Phi}\otimes\denot{\Delta}}\otimes (\varphi \circ D_q(h_1))) \circ 
  (t^q_{\denot{\Phi}} \otimes id_{D_q\denot{\Gamma} \otimes  \denot{\Delta}}).
\]
Above, the map 
$\varphi : D_q \denot{\sigma + \tau} \rightarrow D_q\denot{\sigma} + D_q\denot{\tau}$
is the identity lens, and the map 
$
\Theta_{X,Y,Z} : X \otimes (Y + Z) \rightarrow (X \otimes Y) + (X \otimes Z)
$ 
is given by the triple 
\begin{align*}
    f_{\Theta}(x,w) &\triangleq 
        \begin{cases}
          inl \ (x,y) & \text{if } w = inl \ y \\ 
          inr \ (x,z) & \text{if } w = inr \ z
        \end{cases}\\
    \tilde{f}_{\Theta}(x,w) &\triangleq  f_{\Theta}(x,w) \\
    b_{\Theta}((x,w),u) &\triangleq  
        \begin{cases}
          (a, inl \ b) & 
                \text{if } u = inl \ (a,b) \\ 
          (a, inr \ c) & 
                \text{if } u = inr \ (a,c).
        \end{cases} 
\end{align*}
We check that the triple 
\[
\Theta_{X,Y,Z} : X \otimes  (Y + Z) \rightarrow (X \otimes  Y) + (X \otimes  Z) 
     \triangleq (f_\Theta,\tilde{f}_{\Theta},b_\Theta)
\] is well-defined when $r_Y,r_Z\neq\infty$.  
\begin{enumerate}[align=left]
\item[Property 1.] For any $x \in X$, $w \in Y+Z$, and
     $u \in (X \otimes  Y) + (X \otimes  Z)$ we are required to show    
\begin{align}
d_{X \otimes  (Y+Z)}\left((x,w),b_{\Theta}((x,w),u)\right)-r_{X\otimes (Y+Z)}\notag\\
\le d_{(X\otimes Y) + (X \otimes Z)}(\tilde{f}_{\Theta}(x,w),u)-r_{(X\otimes Y)+(X\otimes Z)}
\label{eq:prod_req1}
\end{align}
supposing 
\begin{equation}
d_{(X\otimes Y) + (X \otimes Z)}(\tilde{f}_{\Theta}(x,w),u) <\infty.
\label{eq:prop1_theta}
\end{equation}
From \Cref{eq:prop1_theta}, and by unfolding definitions, we have
\begin{enumerate}
\item 
if $w = inl \ y$ for some $y \in Y$, then $u = inl \ (x_1,y_1)$ 
     for some $(x_1,y_1) \in X \otimes  Y$
\item 
if $w = inr \ z$ for some $z \in Z$, then $u = inr \ (x_1,z_1)$ 
     for some $(x_1,z_1) \in X \otimes  Z$.
\end{enumerate}
In both cases, the \Cref{eq:prod_req1} is an equality.
\item[Property 2.] 
For any $x \in X$, $w \in Y+Z$, and
     $u \in (X \otimes  Y) + (X \otimes  Z)$ we are required to show 
\begin{equation}
{f_{\Theta}(b_{\Theta}((x,w),u)) = u} \label{eq:prop2_theta}
\end{equation} 
supposing \Cref{eq:prop1_theta} holds.

We consider the cases when $u = inl \ (x_1,y_1)$ for some 
$(x_1,y_1) \in X \otimes  Y$ and when 
$u = inr \  (x_1,z_1)$ for some $(x_1,z_1) \in X \otimes  Z$ 
as we did for Property 1.

In the first case, we have 
\begin{align*}
f_{\Theta}(b_{\Theta}((x,w),u)) &= f_{\Theta}(x_1, inl \ y_1) \\
&= inl \ (x_1,y_1).
\end{align*}
In the second case we have 
\begin{align*}
f_{\Theta}(b_{\Theta}((x,w),u))&= f_{\Theta}(x_1, inr \ z_1) \\
&= inr \ (x_1,z_1).
\end{align*}
\end{enumerate}
% SUM INTRO L
\item[Case ($+ \ \text{I}_{L,R}$).] Given the maps
\begin{align*}
h_l &= \denot{\Phi\mid \Gamma \vdash e : \sigma} : 
  \denot{\Phi} \otimes \denot{\Gamma} \rightarrow \denot{\sigma} \\
h_r &= \denot{\Phi\mid \Gamma \vdash e : \sigma} : 
  \denot{\Phi} \otimes \denot{\Gamma} \rightarrow \denot{\tau}
\end{align*}
define the maps 
\begin{align*}
\denot{\Phi\mid \Gamma \vdash \inl e  : \sigma + \tau} &\triangleq in_1\circ h_l \\
\denot{\Phi\mid \Gamma \vdash \inr e  : \sigma + \tau} &\triangleq in_2\circ h_r .
\end{align*}
% LET
\item[Case (Let).] Given the maps
\begin{align*}
h_1 &= \denot{\Phi\mid \Gamma \vdash e : \tau} : 
  \denot{\Phi} \otimes \denot{\Gamma} \rightarrow {\denot{\tau}}\\
h_2 &= \denot{\Phi\mid \Delta,x :_r \tau \vdash f : \sigma} : 
  \denot{\Phi} \otimes \denot{\Delta} \otimes D_r\denot{\tau} \rightarrow 
  {\denot{\sigma}}
\end{align*}
we define the map 
$\denot{\Phi\mid r + \Gamma,\Delta \vdash \slet x e f  : \sigma } $
as the following composition:
\[  h_2 \circ  
  (D_r(h_1) \otimes id_{\denot{\Phi}\otimes\denot{\Delta}}) \circ 
  (m_{r,\denot{\Phi},\denot{\Gamma}} \otimes id_{\denot{\Phi}\otimes\denot{\Delta}}) \circ
  (t^r_{\denot{\Phi}} \otimes id_{D_r\denot{\Gamma} \otimes \denot{\Delta}}) .
\] 
% DISCRETE INTRO
\item[Case (Disc).] Given the lens $h = \denot{\Phi \mid \Gamma \vdash e :
  \sigma}$ from the premise, we can define the map $\denot{\Phi \mid \Gamma \vdash
  e : m(\sigma)}$ as the composition $\eta\circ h$, where $\eta$ is the unit of
  the discrete monad $M$.
% DISCRETE LET
\item[Case (DLet).] 
Given the maps
\begin{align*}
h_1 &= \denot{\Phi\mid \Gamma \vdash e : \alpha} : 
  \denot{\Phi} \otimes \denot{\Gamma} \rightarrow {\denot\alpha}\\
h_2 &= \denot{\Phi,z :\alpha \mid \Delta \vdash f : \sigma} : 
  \denot{\Phi} \otimes \denot{\alpha} \otimes \denot{\Delta} \rightarrow 
  {\denot{\sigma}}
\end{align*}
define the map 
$\denot{\Phi \mid \Gamma,\Delta \vdash \dlet z e f  : \sigma } $ as the composition
\[
h_2 \circ (h_1 \otimes id_{\denot{\Phi} \otimes \denot{\Delta}} )\circ 
(t_{\denot{\Phi}} \otimes id_{\denot{\Gamma} \otimes \denot{\Delta}}).
\]
% ADD
\item[Case (Add).] Suppose the contexts $\Phi$ and $\Gamma$ have total length
  $i$. We define the map
  \[
    \denot{\Phi\mid \Gamma, x :_{\varepsilon + q} \num, y :_{\varepsilon + r} \num \vdash
    \add x y : \num }
  \]
as the composition
\[ 
  \pi_{i + 1}
  \circ
  (id_{\denot{\Phi}}
  \otimes (\varepsilon_{\denot{\Gamma}}
  \circ \overline{m_{0 \leq q_j, \denot{\sigma_j}}})
  \otimes id_{\denot{\num}})
  \circ
  (id_{\denot{\Phi} \otimes \denot{\Gamma}} \otimes (
  \mathcal{L}_{add}\circ
  (m_{\varepsilon \le\varepsilon + q, \denot{\num}}
  \otimes m_{\varepsilon \le \varepsilon + r, \denot{\num}}))) ,
\]
where the map $\overline{m_{0 \leq q_j, \denot{\sigma_j}}}$ applies the map $m_{0
\leq q, X} : D_q X \to D_0 X$ to each binding $\denot{x :_q \sigma}$ in the
context $\denot{\Gamma}$.

The lens $\mathcal{L}_{add} : D_\varepsilon(\R) \otimes D_\varepsilon(\R) \rightarrow \R$ 
is given by the triple
\begin{align}
  f_{add}(x_1,x_2) &\triangleq x_1 + x_2 \label{eq:f_add} \\
  \tilde{f}_{add}(x_1,x_2)&\triangleq  (x_1 + x_2)e^\delta; \quad |\delta|
  \le \varepsilon \label{eq:f_app_add}\\
  b_{add}((x_1,x_2),x_3) &\triangleq  
  \begin{cases}
    (x_1,x_2) & \text{if $x_1+x_2=x_3=0$} \\
    \left(\frac{x_3x_1}{x_1+x_2},
    \frac{x_3x_2}{x_1+x_2}\right)\label{eq:b_add} 
    & \text{otherwise}
  \end{cases}
\end{align}
where $\varepsilon = u/(1-u)$ and $u$ is the unit roundoff.
We now show $\mathcal{L}_{add} : D_\varepsilon(\R) \otimes D_\varepsilon(\R) \rightarrow
\R$ is well-defined. It is clear that $\mathcal{L}_{add}$ satisfies Property 2 of a
backward error lens, and so we are left with checking Property 1:
Assuming, for any $x_1,x_2,x_3 \in \R$,
\begin{equation}
  d_{\R}(\tilde{f}_{add}(x_1,x_2),x_3)<\infty,
  \label{eq:add_asum}
\end{equation}
we are required to show
\begin{align*}
	d_{\R \otimes \R}((x_1,x_2),b_{add}((x_1,x_2),x_3)) - r_{\R\otimes\R}
	&\le d_{\R}(\tilde{f}_{add}(x_1,x_2),x_3) \\
  &= d_{\R}((x_1+x_2)e^\delta,x_3). & \text{\Cref{eq:f_app_add}}
\end{align*}

\Cref{eq:add_asum} implies $d_\R((x_1+x_2)e^\delta, x_3)<\infty$;
by \Cref{eq:olver}, we have that $(x_1+x_2)$ and $x_3$ are either both zero, or
both non-zero and of the same sign. 

First, suppose $x_1+x_2=x_3=0$. Thus, as self-distance is bounded by slack,
\begin{align*}
  d_{\R \otimes \R}((x_1,x_2),b_{add}((x_1,x_2),x_3)) - r_{\R\otimes\R} 
  &= d_{\R \otimes \R}((x_1,x_2),(x_1,x_2)) - r_{\R\otimes\R} & \text{\Cref{eq:b_add}}\\
  &\leq 0 = d_\R((x_1+x_2)e^\delta,x_3), & \text{\Cref{eq:olver}}
\end{align*}
as desired.
Second, suppose $(x_1+x_2)$ and $x_3$ are both non-zero and of the same sign.
Thus, by \Cref{eq:b_add} and \Cref{eq:d_tensor},
\begin{equation}
  d_{\R \otimes \R}((x_1,x_2),b_{add}((x_1,x_2),x_3)) - r_{\R\otimes\R}
  =\max\left\{d_\R \left(x_1,\frac{x_3x_1}{x_1+x_2}\right)
    -\varepsilon,d_\R\left(x_2,\frac{x_3x_2}{x_1+x_2}\right)-\varepsilon\right\}.
    \label{eq:add_max}
\end{equation}
If $x_1=0$ and $x_2\neq 0$, then $x_1=x_3x_1/(x_1+x_2)=0$, and $x_2$ and $x_3x_2/(x_1+x_2)$ 
have the same sign, so both components of the max are finite. 
The same holds if $x_1\neq 0$ and $x_2=0$, or both $x_1,x_2\neq 0$.
Now, suppose that 
\begin{equation}
  d_\R\left(x_2,\frac{x_3x_2}{x_1+x_2}\right)  \le 
  d_\R\left(x_1,\frac{x_3x_1}{x_1+x_2}\right);
\end{equation}
this implies that $x_1\neq 0$, and the other case is identical. 
Under this assumption, by \Cref{eq:add_max}, we have
\begin{equation}
  d_{\R \otimes  \R}\left((x_1,x_2),b_{add}((x_1,x_2),x_3)\right)-r_{\R\otimes\R} =
  d_\R\left(x_1,\frac{x_3x_1}{x_1+x_2}\right)-\varepsilon
\end{equation}
and we are now required to show
\begin{align}
  d_\R\left(x_1,\frac{x_3x_1}{x_1+x_2}\right) &\le
  d_{\R}((x_1+x_2)e^\delta,x_3) + \varepsilon.
\end{align}

Using the distance function given in \Cref{eq:olver}, this inequality becomes
\begin{align}
  \left|{ln\left(\frac{x_1+x_2}{x_3}\right)}\right| \le
  \left|{ln\left(\frac{x_1+x_2}{x_3}\right)} + \delta\right| + \varepsilon,
  \label{eq:p1_add_end}
\end{align}
which holds under the assumptions of $|\delta| \le \varepsilon$ and $0 <
\varepsilon$.
Set $\alpha = ln\left((x_1+x_2)/{x_3}\right)$,
and we detail the case that $\alpha < 0$. If
$\alpha + \delta < 0$ then $|\alpha| = -\alpha$ and
$|\alpha + \delta| = -(\alpha + \delta)$; the inequality in \Cref{eq:p1_add_end}
reduces to $\delta \le \varepsilon$, which follows by assumption.
Otherwise, if $0 \le \alpha + \delta$, then $-\alpha \le \delta \le \varepsilon$
and it suffices to show that $\varepsilon \le \alpha + \delta + \varepsilon$.

%%%%%%%%%%%%%%%%
% SUB

\item[Case (Sub).] We proceed the same as the case for (Add). We define a lens
  $\mathcal{L}_{sub} : D_\varepsilon(\R) \otimes D_\varepsilon(\R) \rightarrow \R$
  given by the triple
\begin{align*}
    f_{sub}(x_1,x_2) &\triangleq x_1 - x_2 \\
    \tilde{f}_{sub}(x_1,x_2)&\triangleq  (x_1 - x_2)e^\delta; \quad |\delta| 
        \le \varepsilon \\
    b_{sub}((x_1,x_2),x_3) &\triangleq  
    \begin{cases}
      (x_1,x_2) & \text{if $x_1-x_2=x_3=0$} \\
      \left(\frac{x_3x_1}{x_1-x_2}, \frac{x_3x_2}{x_1-x_2}\right) &\text{otherwise.}
    \end{cases}
\end{align*}

We check that $\mathcal{L}_{sub} : D_\varepsilon(\R) \otimes D_\varepsilon(\R)
\rightarrow \R$ is well-defined. 
For any $x_1,x_2,x_3 \in \R$ such that 
\begin{equation}
d_{\R}(\tilde{f}_{sub}(x_1,x_2),x_3)<\infty,
\label{eq:sub_asum}
\end{equation}
we need to check that $\mathcal{L}_{sub}$ satisfies the properties of an 
error lens. We take the distance function $d_\R$ as the metric given in 
\Cref{eq:olver}, so \Cref{eq:sub_asum} implies that $(x_1-x_2)$ and $x_3$ 
are either both zero or are both non-zero and of the same sign. 
The case where they are both zero is straightforward; we detail the other case.

\begin{enumerate}[align=left]
\item[Property 1.] We are required to show that 
\begin{align*}
d_{\R \otimes  \R}\left((x_1,x_2),b_{sub}((x_1,x_2),x_3)\right) - r_{\R\otimes\R} &\le 
     d_{\R}(\tilde{f}_{sub}(x_1,x_2),x_3) \\
    &\le d_{\R}((x_1-x_2)e^\delta,x_3).
\end{align*}
As in the case of the Add rule, we can easily show that the left-hand side is finite.
We detail now the case when 
\[
d_{\R \otimes  \R}\left((x_1,x_2),b_{sub}((x_1,x_2),x_3)\right) -r_{\R\otimes \R}=
    d_\R\left(x_1,\frac{x_3x_1}{x_1-x_2}\right)-\varepsilon; 
\]
that is, 
\[
d_\R\left(x_2,\frac{x_3x_2}{x_1-x_2}\right)  \le 
  d_\R\left(x_1,\frac{x_3x_1}{x_1-x_2}\right).
\]
Unfolding the definition of the distance function given in \Cref{eq:olver}, 
we are required to show 
\begin{equation}
\left|{ln\left(\frac{x_1-x_2}{x_3}\right)} \right| \le 
\left|{ln\left(\frac{x_1-x_2}{x_3}\right)} + \delta\right| + \varepsilon
\end{equation}
which holds under the assumptions of $|\delta| \le \varepsilon$ and 
$0 < \varepsilon$; the proof is identical to that given for the case of the 
Add rule. 
\item[Property 2.] 
\begin{align*}
f_{sub}\left(b_{sub}((x_1,x_2),x_3)\right) = 
     f_{sub}\left(\frac{x_3x_1}{x_1-x_2},\frac{x_3x_2}{x_1-x_2}\right) 
    = x_3.
\end{align*} 
\end{enumerate}
%%%%%%%
% MUL
%%%
\item[Case (Mul).]
  We proceed the same as the case for Add, with slightly different indices.
  We define a
  lens $\mathcal{L}_{mul} : D_{\varepsilon/2} (\R) \otimes D_{\varepsilon/2}(\R) \rightarrow \R$ 
  by the triple
\begin{align*}
    f_{mul}(x_1,x_2) &\triangleq x_1  x_2 \\
    \tilde{f}_{mul}(x_1,x_2)&\triangleq  x_1 x_2e^{\delta}; \quad |\delta| 
  \le \varepsilon \\
    b_{mul}((x_1,x_2),x_3) &\triangleq  
    \begin{cases}
      (x_1,x_2) & \text{if $x_1x_2=x_3=0$} \\
      \left(x_1\sqrt{\frac{x_3}{x_1x_2}},
      x_2\sqrt{\frac{x_3}{x_1x_2}}\right) & \text{otherwise.}
    \end{cases}
\end{align*}

We check that $\mathcal{L}_{mul} : D_{\varepsilon/2}(\R) \otimes
D_{\varepsilon/2}(\R) \rightarrow  \R$ is well-defined.
For any $x_1,x_2,x_3 \in \R$ such that 
\begin{equation} 
d_{\R}(\tilde{f}_{mul}(x_1,x_2),x_3)<\infty,
\label{eq:mul_asum}
\end{equation}
we need to check the that $\mathcal{L}_{mul}$ satisfies the properties 
of an error lens. We again take the distance function $d_\R$ as the metric 
given in \Cref{eq:olver}, so \Cref{eq:mul_asum} implies that $(x_1x_2)$ 
and $x_3$ are either both zero or are both non-zero and of the same sign; 
this guarantees that the backward map (containing square roots) is indeed 
well-defined. We detail the case they are both non-zero and of the same sign.
\begin{enumerate}[align=left]
\item[Property 1.] We are required to show
\begin{align*} \label{eq:p1_mul}
d_{\R \otimes  \R}\left((x_1,x_2),b_{mul}((x_1,x_2),x_3)\right) - r_{\R\otimes \R}&\le 
     d_{\R}(\tilde{f}_{mul}(x_1,x_2),x_3) \\
    &= d_{\R}(x_1 x_2e^{\delta},x_3).
\end{align*}
Unfolding the definition of the distance function (\Cref{eq:olver}), 
and noting that $x_1,x_2\neq 0$, we have
\begin{align*}
d_{\R \otimes \R}\left((x_1,x_2),b_{mul}((x_1,x_2),x_3)\right)-r_{\R\otimes\R}
&= d_\R\left(x_1,x_1\sqrt{\frac{x_3}{x_1x_2}}\right)-\frac{\varepsilon}{2} \\
&= d_\R\left(x_2,x_2\sqrt{\frac{x_3}{x_1x_2}}\right)-\frac{\varepsilon}{2} \\
&= \frac{1}{2}\left|ln \left( \frac{x_1x_2}{x_3}\right)\right|-\frac{\varepsilon}{2},
\end{align*}
and so we are required to show 
\begin{equation}  \label{eq:p1_mul_end}
\frac{1}{2}\left|ln \left( \frac{x_1x_2}{x_3}\right)\right| \le
\left|ln \left( \frac{x_1x_2}{x_3}\right) + \delta \right| + 
\frac{\varepsilon}{2}
\end{equation}
which holds under the assumptions of $|\delta| \le \varepsilon$ and $0 < \varepsilon$. 
Setting $\alpha = ln(x_1x_2/x_3)$, we detail the case
that $\alpha < 0$. If $\alpha + \delta < 0 $ then $\alpha < - \delta$ and 
it suffices to show that $-\frac{1}{2} \delta \le - \delta + \frac{1}{2} \varepsilon$, 
which follows by assumption. Otherwise, if $0 \le \alpha + \delta $ then 
$-\alpha \le \delta$ and it suffices to show that 
$\frac{1}{2} \delta \le \alpha + \delta + \frac{1}{2} \varepsilon$,
which follows by assumption.

\item[Property 2.] 
\begin{align*}
f_{mul}\left(b_{mul}((x_1,x_2),x_3)\right) = 
     f_{mul}\left(x_1\sqrt{\frac{x_3}{x_1x_2}},x_2\sqrt{\frac{x_3}{x_1x_2}}\right)
    = x_3.
\end{align*} 
\end{enumerate}
%%%
% DIV
%%%
\item[Case (Div).]
  We proceed the same as the case for (Add), with slightly different indices.
  We define a lens $\mathcal{L}_{div} : D_{\varepsilon/2} (\R) \otimes
  D_{\varepsilon/2}(\R) \rightarrow (\R + 1)$ (where $1=(\{\star\},\underline{0},0)$) by the triple
\begin{align*}
    f_{div}(x_1,x_2) &\triangleq
      \begin{cases}
        inl\ {x_1}/{x_2} & \text{ if } x_2 \neq 0 \\
        inr\ \star    & \text{ otherwise }
      \end{cases}  \\
    \tilde{f}_{div}(x_1,x_2)&\triangleq        
        \begin{cases}
        inl\ {x_1}e^\delta/{x_2} & \text{ if } x_2 \neq 0 ; \quad |\delta| \le \varepsilon \\
        inr\ \star    & \text{ otherwise }
      \end{cases}  \\
    b_{div}((x_1,x_2),x) &\triangleq 
        \begin{cases}
         \left(\sqrt{x_1x_2x_3},\sqrt{{x_1x_2}/x_3}\right) & \text{ if } x = inl \ x_3 \text{ and }x_3\neq 0\\
         (x_1,x_2)    & \text{ otherwise. }
      \end{cases}
\end{align*}
We check that $\mathcal{L}_{div} : D_{\varepsilon/2} (\R) \otimes
D_{\varepsilon/2}(\R) \rightarrow (\R + 1)$ is well-defined.
For any $x_1,x_2 \in \R$ and $x \in \R + 1$ such that  
\begin{equation}\label{eq:div_asum}
d_{\R+ 1}(\tilde{f}_{div}(x_1,x_2),x) <\infty,
\end{equation}
we are required to show that $\mathcal{L}_{div}$ satisfies the 
properties of an error lens. From \Cref{eq:div_asum} and again assuming the 
distance function is given by \Cref{eq:olver}, 
we have three cases:
\begin{enumerate}
  \item $x_2=0$ and $x=inr\ \star$,
  \item $x_2\neq 0$, $x=inl\ x_3$ for some $x_3\in\R$, and $x_1/x_2=x_3=0$,
  \item $x_2\neq 0$, $x=inl\ x_3$ for some $x_3\in\R$, and $x_1/x_2$ and $x_3$ 
    are both non-zero and of the same sign.
\end{enumerate}
It is straightforward to prove the two properties in the first two cases.
We detail here the third case.
\begin{enumerate}[align=left]
\item[Property 1.] We need to show 
\begin{align}\label{eq:p1_div}
d_{\R \otimes \R}\left((x_1,x_2),b_{div}((x_1,x_2),x)\right) -r_{\R\otimes\R} &\le 
     d_{\R + 1}(\tilde{f}_{div}(x_1,x_2),x) \nonumber-\varepsilon/2 \\
    &\le d_{\R}\left(\frac{x_1}{x_2}e^\delta,x_3\right)-\varepsilon/2.
\end{align}
Unfolding the definition of the distance function (\Cref{eq:olver}), we have
\begin{align*}
d_{\R \otimes \R}\left((x_1,x_2),b_{div}((x_1,x_2),x_3)\right)-r_{\R\otimes\R}  
&= d_\R(x_1,\sqrt{{x_1}{x_2x_3}})-\varepsilon/2 \\
&= d_\R(x_2,\sqrt{{x_1x_2}/{x_3}})-\varepsilon/2 \\
&= \frac{1}{2}\left|ln \left( \frac{x_1}{x_2x_3}\right)\right|-\frac{\varepsilon}{2},
\end{align*}
and so we are required to show 
\begin{align*}
\frac{1}{2}\left|ln \left( \frac{x_1}{x_2x_3}\right)\right|\le
\left|ln \left( \frac{x_1}{x_2x_3}\right) + \delta \right| + 
  \frac{\varepsilon}{2},
\end{align*}
which holds under the assumptions of $|\delta| \le \varepsilon$ and $0 < \varepsilon$; 
the proof is identical to that given for the case of the Mul rule. 
\item[Property 2.] 
\begin{align*}
f_{div}(b_{div}((x_1,x_2),x_3)) = 
     f_{div}(\sqrt{x_1x_2x_3},\sqrt{{x_1x_2}/x_3}) 
    = x_3.
\end{align*} 
\end{enumerate}
%%%
% DMUL
%%%
\item[Case (DMul).]
  We proceed similarly as for (Add). We define a
  lens $\mathcal{L}_{dmul} : M(\R) \otimes  D_{\varepsilon}(\R) \rightarrow \R$
  by the triple
\begin{align*}
    f_{dmul}(x_1,x_2) &\triangleq x_1  x_2 \\
    \tilde{f}_{dmul}(x_1,x_2)&\triangleq  x_1 x_2e^{\delta}; \quad |\delta| 
  \le \varepsilon \\
    b_{dmul}((x_1,x_2),x_3) &\triangleq 
    \begin{cases}
      (x_1,x_2) & \text{if $x_1x_2=x_3=0$} \\
      (x_1,x_3/x_1) & \text{otherwise.}
    \end{cases}
\end{align*}
%%%
We check that 
$\mathcal{L}_{dmul} : M(\R) \otimes  D_{\varepsilon}(\R) \rightarrow \R$ is well-defined.
For any $x_1,x_2,x_3 \in \R$ such that 
\begin{equation} 
d_{\R}(\tilde{f}_{dmul}(x_1,x_2),x_3) <\infty,
\label{eq:mul_asum}
\end{equation}
we need to check that $\mathcal{L}_{dmul}$ satisfies the properties 
of an error lens. We again take the distance function $d_\R$ as the metric 
given in \Cref{eq:olver}, so \Cref{eq:mul_asum} implies that $(x_1x_2)$ 
and $x_3$ are either both zero or are both non-zero and of the same sign.
The case where they are both zero is straightforward; we detail the other case.
\begin{enumerate}[align=left]
\item[Property 1.] We are required to show
\begin{align*} \label{eq:p1_mul}
d_{M(\R) \otimes  \R}\left((x_1,x_2),b_{dmul}((x_1,x_2),x_3)\right) -r_{M(\R)\otimes \R} 
  &\le d_\R(\tilde{f}_{dmul}(x_1,x_2),x_3) \\
    &\le d_{\R}(x_1 x_2e^{\delta},x_3).
\end{align*}
Unfolding the definition of the distance function (\Cref{eq:olver}), we have
\begin{align*}
d_{M(\R) \otimes  \R}
  \left((x_1,x_2),b_{dmul}((x_1,x_2),x_3)\right) -r_{M(\R)\otimes \R} 
&= \max\{d_\alpha(x_1,x_1),d_\R(x_2,x_3/x_1) -\varepsilon \} \\
&= \max\left\{0,\left|ln \left( \frac{x_1x_2}{x_3}\right)\right| -\varepsilon\right\},
\end{align*}
and so we are required to show 
\begin{equation}  \label{eq:p1_mul_end}
\left|ln \left( \frac{x_1x_2}{x_3}\right)\right| \le
\left|ln \left( \frac{x_1x_2}{x_3}\right) + \delta \right| + 
  \varepsilon
\end{equation}
which holds under the assumptions of $|\delta| \le \varepsilon$ and $0 < \varepsilon$;
the proof is identical to that given in the Add rule. 

\item[Property 2.] 
\begin{align*}
f_{dmul}\left(b_{dmul}((x_1,x_2),x_3)\right) = 
     f_{dmul}\left(x_1,x_3/x_1\right)
    = x_3.
\end{align*} 
\end{enumerate}

\end{description}

\section{\LangS{}: A Language for Projecting \bea{} into \Set \ }
\label{app:LangS}
\begin{figure}
%% ROW1
\begin{center}
%% var
\AXC{}
\RightLabel{(Var)}
\UIC{$\Gamma, x: \sigma, \Delta \vdash x : \sigma$}
\bottomAlignProof
\DisplayProof
\hskip 0.5em
%% unit
\AXC{}
\RightLabel{(Unit)}
\UIC{$\Gamma \vdash (): \mathbf{unit}$}
\bottomAlignProof
\DisplayProof
\hskip 0.5em
%% const
\AXC{$k \in R$}
\RightLabel{(Const)}
\UIC{$\Gamma \vdash k : \textbf{num}$}
\bottomAlignProof
\DisplayProof
\hskip 0.5em
\vskip 1em
%%
%% ROW2
%% prod intro
\AXC{$\Phi,\Gamma \vdash e : \sigma$}
\AXC{$\Phi,\Delta \vdash f : \tau$}
\RightLabel{($\otimes $ I)}
\BinaryInfC{$\Phi,\Gamma, \Delta \vdash ( e, f ): \sigma \otimes   \tau $}
\bottomAlignProof
\DisplayProof
\vskip 1em
%% prod elim
\AXC{$\Phi,\Gamma \vdash e : \tau_1 \otimes  \tau_2$}
\AXC{$\Phi,\Delta, x : \tau_1, y : \tau_2 \vdash f : \sigma $}
\RightLabel{($\otimes $ E)}
\BIC{$\Phi,\Gamma, \Delta  \vdash \slet {(x,y)} e f : \sigma$}
\bottomAlignProof
\DisplayProof
\vskip 1em
%%
%% ROW3
% sum elim
\AXC{$\Phi,\Gamma \vdash e' : \sigma+\tau$}
\AXC{$\Phi,\Delta, x: \sigma \vdash e : \rho$}
\AXC{$\Phi,\Delta, y: \tau \vdash f: \rho$}
\RightLabel{($+$ E)}
\TIC{$\Phi,\Gamma,\Delta \vdash \mathbf{case} \ e' 
    \ \mathbf{of} \ (\inl x.e \ | \ \inr y.f) : \rho$}
\bottomAlignProof
\DisplayProof
\vskip 1em
%%
%% ROW 5
%% ind sum intro
\AXC{$\Phi,\Gamma \vdash e : \sigma$ }
\RightLabel{($+$ $\text{I}_L$)}
\UIC{$\Phi,\Gamma \vdash \inl \ e : \sigma + \tau$}
\bottomAlignProof
\DisplayProof
\hskip 0.5em
%% ind sum intro
\AXC{$\Phi,\Gamma \vdash e : \tau$ }
\RightLabel{($+$ $\text{I}_R$)}
\UIC{$\Phi,\Gamma \vdash \inr \ e : \sigma + \tau$}
\bottomAlignProof
\DisplayProof
\vskip 1em
%%
%%% ROW 6
% let 
\AXC{$\Phi,\Gamma \vdash e :  \tau$}
\AXC{$\Phi,\Delta, x : \tau \vdash f : \sigma$}
\RightLabel{(Let)}
\BIC{$\Phi,\Gamma,\Delta \vdash \slet x e f : \sigma$}
\bottomAlignProof
\DisplayProof
\vskip 1em
%%
%%% ROW 7
% add, sub
\AXC{$\Phi,\Gamma \vdash e : \num$}
\AXC{$\Phi,\Delta \vdash f : \num $}
\AXC{$ \mathbf{Op} \in \{\mathbf{add}, \mathbf{sub}, 
    \mathbf{mul}\}$}
\RightLabel{(Op)}
\TIC{$\Phi,\Gamma, \Delta \vdash \mathbf{Op} \ e \ f : \num$}
\bottomAlignProof
\DisplayProof
\vskip 1em
%%
%%% ROW 8
%  div 
\AXC{$\Phi,\Gamma \vdash e : \num$}
\AXC{$\Phi,\Delta \vdash f : \num $}
\RightLabel{(Div)}
\BIC{$\Phi,\Gamma, \Delta \vdash \fdiv e f : \num + \unit$}
\bottomAlignProof
\DisplayProof

\end{center}
    \caption{Full typing rules for \LangS. }
    \label{fig:typing_rules_2_full}
\end{figure}

This section of the appendix provides the details of the intermediate 
language \LangS{} briefly described in \Cref{sec:metatheory}.
The type system of \LangS{} corresponds closely to \bea's, with one important
difference: there is no discrete type operator in \LangS{}.
\begin{align}
  \sigma,\tau \ ::=~ \num \mid \unit \mid \sigma \otimes \sigma \mid 
                \sigma + \sigma \tag*{(\LangS{} types)}
\end{align}
We project types from \bea{} into \LangS{} using an operator, $\Lambda$,
which inductively takes discrete types $m(\sigma)$ to $\sigma$.

Terms are typed
with judgments of the form $\Phi, \Gamma \vdash e : \tau$, where the typing
context $\Gamma$ corresponds to a linear typing context of \bea{} with all
of the grade information erased, and the typing context $\Phi$ corresponds to
a discrete typing context of \bea{} using \LangS{} types.
We will denote the erasure of grade information and conversion to \LangS{} types 
from a \bea{} typing environment $\Delta$ as $\Delta^\circ$.
Thus, the disjoint union of the contexts $\Phi^\circ, \Delta^\circ$ is well-defined.

The grammar of terms in \LangS{} is mostly unchanged from the grammar of
\bea{}, excluding the DLet and Disc rules, except that \LangS{} extends 
$\bea{}$ to include primitive constants drawn from a signature $R$:
\begin{align*}
  e,f \ ::=~ \dots \mid k \in R \tag*{(\LangS{} terms)}
\end{align*}
The typing relation of \LangS{} is entirely standard for a first-order simply
typed language; the full rules are given in \Cref{fig:typing_rules_2_full}.  
We also use $\Lambda$ to project \bea{} terms into \LangS{} terms, essentially stripping away 
constructs specific to linear and discrete variables.
\begin{definition}{(Projection from \bea{} terms to \LangS{} terms.)}
  \begin{gather*}
    \Lambda(x)\triangleq x \qquad \Lambda(())\triangleq() \qquad \Lambda(!e)\triangleq \Lambda(e) 
    \qquad \Lambda((e,f))\triangleq(\Lambda(e), \Lambda(f)) \\
    \Lambda(\inl e)\triangleq\inl(\Lambda(e)) \qquad \Lambda(\inr e)\triangleq\inr (\Lambda(e)) \\
    \Lambda(\slet{x}{e}{f})\triangleq\Lambda(\dlet{x}{e}{f})\triangleq\slet{x}{\Lambda(e)}{\Lambda(f)} \\
    \Lambda(\slet{(x,y)}{e}{f})\triangleq\Lambda(\dlet{(x,y)}{e}{f})\triangleq\slet{(x,y)}{\Lambda(e)}{\Lambda(f)} \\
    \Lambda(\case{e'}{\inl x.e}{\inr y.f})\triangleq \case{\Lambda(e')}{\inl x.\Lambda(e)}{\inr y.\Lambda(f)} \\
    \Lambda(\mathbf{Op}\ e\ f)\triangleq\mathbf{Op}\ \Lambda(e)\ \Lambda(f),\ \mathbf{Op}\in\{\mathbf{add}, \mathbf{sub}, 
    \mathbf{mul}, \mathbf{div}\} \qquad
    \Lambda(\mathbf{dmul}\ e\ f)\triangleq \mathbf{mul}\ \Lambda(e)\ \Lambda(f)
  \end{gather*}
\end{definition}

With $\Lambda$, the close correspondence between derivations in
$\bea{}$ and derivations in \LangS{} is summarized in the following lemma.  

\begin{lemma} \label{lem:derive_ls} Let 
  $\Phi \mid \Gamma \vdash e : \tau$ be a well-typed term in $\bea{}$.  Then
  there is a derivation of $\Phi^\circ,\Gamma^\circ \vdash \Lambda(e) : \Lambda(\tau)$ in \LangS.
\end{lemma}

\begin{proof}
The proof of \Cref{lem:derive_ls} follows by induction on the \bea{} derivation
$\Phi \mid \Gamma \vdash e : \tau$.  Most cases are immediate by application of
the corresponding \LangS{} rule.  The rules for primitive operations require
application of the \LangS{} (Var) rule.  We demonstrate the derivation for the
case of the (Add) rule:
\begin{description}
\item[\textbf{Case (Add).}]
Given a \bea{} derivation of 
\[
  \Phi \mid \Gamma, x :_{\varepsilon + r_1} \num,
    y:_{\varepsilon + r_2}\num \vdash \add x y : \num
\] 
we are required to show a \LangS{} derivation of 
\[
  \Phi^\circ,\Gamma^\circ, x:\num, y:\num \vdash \add x y : \num
\]
which follows by application of the Var rule for \LangS : \\
  \begin{center}
  \AXC{}
  \RightLabel{(Var)}
  \UIC{$\Phi^\circ,\Gamma^\circ, x: \num \vdash x :  \num$}
  \AXC{}
  \RightLabel{(Var)}
  \UIC{$y : \num \vdash y : \num$}
  \RightLabel{(Add)}
  \BIC{ $\Phi^\circ,\Gamma^\circ, x : \num, y:\num \vdash \add x y : \num$ }
  \bottomAlignProof
  \DisplayProof
  \end{center}
\end{description}
\end{proof}

\LangS{} satisfies the basic properties of weakening and substitution:

\begin{lemma}[Weakening] 
  Let $\Gamma \vdash e : \tau$ be a well-typed \LangS{} term. Then for any
  typing environment $\Delta$ disjoint with $\Gamma$, there is a derivation of
  $\Gamma, \Delta \vdash e : \tau$.  \label{lem:weakening}
\end{lemma}

Values, the subset of terms that are allowed as
results of evaluation, are defined as follows.
  \begin{alignat*}{1}
         &\text{Values } v ~::=~ 
         ()
         \mid k \in R
         \mid  (v, v)
         \mid \inl v
         \mid \inr v
  \end{alignat*}
We write $e[v/x]$ for the capture-avoiding substitution of the value $v$ for
all free occurrences of $x$ in $e$. Given a typing environment $x_1 : \tau_1,
\cdots, x_i : \tau_i = \Gamma$, we denote the simultaneous substitution of a
vector of values $v_1, \cdots, v_i = \bar{v}$ for the variables in $\Gamma$ as
$e[\bar{v}/\dom(\Gamma)]$.

We describe a notational convention. 
For a vector $\gamma_1,\dots, \gamma_i = \bar{\gamma}$ of 
well-typed closed values and 
a typing environment $x_1:\sigma_1,\dots, x_i:\sigma_i =\Gamma$ (note the tacit 
assumption that $\gamma$ and $\Gamma$ have the same length) we write 
$\bar{\gamma} \vDash \Gamma$ to denote the following:
\begin{align}
\bar{\gamma} \vDash \Gamma \triangleq
  \forall x_i \in \dom(\Gamma). ~ \emptyset \vdash \gamma_i : \Gamma(x_i).
  \label{eq:subst_sat}
\end{align}

\begin{theorem}[Substitution] \label{thm:subst}
  Let $\Gamma \vdash e : \tau$ be a well-typed \LangS{} term. Then for any
  well-typed substitution $\bar{\gamma} \vDash \Gamma$ of closed values, there
  is a derivation $\emptyset \vdash e[\bar{\gamma}/\dom(\Gamma)] : \tau.$
\end{theorem}

\begin{proof}
We state and prove a stronger theorem, from which this theorem immediately follows
by setting $\Delta=\emptyset$.
If $\Gamma,\Delta\vdash e:\tau$ is a well-typed \LangS{} term, then for any 
well-typed substitution $\bar{\gamma}\vDash \Gamma$ of closed values, 
there is a derivation $\Delta\vdash e[\bar\gamma/\dom(\Gamma)]:\tau$.

We prove this by induction on the structure of the derivation $\Gamma \vdash e : \tau$. 
We detail (Var) and ($\otimes$ E) here.
\begin{description}
\item[\textbf{Case (Var)}]
Suppose $\Gamma,x:\sigma,\Delta\vdash x:\sigma$ and let $\bar{\gamma}\vDash \Gamma,x:\sigma$ be a well-typed
substitution of closed values. In particular, there exists some $y\in\dom(\Gamma)$ such that 
$\emptyset\vdash y:\sigma$. We apply \Cref{lem:weakening} to conclude that $\Delta\vdash y:\sigma$.

\item[\textbf{Case ($\otimes$ E).}] 
Suppose $\Phi,\Gamma,\Delta\vdash\slet{(x,y)}{e}{f}:\sigma$ and, for simplicity, let 
$\bar{\gamma}\vDash\Phi,\Gamma,\Delta$ be a well-typed substitution of closed values. 
From $\bar\gamma$ we derive a substitution $\bar\gamma'\vDash \Phi,\Gamma$, and from the 
induction hypothesis on the left premise we have 
$\emptyset\vdash e[\bar\gamma'/\dom(\Phi,\Gamma)]:\tau_1\otimes\tau_2$.
From $\bar\gamma$ we derive a substitution $\bar\gamma''\vDash \Phi,\Delta$, and from the
induction hypothesis on the right premise we have 
$x:\tau_1,y:\tau_2\vdash f[\bar\gamma''/\dom(\Phi,\Delta)]:\sigma$.
This provides the premises needed to apply the typing rule ($\otimes$
E). The desired conclusion then follows from the definition of substitution. 
\end{description}
\end{proof}

\paragraph*{An Operational Semantics for \LangS{}}
\begin{figure}
%% ROW1
\begin{center}
\AXC{}
\UIC{$() \Downarrow ()$}
\bottomAlignProof
\DisplayProof
\hskip 0.5em
\AXC{$e \Downarrow u$}
\AXC{$f \Downarrow v$}
\BIC{$(e,f) \Downarrow (u,v)$}
\bottomAlignProof
\DisplayProof
\hskip 0.5em
\AXC{$e \Downarrow (u,v)$}
\AXC{$f[u/x][v/y] \Downarrow w$}
\BIC{$\slet {(x,y)} e f \Downarrow w$}
\bottomAlignProof
\DisplayProof
\hskip 0.5em
\vskip 1em
%% 
%% ROW2
%%
%% 
\AXC{}
\UIC{$k \in R \Downarrow k \in R $}
\bottomAlignProof
\DisplayProof
\hskip 0.5em
\AXC{$e \Downarrow v$ }
\UIC{$\inl e \Downarrow \inl v$}
\bottomAlignProof
\DisplayProof
\hskip 0.5em
\AXC{$e \Downarrow v$ }
\UIC{$\inr e \Downarrow \inr v$}
\bottomAlignProof
\DisplayProof
\vskip 1em
%%
%% ROW3
% 
\AXC{$e \Downarrow u$ }
\AXC{$f[u/x] \Downarrow v$ }
\BIC{$\slet x e f \Downarrow v$}
\bottomAlignProof
\DisplayProof
\vskip 1em
%%
%% ROW3
% 
\AXC{$e \Downarrow \inl v$}
\AXC{$e_1[v/x] \Downarrow w$}
\BIC{$\mathbf{case} \ e \ \mathbf{of} \ (\inl x.e_1 \ |\ \inr y.e_2) \Downarrow w$}
\bottomAlignProof
\DisplayProof
\hskip 0.5em
\AXC{$e \Downarrow \inr v$}
\AXC{$e_2[v/y] \Downarrow w$}
\BIC{$\mathbf{case} \ e \ \mathbf{of} \ (\inl x.e_1 \ | \ \inr y.e_2) \Downarrow w$}
\bottomAlignProof
\DisplayProof
\vskip 1em
%%
%% ROW 4
%%
\AXC{$e_1 \stepid k_1$}
\AXC{$e_2 \stepid k_2$}
\AXC{$\mathbf{Op} \in \{\mathbf{add}, \mathbf{sub}, \mathbf{mul}, \mathbf{div}\}$}
\TIC{$\mathbf{Op} \ {e_1} \ {e_2} \stepid f_{op} \ ({k_1}, {k_2})$}
\bottomAlignProof
\DisplayProof
\vskip 1em
%%
%% ROW 5
%%
\AXC{$e_1 \stepap k_1$}
\AXC{$e_2 \stepap k_2$}
\AXC{$\mathbf{Op} \in \{\mathbf{add}, \mathbf{sub}, \mathbf{mul}, \mathbf{div}\}$}
\TIC{$\mathbf{Op} \ {e_1} \ {e_2} \stepap \tilde{f}_{op} \ ({k_1}, {k_2})$}
\bottomAlignProof
\DisplayProof
\vskip 1em

\end{center}
    \caption{Evaluation rules for \LangS. 
    A generic step relation ($\Downarrow$) is used when the rule is identical for 
    both the ideal ($\stepid$) and approximate ($\stepap$) step relations.}
    \label{fig:op_semantics_full}
\end{figure}

Intuitively, an ideal problem and its approximating program can behave
differently given the same input. Following this intuition, we allow programs
in \LangS{} to be executed under an ideal or approximate big-step operational
semantics. Evaluation rules are given in \Cref{fig:op_semantics_full}.
We write $e \stepid v$ (resp., $e \stepap v$) to
denote that a term $e$ evaluates to value $v$ under the ideal (resp.,
approximate) semantics. 
An important feature of \LangS{} is that it is deterministic and strongly
normalizing:

\begin{theorem}[Strong Normalization] \label{thm:normalizing}
  If $\emptyset \vdash e : \tau$ then the well-typed closed values $\emptyset
  \vdash v, v' : \tau$ exist such that $e \stepid v$ and $e \stepap v'$. 
\end{theorem}

In our main result of backward error soundness, we will relate the ideal and
approximate operational semantics given above to the backward error lens
semantics of \bea{} via an interpretation of programs in \LangS{} as morphisms
in the category \Set.

\subsection{Interpreting \LangS}
Our main backward error soundness theorem requires that we have explicit access
to each transformation in a backward error lens. We achieve this by lifting the
close syntactic correspondence between \LangS{} and \bea{} to a close semantic
correspondence using the forgetful functors $U_{id} : \Bel \rightarrow \Set$
and $U_{ap} : \Bel \rightarrow \Set$ to interpret \LangS{} programs in \Set. 

We start with the interpretation of \LangS{} types, defined as follows:
\begin{gather*}
    \pdenot{\num} \triangleq U\denot{\num}  \qquad
    \pdenot{\textbf{unit}} \triangleq U\denot{(\{\star\},\underline{0},0)} \\
    \pdenot{\sigma \otimes \tau}  \triangleq 
      \pdenot{\sigma} \times  \pdenot{\tau}  \qquad
   \pdenot{\sigma + \tau} \triangleq \pdenot{\sigma} + \pdenot{\tau} \qquad
   \pdenot{M(\sigma)}\triangleq \pdenot\sigma
\end{gather*}  

Given the above interpretation of types, the interpretation $\pdenot{\Gamma}$
of a \LangS{} typing context $\Gamma$ is then defined as 
\begin{align*}
  \pdenot{\emptyset} \triangleq U\denot{I} 
  \qquad \qquad
  \pdenot{\Gamma,x:\sigma} \triangleq 
  \pdenot{\Gamma} \times \pdenot{\sigma}
\end{align*}  

Now, using the above definitions for the interpretations of \LangS{} types and
contexts, we can use the interpretation of \bea{} (\Cref{def:interpL}) terms
along with the functors $U_{id}$ and $U_{ap}$ to define the interpretation of
\LangS{} programs as morphisms in \Set:

\begin{definition}(Interpretation of \LangS{} terms.) \label{def:interpS}
  Each typing derivation $\Gamma \vdash e : \tau$ in \LangS{} yields the set
  maps $\pdenot{e}_{id} : \pdenot{\Gamma} \rightarrow \pdenot{\tau}$ and
  $\pdenot{e}_{ap} : \pdenot{\Gamma} \rightarrow \pdenot{\tau}$, by structural
  induction on the \LangS{} typing derivation $\Gamma \vdash e : \tau$.
\end{definition}
\noindent We give the detailed constructions for \Cref{def:interpS} in
\Cref{app:interp_LS}. 

Given \Cref{def:interpS}, we can show that \LangS{} is semantically sound
and computationally adequate: a \LangS{} program computes to a value if and only
if their interpretations in \Set{} are equal. Because \LangS{} has an ideal and
approximate operational semantics as well as an ideal and approximate
denotational semantics, we have two versions of the standard theorems for
soundness and adequacy. 

\begin{theorem}[Soundness of $\pdenot{-}$]\label{thm:soundid}
  Let $\Gamma \vdash e : \tau$ be a well-typed \LangS{} term.  Then for any
  well-typed substitution of closed values $\bar{\gamma} \vDash \Gamma$, if
  $e[\bar{\gamma}/\dom(\Gamma)]\stepid v$ for some value $v$, then
  $\pdenot{\Gamma \vdash e : \tau}_{id}(\bar\gamma) = v$ (and similarly for $\stepap$ and $\pdenotap{-}$).
\end{theorem} 

\begin{proof}
By induction on the structure of the \LangS{} derivations $\Gamma \vdash e :
\tau$. The cases for (Var), (Unit), (Const), and (+ I) are trivial. In each 
case we apply inversion on the step relation to obtain the premise for the 
induction hypothesis. 

Applications of the symmetry map $s_{X,Y} : X \times Y \rightarrow Y \times X$
are elided for succinctness. Recall the diagonal map $t_X : X \rightarrow X
\times X$ on $\Set$, which is used frequently in the interpretation of
\LangS{}. 

\begin{description}
% TENS INTRO
\item[\textbf{Case ($\otimes $ I).}] 
We are required to show
\[
\pdenotid{\Phi, \Gamma, \Delta \vdash (e,f) : \sigma \otimes  \tau}
(\bar\gamma) = (u,v)
\] 
for some well-typed closed substitution  
$\bar{\gamma} \vDash \Phi, \Gamma, \Delta$ 
and value $(u,v)$ such that  
\[
(e,f)[\bar{\gamma}/\dom(\Phi,\Gamma,\Delta)] \stepid (u,v).
\]
From $\bar{\gamma}$ we derive the substitutions $\bar{\gamma}' \vDash \Phi$,
$\bar{\gamma}_1 \vDash \Gamma$, and $\bar{\gamma}_2 \vDash \Delta$.
By inversion on the step relation we then have 
\begin{align*}
e[\bar{\gamma}',\bar{\gamma}_1/\dom(\Phi,\Gamma)] &\stepid u \\
f[\bar{\gamma}',\bar{\gamma}_2/\dom(\Phi,\Delta)] &\stepid v.
\end{align*}
We conclude as follows: 
\begin{align*}
  \pdenotid{\Phi, \Gamma, \Delta \vdash (e,f) : \sigma \otimes  \tau}
  (\bar{\gamma})
  &= ((t_{\pdenot{\Phi}},id_{\pdenot{\Gamma}\times \pdenot{\Delta}});
  (\pdenotid{\Phi,\Gamma \vdash e : \sigma}, 
  \pdenotid{\Phi,\Delta \vdash f : \tau})) (\bar{\gamma}) \\
  \hfill \tag*{(\Cref{def:interpS})}\\
  &= (\pdenotid{\Phi,\Gamma \vdash e : \sigma}, 
  \pdenotid{\Phi,\Delta \vdash f : \tau}) 
  (\bar{\gamma}',\bar{\gamma}_1,
  \bar{\gamma}',\bar{\gamma}_2) \\
  \hfill \tag*{(Definition of $t_{\pdenot{\Phi}}$)} \\
  &= (u,v) \tag*{(IH)}.
\end{align*}
% TENS ELIM
\item[\textbf{Case ($\otimes $ E).}] 
We are required to show 
\[
\pdenotid{\Phi,\Gamma,\Delta \vdash \slet {(x,y)} e f: \sigma}
(\bar{\gamma}) = w
\] 
for some well-typed closed substitution $\bar{\gamma} \vDash
\Phi,\Gamma,\Delta$ and value $w$ such that 
\[
(\slet {(x,y)} e f)[\bar{\gamma}/\dom(\Phi,\Gamma,\Delta)] \stepid w.
\]
From $\bar{\gamma}$ we derive the substitutions
$\bar{\gamma}' \vDash \Phi$, $\bar{\gamma}_1 \vDash \Gamma$,
and $\bar{\gamma}_2 \vDash \Delta$. By inversion on the step 
relation we then have 
\begin{align*}
e[\bar{\gamma}',\bar{\gamma}_1/\dom(\Phi,\Gamma)] &\stepid (u,v) \\
f[\bar{\gamma}',\bar{\gamma}_2/\dom(\Phi,\Delta)][u/x][v/y] &\stepid w.
\end{align*}
We conclude as follows: 
\begin{align*}
  &\pdenotid{\Phi,\Gamma,\Delta \vdash \slet {(x,y)} e f: \sigma}
  (\bar{\gamma}) \\
  &= ((t_{\pdenot{\Phi}},id_{\pdenot{\Gamma}\times\pdenot{\Delta}}); 
  (h_1, id_{\pdenot{\Phi} \times \pdenot{\Delta}});h_2) 
  (\bar{\gamma}) \tag*{(\Cref{def:interpS})}\\
  &= ((h_1, id_{\pdenot{\Phi} \times \pdenot{\Delta}});h_2) 
  (\bar{\gamma}',\bar{\gamma}_1,
  \bar{\gamma}',\bar{\gamma}_2) 
  \tag*{(Definition of $t_{\pdenot{\Phi}}$)} \\
  &= (\pdenotid{\Phi,\Delta,x:\tau_1,y:\tau_2 \vdash f : \sigma}) 
  (\bar{\gamma}',\bar{\gamma}_2,u, v) \tag*{(IH)} \\
  &= w. \tag*{(IH)}
\end{align*}
% SUM ELIM
\item[\textbf{Case ($+$ E).}] 
We are required to show 
\[
\pdenotid{\Phi,\Gamma, \Delta \vdash \case {e'} {\inl x.e} {\inr y.f} : \rho}
(\bar{\gamma}) = w
\]
for some well-typed closed substitution $\bar{\gamma} \vDash
\Phi,\Gamma,\Delta$ and value $w$ such that 
\[
(\case {e'} {\inl x.e} {\inr y.f})[\bar{\gamma}/\dom(\Phi,\Gamma,\Delta)] \stepid w.
\]

We consider the case when $e' = \inl e_1$ for some $e_1 : \sigma$. 
From $\bar{\gamma}$ we derive the substitutions
$\bar{\gamma}' \vDash \Phi$, $\bar{\gamma}_1 \vDash \Gamma$,
and $\bar{\gamma}_2 \vDash \Delta$. By inversion on the step 
relation we then have 
\begin{align*}
e'[\bar{\gamma}',\bar{\gamma}_1/\dom(\Phi,\Gamma)] &\stepid \inl v \\
e[\bar{\gamma}',\bar{\gamma}_2/\dom(\Phi,\Delta)][v/x] &\stepid w.
\end{align*}
We conclude as follows: 
\begin{align*}
  &\pdenotid{\Phi,\Gamma, \Delta \vdash \case {e'} {\inl x.e} {\inr y.f} : \rho}
  (\bar{\gamma}) \\
  &= ((t_{\pdenot{\Phi}},
  id_{\pdenot{\Gamma} \times \pdenot{\Delta}}); 
  (id_{\pdenot{\Phi} \times \pdenot{\Delta}}, h_1);
  \Theta^S_{\pdenot{\Phi} \times \pdenot{\Delta},\pdenot{\sigma},
  \pdenot{\tau}};[h_2,h_3]) 
  (\bar{\gamma}) \tag*{(\Cref{def:interpS})}\\
  &= ((id_{\pdenot{\Phi} \times \pdenot{\Delta}}, h_1);
  \Theta^S_{\pdenot{\Phi} \times \pdenot{\Delta},\pdenot{\sigma},
  \pdenot{\tau}};[h_2,h_3]) 
  (\bar{\gamma}',\bar{\gamma}_1,
  \bar{\gamma}',\bar{\gamma}_2) 
  \tag*{(Definition of $t_{\pdenotid{\Phi}}$)} \\
  &= (\Theta^S_{\pdenot{\Phi} \times \pdenot{\Delta},\pdenot{\sigma},
  \pdenot{\tau}};[h_2,h_3]) 
  (\bar{\gamma}',\bar{\gamma}_2,
  \inl v) \tag*{(IH)} \\
  &= w. \tag*{(IH)}
\end{align*}
% LET 
\item[\textbf{Case (Let).}]
We are required to show 
\[
  \pdenotid{\Gamma, \Delta \vdash \slet x e f : \tau}
  (\bar{\gamma}) = v
\]

for some well-typed closed substitution $\bar{\gamma} \vDash
\Phi,\Gamma,\Delta$ and value $w$ such that 
\[
(\slet {x} e f)[\bar{\gamma}/\dom(\Phi,\Gamma,\Delta)] \stepid v.
\]
From $\bar{\gamma}$ we derive the substitutions
$\bar{\gamma}' \vDash \Phi$, $\bar{\gamma}_1 \vDash \Gamma$,
and $\bar{\gamma}_2 \vDash \Delta$. By inversion on the step 
relation we then have 
\begin{align*}
e[\bar{\gamma}',\bar{\gamma}_1/\dom(\Phi,\Gamma)] &\stepid u \\
f[\bar{\gamma}',\bar{\gamma}_2/\dom(\Phi,\Delta)][u/x] &\stepid v.
\end{align*}
We conclude as follows: 
\begin{align*}
  \pdenotid{\Phi,\Gamma,\Delta \vdash \slet {x}{e}{f}: \sigma}
  (\bar{\gamma}) &= ((t_{\pdenot{\Phi}},id_{\pdenot{\Gamma}\times\pdenot{\Delta}}); 
  (h_1, id_{\pdenot{\Phi} \times \pdenot{\Delta}});h_2) 
  (\bar{\gamma}) \tag*{(\Cref{def:interpS})}\\
  &= ((h_1, id_{\pdenot{\Phi} \times \pdenot{\Delta}});h_2) 
  (\bar{\gamma}',\bar{\gamma}_1,
  \bar{\gamma}',\bar{\gamma}_2) 
  \tag*{(Definition of $t_{\pdenot{\Phi}}$)} \\
  &= (\pdenotid{\Phi,\Delta,x:\tau_1 \vdash f : \sigma}) 
  (\bar{\gamma}',\bar{\gamma}_2,u ) 
  \tag*{(IH)} \\
  &= v. \tag*{(IH)}
\end{align*}
% OP 
\item[\textbf{Case (Op).}] 
We are required to show 
\[
\pdenotid{\Phi, \Gamma, \Delta \vdash \textbf{Op} ~e ~f : \num }
  (\bar{\gamma}) = f_{op}(k_1,k_2)
\]
for some well-typed closed substitution $\gamma \vDash \Phi, \Gamma, \Delta$
and value $f_{op}(k_1,k_2)$ such that 
\[
(\Phi, \Gamma, \Delta \vdash \textbf{Op} ~e ~f)
  [\bar{\gamma}/\dom(\Phi,\Gamma,\Delta)]\stepid ~ f_{op}(k_1,k_2).
\]
From $\bar{\gamma}$ we derive the substitutions
$\bar{\gamma}' \vDash \Phi$, $\bar{\gamma}_1 \vDash \Gamma$,
and $\bar{\gamma}_2 \vDash \Delta$. By inversion on the step 
relation we then have 
\begin{align*}
e[\bar{\gamma}',\bar{\gamma}_1/\dom(\Phi,\Gamma)] &\stepid k_1 \\
f[\bar{\gamma}',\bar{\gamma}_2/\dom(\Phi,\Delta)] &\stepid k_2.
\end{align*}
We conclude as follows: 
\begin{align*}
  \pdenotid{\Phi,\Gamma,\Delta \vdash \textbf{Op} ~ e ~f: \num}(\bar{\gamma})
  &= ((t_{\pdenot{\Phi}}, id_{\pdenot{\Gamma}\times\pdenot{\Delta}});(h_1, h_2); f_{op}) 
  (\bar{\gamma}) \tag*{(\Cref{def:interpS})}\\
  &= ((h_1, h_2);f_{op}) 
  (\bar{\gamma}',\bar{\gamma}_1,
  \bar{\gamma}',\bar{\gamma}_2) 
  \tag*{(Definition of $t_{\pdenot{\Phi}}$)} \\
  &= f_{op}(k_1,k_2). \tag*{(IH)}
\end{align*}
% OP 
\item[\textbf{Case (Div).}] Identical to the proof for (Op).
\end{description}
\end{proof}

\begin{theorem}[Adequacy of $\pdenot{-}$]\label{thm:adequacy}
  Let $\Gamma \vdash e : \tau$ be a well-typed \LangS{} term. Then for any
  well-typed substitution of closed values $\bar{\gamma} \vDash \Gamma$, if $
  \pdenot{\Gamma \vdash e : \tau}_{id}(\bar{\gamma}) =v$ for some value $v$, 
  then $e[\bar{\gamma}/\dom(\Gamma)]\stepid
  v$ (and similarly for $\stepap$ and $\pdenotap{-}$).
\end{theorem}

\begin{proof}
The proof follows directly by cases on $e$. Many cases are immediate and the
remaining cases, given that \LangS{} is deterministic, follow by substitution
(\Cref{thm:subst}) and normalization (\Cref{thm:normalizing}).  We show two
representative cases. 
\begin{description}
% VAR
\item[Case (Var).] 
Given $\Gamma, x : \sigma, \Delta \vdash x : \sigma$ and $\pdenotid{\Gamma, x :
\sigma, \Delta \vdash x : \sigma} (\bar{\gamma}) = v$ for
some value $v$ and some well-typed substitution $\bar{\gamma} \vDash \Gamma,
x:\sigma, \Delta$ we are required to show 
\[
x[\bar{\gamma}/\dom(\Gamma, x : \sigma, \Delta)] \stepid v
\] 
which follows by substitution (\Cref{thm:subst}) and normalization
(\Cref{thm:normalizing}). 
% TENS ELIM
\item[Case ($\otimes$ E).] 
Given $\Gamma, \Delta \vdash \slet x e f :  \tau$ and $\pdenotid{\Gamma, \Delta
\vdash \slet x e f : \tau} (\bar{\gamma}) = w$ for some
value $w$ and some well-typed derivation $\bar{\gamma} \vDash \Gamma, \Delta$
we are required to show 
\[
(\slet x e f)[\bar{\gamma}/\dom(\Gamma,\Delta)] \stepid w
\] 
which follows by substitution (\Cref{thm:subst}) and normalization
(\Cref{thm:normalizing}).  \qedhere
\end{description}
\end{proof}

Our main error backward error soundness theorem requires one final piece of
information: we must know that the functors $U_{id}$ and $U_{ap}$ project
directly from interpretations of \bea{} programs in \Bel{} (\Cref{def:interpL})
to interpretations of \LangS{} programs in \Set{} (\Cref{def:interpS}):  
\begin{lemma}[Pairing]\label{lem:pairing}
  Let $\Phi \mid \Gamma \vdash e : \sigma$ be a \bea{} program. Then we have 
\[
  U_{id}\denot{\Phi \mid \Gamma \vdash e : \sigma} = 
  \pdenotid{\Phi^\circ,\Gamma^\circ \vdash \Lambda(e) : \Lambda(\sigma)} 
  \text{ and } 
  U_{ap}\denot{\Phi \mid \Gamma \vdash e : \sigma} = 
  \pdenotap{\Phi^\circ,\Gamma^\circ \vdash \Lambda(e) : \Lambda(\sigma)}.
\]
\end{lemma}

\begin{proof}
The proof of \Cref{lem:pairing} follows by induction on the structure of the 
\bea{} derivation $\Phi\mid \Gamma \vdash e : \sigma$. We detail here the 
cases of pairing for the ideal semantics. 
\begin{description}
% VAR
\item[\textbf{Case (Var).}] 
\begin{align*}
  \Uid\denot{\Phi\mid \Gamma, x:_r \sigma \vdash x : \sigma} &= 
  \Uid(\pi_{i+j}\circ (id_{\denot{\Phi}}\otimes\varepsilon_{\denot{\Gamma}\otimes\denot\sigma})
    \circ (id_{\denot\Phi}\otimes \overline{m_{0\leq q_k,\denot{\sigma_k}}})) \tag*{(\Cref{def:interpL})} \\
  &= \pi_{i+j} \tag*{(Definition of $\Uid$)} \\
  &= \pdenotid{\Phi^\circ,\Gamma^\circ,x:\Lambda(\sigma) \vdash x : \Lambda(\sigma) }.
                                                \tag*{(\Cref{def:interpS})} 
\end{align*}
% DVAR
\item[\textbf{Case (DVar).}] 
\begin{align*}
\Uid\denot{\Phi,z: \alpha\mid\Gamma \vdash z : \alpha} &= 
\Uid(\pi_i\circ(id_{\denot{\Phi}\otimes \denot{\alpha}}\otimes \varepsilon_{\denot{\Gamma}})\circ
  (id_{\denot{\Phi}\otimes \denot{\alpha}}\otimes \overline{m_{0\leq q_k,\denot{\sigma_k}}})) \tag*{(\Cref{def:interpL})} \\
&=\pi_i \tag*{(Definition of $\Uid$)}\\
&= \pdenotid{\Phi^\circ,x:\Lambda(\sigma),\Gamma^\circ \vdash x : \Lambda(\sigma) }. \tag*{(\Cref{def:interpS})}
\end{align*}
% UNIT
\item[\textbf{Case (Unit).}] 
\begin{align*}
  \Uid\denot{\Phi\mid \Gamma \vdash () : \unit} &= \underline\star
  \tag*{(\Cref{def:interpL})} \\ 
  &= \pdenotid{\Phi^\circ\mid \Gamma^\circ \vdash () : \unit}.
                            \tag*{(\Cref{def:interpS})}
\end{align*}
% TENS INTRO
\item[\textbf{Case ($\otimes $ I).}] 
From the induction hypothesis we have
\begin{align*}
  \Uid{(h_1)} &= \pdenotid{\Phi^\circ,\Gamma^\circ \vdash \Lambda(e) : \Lambda(\sigma)} \\
  \Uid{(h_2)} &= \pdenotid{\Phi^\circ,\Delta^\circ \vdash \Lambda(f) : \Lambda(\tau)}.
\end{align*}
We conclude as follows:
\begin{align*}
\Uid\denot{\Phi \mid \Gamma, \Delta \vdash (e,f) : \sigma \otimes  \tau} &=
  \Uid(t_{\denot{\Phi}} \otimes  id_{\denot{\Gamma, \Delta}}); \Uid(h_1 \otimes
  h_2) \tag*{(\Cref{def:interpL})} \\ 
  &= (t_{\pdenot{\Phi}},id_{\pdenot{\Gamma} \times \pdenot{\Delta}});
  (\Uid(h_1), \Uid(h_2)) \tag*{(Definition of $\Uid$)} \\
  &= \pdenotid{\Phi^\circ,\Gamma^\circ, \Delta^\circ \vdash \Lambda((e,f)) : \Lambda(\sigma \otimes \tau)}.
  \tag*{(IH \& \Cref{def:interpS})}
\end{align*}
% TENS ELIM
\item[\textbf{Case ($\otimes $ E$_\sigma$).}] 
From the induction hypothesis we have 
\begin{align*}
  \Uid{(h_1)} &= \pdenotid{\Phi^\circ,\Gamma^\circ \vdash \Lambda(e) : \Lambda(\tau_1 \otimes \tau_2)} \\
  \Uid{(h_2)} &= \pdenotid{\Phi^\circ,\Delta^\circ,x : \Lambda(\tau_1), y:\Lambda(\tau_2)
  \vdash \Lambda(f) : \Lambda(\sigma)}.
\end{align*}
We conclude with the following: 
\begin{align*}
  &\Uid\denot{\Phi\mid  r+ \Gamma, \Delta \vdash \slet {(x,y)} e f  :\sigma} 
  \nonumber \\ 
  &= \Uid(h_2 \circ (
  (m^{-1}_{r,\denot{\tau_1},\denot{\tau_2}} \circ 
  D_r(h_1)) \otimes id_{\denot{\Phi}\otimes\denot{\Delta}})
  \circ (t_{\denot{\Phi}}^r\otimes id_{D_r\denot{\Gamma}\otimes\denot{\Delta}}))
  \tag*{(\Cref{def:interpL})} \\ 
  &= (t_{\pdenot{\Phi}},id_{\pdenot{\Gamma} \times  \pdenot{\Delta}});
  (\Uid (h_1), id_{\pdenot{\Phi} \times \pdenot{\Delta}}); \Uid(h_2)
  \tag*{(Definition of $\Uid$)} \\
  &= \pdenotid{\Phi^\circ,\Gamma^\circ,\Delta^\circ \vdash \slet {(x,y)} {\Lambda(e)} {\Lambda(f)}: \Lambda(\sigma)}. 
  \tag*{(IH \& \Cref{def:interpS})}
\end{align*} 
% TENS ELIM
\item[\textbf{Case ($\otimes $ E$_\alpha$)}] 
From the induction hypothesis we have 
\begin{align*}
  \Uid{(h_1)} &= \pdenotid{\Phi^\circ,\Gamma^\circ \vdash \Lambda(e) : \Lambda(\alpha_1 \otimes \alpha_2)} \\
  \Uid{(h_2)} &= \pdenotid{\Phi^\circ,\Delta^\circ,x : \Lambda(\alpha_1), y:\Lambda(\alpha_2) \vdash \Lambda(f) : \Lambda(\sigma)}.
\end{align*}
We conclude with the following: 
\begin{align*}
  \denot{\Phi\mid  \Gamma, \Delta \vdash \dlet {(x,y)} e f  :\sigma}
  &= h_2 \circ (h_1 \otimes  id_{\denot{\Phi} \otimes  \denot{\Delta}})
  \circ (t_{\denot{\Phi}} \otimes  id_{\denot{\Gamma} \otimes  \denot{\Delta}})
  \tag*{(\Cref{def:interpL})}\\ 
  &= (t_{\pdenot{\Phi}}, id_{\pdenot{\Gamma}\times\pdenot{\Delta}});
    (\Uid (h_1), id_{\pdenot{\Phi} \times \pdenot{\Delta}}); \Uid(h_2) \\
  \hfill \tag*{(Definition of $\Uid$)} \\
  &= \pdenotid{\Phi^\circ,\Gamma^\circ,\Delta^\circ \vdash 
  \slet {(x,y)} {\Lambda(e)} {\Lambda(f)}: \Lambda(\sigma)}. \\
  \hfill \tag*{(IH \& \Cref{def:interpS})}
\end{align*} 
% SUM ELIM
\item[\textbf{Case ($+$ E).}]
From the induction hypothesis, we have 
\begin{align*}
  \Uid{(h_1)}
  &= \pdenotid{\Phi^\circ,\Gamma^\circ \vdash \Lambda(e') : \Lambda(\sigma + \tau)}\\
  \Uid{(h_2)} 
  &= \pdenotid{\Phi^\circ,\Delta^\circ, x : \Lambda(\sigma) \vdash \Lambda(e): \Lambda(\rho)}\\
  \Uid{(h_3)} 
  &= \pdenotid{\Phi^\circ,\Delta^\circ, y: \Lambda(\tau) \vdash \Lambda(f): \Lambda(\rho)}.
\end{align*}
We conclude with the following: 
\begin{align*}
&\Uid{\denot{\Phi\mid q+\Gamma, \Delta \vdash \case {e'} {\inl x.e} {\inr y.f}: \sigma}} \\
&= \Uid(  
  [h_2,h_3] \circ \Theta \circ
  (id_{\denot{\Phi}\otimes\denot{\Delta}}\otimes(\varphi \circ D_q(h_1))) \circ
  (t^q_{\denot{\Phi}} \otimes id_{D_q\denot{\Gamma} \otimes  \denot{\Delta}})) \\
\hfill \tag*{(\Cref{def:interpL})} \\
&= (t_{\pdenot{\Phi}}, id_{\pdenot{\Gamma} \times \pdenot{\Delta}});
(id_{\pdenot{\Phi}\times\pdenot{\Delta}}, \Uid(h_1));
\Uid(\Theta) ;[\Uid(h_2),\Uid(h_3)] \\
\hfill \tag*{(Definition of $\Uid$)} \\
&= \pdenot{\Phi^\circ,\Gamma^\circ, \Delta^\circ \vdash 
\case {\Lambda(e')} {\inl x.\Lambda(e)} {\inr y.\Lambda(f)}: \Lambda(\sigma)}_{id}. \tag*{(IH \& \Cref{def:interpS})}
\end{align*}
% SUM INTRO
\item[\textbf{Case (+ I).}]  
We detail the case for (+I$_L$). From the induction hypothesis, we have
\[\Uid(h) = \pdenot{\Phi^\circ,\Gamma^\circ \vdash \Lambda(e) : \Lambda(\sigma)}.\]
\begin{align}
  \Uid(\denot{\Phi\mid \Gamma \vdash \inl e: \sigma + \tau}) 
  &= \Uid\left(in_1\circ h \right)\tag*{(\Cref{def:interpL})} \\
  &= \Uid(h);\Uid(in_1) \tag*{(Definition of $\Uid$)} \\
  &= \pdenotid{\Phi^\circ, \Gamma^\circ \vdash \Lambda(\inl e) : \Lambda(\sigma + \tau)}.  
  \tag*{(IH \& \Cref{def:interpS})}
\end{align}
% LET
\item[\textbf{Case (Let).}]
From the induction hypothesis we have
\begin{align*}
  \Uid{(h_1)} &= 
  \pdenotid{\Phi^\circ,\Gamma^\circ \vdash \Lambda(e) : \Lambda(\tau)} \\
  \Uid{(h_2)} &= 
  \pdenotid{\Phi^\circ,\Delta^\circ, x : \Lambda(\tau) \vdash \Lambda(f) : \Lambda(\sigma)}.
\end{align*}

We conclude with the following: 
\begin{align*}
  &\Uid \denot{\Phi\mid r + \Gamma, \Delta \vdash \slet x e f: \sigma} \\
  &=\Uid ( 
  h_2 \circ (D_r(h_1) \otimes id_{\denot{\Phi}\otimes\denot{\Delta}})
  \circ (m_{r,\denot{\Phi},\denot{\Gamma}} \otimes id_{\denot{\Phi}\otimes\denot{\Delta}}) 
  \circ (t^r_{\denot{\Phi}} \otimes id_{D_r\denot{\Gamma} \otimes \denot{\Delta}}))\\
  \hfill \tag*{(\Cref{def:interpL})} \\
  &=(t_{\pdenot{\Phi}}, id_{\pdenot{\Gamma} \times \pdenot{\Delta}});
  (\Uid(h_1),id_{\pdenot{\Phi}\times \pdenot{\Delta}}); \Uid(h_2) 
  \tag*{(Definition of $\Uid$)} \\
  &=\pdenot{\Phi^\circ, (r+\Gamma)^\circ, \Delta^\circ \vdash 
  \slet x {\Lambda(e)} {\Lambda(f)}: \Lambda(\sigma)}_{id}. \tag*{(IH \& \Cref{def:interpS})}
\end{align*}
% DISC
\item[\textbf{Case (Disc).}] 
From the induction hypothesis we have
\begin{align*}
  \Uid (h)= \pdenot{\Phi^\circ,\Gamma^\circ\vdash \Lambda(e):\Lambda(\sigma)}_{id}.
\end{align*}
We conclude with the following: 
\begin{align*}
  \Uid \denot{\Phi\mid\Gamma\vdash\ !e : m(\sigma)} 
  &= U_{id}(\eta\circ h) \tag*{(\Cref{def:interpL})} \\
  &= U_{id}(h) &\tag*{(Definition of $\Uid$)} \\
  &= \pdenot{\Phi^\circ,\Gamma^\circ\vdash \Lambda(e):\Lambda(\sigma)}_{id}. \tag*{(IH)}
\end{align*}
% DLET
\item[\textbf{Case (DLet).}]
From the induction hypothesis we have
\begin{align*}
  \Uid{(h_1)} &= 
  \pdenotid{\Phi^\circ,\Gamma^\circ \vdash \Lambda(e) : \Lambda(\alpha)} \\
  \Uid{(h_2)} &= 
  \pdenotid{\Phi^\circ,z : \Lambda(\alpha),\Delta^\circ \vdash \Lambda(f) : \Lambda(\sigma)}.
\end{align*}

We conclude with the following: 
\begin{align*}
  &\Uid \denot{\Phi\mid\Gamma, \Delta \vdash \dlet z e f: \sigma} \\
  &=\Uid ( 
  h_2 \circ (h_1 \otimes id_{\denot{\Phi}\otimes\denot{\Delta}})
  \circ (t_{\denot{\Phi}} \otimes id_{\denot{\Gamma} \otimes \denot{\Delta}}))
  \tag*{(\Cref{def:interpL})} \\
  &=(t_{\pdenot{\Phi}}, id_{\pdenot{\Gamma} \times \pdenot{\Delta}});
  (\Uid(h_1),id_{\pdenot{\Phi}\times \pdenot{\Delta}}); \Uid(h_2) 
  \tag*{(Definition of $\Uid$)} \\
  &=\pdenotid{\Phi^\circ, \Gamma^\circ, \Delta^\circ \vdash 
  \slet z {\Lambda(e)} {\Lambda(f)}: \Lambda(\sigma)}. \tag*{(IH \& \Cref{def:interpS})}
\end{align*}
% ADD
\item[\textbf{Case (Add).}] 
From \Cref{def:interpL} we have 
\begin{align*}
&\denot{\Phi\mid \Gamma,x:_{\varepsilon+q}\num,
    y:_{\varepsilon+r}\num \vdash \add x y : \num}\\
& \qquad \qquad =
\pi_{i + j+1} \circ \dots \circ
  (id_{\denot{\Phi} \otimes \denot{\Gamma}}
  \otimes (\mathcal{L}_{add} \circ
  (m_{\varepsilon \le\varepsilon + q, \denot{\num}}
  \otimes m_{\varepsilon \le \varepsilon + r, \denot{\num}}))).
\end{align*}
From \Cref{def:interpS} we have 
\begin{align*}
  h_1&= \pdenot{\Phi^\circ,\Gamma^\circ,x:\num\vdash x:\num}_{id} =\pi_{i+j+1} \\ 
  h_2&= \pdenot{\Phi^\circ,y:\num\vdash y:\num}_{id} =\pi_{i+1}.
\end{align*}
We conclude as follows:
\begin{align*}
  &\Uid{\denot{\Phi\mid \Gamma,x:_{\varepsilon+q}\num,y:_{\varepsilon+r}\num 
  \vdash \add x y : \num}} \\
  &= (id_{\pdenot{\Phi} \times \pdenot{\Gamma}}, f_{add}); \pi_{i + j+1} 
  \tag*{(\Cref{def:interpL} \& Definition of $\Uid$)} \\
  &=(t_{\pdenot{\Phi}},id_{\pdenot{\Gamma}},id_{\pdenot{\num}},id_{\pdenot{\num}});(h_1,h_2);f_{add} 
  \tag*{(IH)}\\
  &=\pdenotid{\Phi^\circ,\Gamma^\circ,x:\num,y:\num 
    \vdash \add x y : \num}. \tag*{(\Cref{def:interpS})}
\end{align*}
\end{description}
The cases for the remaining arithmetic operations are nearly identical 
to the case for \textbf{Add}.
\end{proof}

\section{Interpreting \LangS{} Terms}\label{app:interp_LS}
This appendix provides the detailed constructions of the interpretation of
\LangS{} terms for \Cref{def:interpS}.  The interpretation of terms is defined
over the typing derivations for \LangS{} given in
\Cref{fig:typing_rules_2_full}.  For each case, the ideal interpretation
$\pdenotid{-}$  is constructed explicitly, but the construction for
$\pdenotap{-}$ is nearly identical, requiring only that the forgetful functor
$\Uap$ is used in place of $\Uid$.  

Applications of the symmetry map $s_{X,Y} : X \times Y \rightarrow Y \times X$
are elided for succinctness. The diagonal map $t_X : X \rightarrow X \times X$
on $\Set$ is used frequently and is not elided.
\begin{description}
% VAR 
\item[\textbf{Case (Var).}] 
Define the maps 
$\pdenot{\Phi,\Gamma,x :\sigma, \Delta \vdash x : \sigma}_{id}$ and 
$\pdenot{\Phi,\Gamma,x :\sigma, \Delta \vdash x : \sigma}_{ap}$ 
in $\Set$ as the appropriate projection $\pi_i$.
% Unit 
\item[\textbf{Case (Unit).}] 
Define the set maps 
$\pdenot{\Phi,\Gamma \vdash () : \unit}_{id}$ and 
$\pdenot{\Phi,\Gamma \vdash () : \unit}_{ap}$ as the constant function 
returning the value $\star$. 
% CONST 
\item[\textbf{Case (Const).}] 
Define the maps 
$\pdenot{\Phi,\Gamma \vdash k : \num}_{id}$  and 
$\pdenot{\Phi,\Gamma \vdash k : \num}_{ap}$ 
in $\Set$ as the constant function taking points in 
$\pdenot{\Phi,\Gamma}$ to the value $k \in R$. 
% PROD INTRO
\item[\textbf{Case ($\otimes $ I).}] 
Given the maps 
\begin{align*}
h_1 &= \pdenotid{\Phi,\Gamma \vdash e : \sigma}: 
  \pdenot{\Phi} \times \pdenot{\Gamma} 
  \rightarrow {\pdenot{\sigma}}\\
h_2 &= \pdenotid{\Phi,\Delta \vdash f : \tau}: 
  \pdenot{\Phi} \times\pdenot{\Delta} 
  \rightarrow {\pdenot{\tau}}
\end{align*}
in $\Set$, define the map $\pdenotid{\Phi,\Gamma, \Delta \vdash (e,f) : \sigma
    \otimes  \tau}$ as 
\[
  (t_{\pdenot{\Phi}},id_{\pdenot{\Gamma}\times\pdenot{\Delta}});(h_1, h_2).
\]
% PROD ELIM
\item[\textbf{Case ($\otimes $ E).}] Given the maps 
\begin{align*}
h_1 &= \pdenot{\Phi,\Gamma \vdash e : \tau_1 \otimes  \tau_2}_{id}: 
  \pdenot{\Phi} \times \pdenot{\Gamma} 
  \rightarrow {\pdenot{\tau_1} \times  \pdenot{\tau_2}}\\
h_2 &= \pdenot{\Phi,\Delta, x : \tau_1, y: \tau_2 \vdash f : \sigma}_{id} : 
  \pdenot{\Phi} \times \pdenot{\Delta} \times \pdenot{\tau_1} \times
  \pdenot{\tau_2} \rightarrow {\pdenot{\sigma}}
\end{align*}
in $\Set$, define $\pdenot{\Phi,\Gamma, \Delta \vdash \slet {(x,y)} e f :
    \sigma}_{id}$ as 
\[
  (t_{\pdenot{\Phi}},id_{\pdenot{\Gamma}\times\pdenot{\Delta}}); (h_1,
  id_{\pdenot{\Phi} \times \pdenot{\Delta}});h_2.
\]
% SUM ELIM
\item[Case ($+$ E).] 
Given the maps 
\begin{align*}
h_1 &= \pdenotid{\Phi,\Gamma \vdash e' : \sigma + \tau}: 
  \pdenot{\Phi} \times \pdenot{\Gamma} 
  \rightarrow {\pdenot{\sigma + \tau}}\\
h_2 &= \pdenotid{\Phi,\Delta, x : \sigma \vdash e : \rho} : 
  \pdenot{\Phi} \times\pdenot{\Delta} \times \pdenot{\sigma} 
  \rightarrow {\pdenot{\rho}}\\
h_3 &= \pdenotid{\Phi,\Delta, y : \tau \vdash f : \rho} : 
  \pdenot{\Phi} \times \pdenot{\Delta} \times \pdenot{\tau} 
  \rightarrow {\pdenot{\rho}}
\end{align*}
in $\Set$, define $\pdenotid{\Phi,\Gamma, \Delta \vdash \case {e'} {\inl x.e} {\inr y.f}
: \rho}$ as 
\[
  (t_{\pdenot{\Phi}},
  id_{\pdenot{\Gamma} \times \pdenot{\Delta}}); 
  (id_{\pdenot{\Phi} \times \pdenot{\Delta}}, h_1);
  \Theta^S_{\pdenot{\Phi} \times \pdenot{\Delta},\pdenot{\sigma},
  \pdenot{\tau}};[h_2,h_3]
\]
where $\Theta^S_{X,Y,Z}$ is a map in $\Set$:
\[
  \Theta_{X,Y,Z} : X \times  (Y + Z) \rightarrow 
  (X \times  Y) + (X \times  Z).
\] 
% SUM INTRO
\item[\textbf{Case ($+$ I$_L$).}] 
Given the map 
\[
  h = \pdenotid{\Phi,\Gamma \vdash e : \sigma}:
  \pdenot{\Phi} \times \pdenot{\Gamma} 
  \rightarrow {\pdenot{\sigma}}
\]
in $\Set$, define the map
\[
\pdenotid{\Phi,\Gamma \vdash \inl e: \sigma + \tau} 
\] 
as the composition 
\[
  h;in_1.
\]
% SUM INTRO
\item[\textbf{Case ($+$ I$_R$).}] 
Given the map 
\[
  h = \pdenotid{\Phi,\Gamma \vdash e : \sigma}: \pdenot{\Phi} \times
  \pdenot{\Gamma} \rightarrow {\pdenot{\sigma}}
\]
in $\Set$, define the map
\[
  \pdenotid{\Phi,\Gamma \vdash \inr e: \sigma + \tau} 
\] 
as the composition 
\[
  h;in_2.
\]
% Let 
\item[\textbf{Case (Let).}] 
Given the maps 
\begin{align*}
h_1 &= \pdenotid{\Phi,\Gamma \vdash e : \sigma} 
  : \pdenot{\Phi} \times \pdenot{\Gamma} \rightarrow \pdenot{\sigma} \\
h_2 &= \pdenotid{\Phi,\Delta, x: \sigma \vdash f : \tau}
  : \pdenot{\Phi} \times \pdenot{\Delta} \times \pdenot{\sigma} 
      \rightarrow \pdenot{\tau} 
\end{align*}
in $\Set$, define the map 
\[
  \pdenotid{\Phi,\Gamma, \Delta \vdash \slet x e f : \tau}
\]
as the composition 
\[
  (t_{\pdenot{\Phi}},
  id_{\pdenot{\Gamma} \times \pdenot{\Delta}}); 
  (h_1, id_{\pdenot{\Phi} \times \pdenot{\Delta}}); h_2.
\]
% OP
\item[\textbf{Case (Op).}] 
Given the maps
\begin{align*}
h_1 &= \pdenotid{\Phi,\Gamma \vdash e : \num}: 
  \pdenot{\Phi} \times \pdenot{\Gamma} \rightarrow {\pdenot{\num}}\\
h_2 &= \pdenotid{\Phi,\Delta \vdash f : \num} :
  \pdenot{\Phi} \times \pdenot{\Delta} \rightarrow {\pdenot{\num}}
\end{align*}
in $\Set$, define the map
\[
  \pdenotid{\Phi,\Gamma,\Delta \vdash \mathbf{Op} ~e ~f : \num} 
\]
as the composition 
\[
  (t_{\pdenot{\Phi}},id_{\pdenot{\Gamma}\times\pdenot{\Delta}});(h_1, h_2); U_{id} \mathcal{L}_{op} =
  (t_{\pdenot{\Phi}},id_{\pdenot{\Gamma}\times\pdenot{\Delta}});(h_1, h_2); f_{op}
\]
for $\mathbf{Op} \in \{\mathbf{add},\mathbf{sub},\mathbf{mul}\}$.
% Div
\item[\textbf{Case (Div).}] Given the maps
\begin{align*}
h_1 &= \pdenotid{\Phi,\Gamma \vdash e : \num}: 
  \pdenot{\Phi} \times \pdenot{\Gamma} \rightarrow {\pdenot{\num}}\\
h_2 &= \pdenotid{\Phi,\Delta \vdash f : \num} :
  \pdenot{\Phi} \times \pdenot{\Delta} \rightarrow {\pdenot{\num}}
\end{align*}
in $\Set$, define the map
\[
  \pdenotid{\Phi,\Gamma,\Delta \vdash \textbf{div} ~ e ~f : \num + \unit}
\]
as the composition
\[
  (t_{\pdenot{\Phi}},id_{\pdenot{\Gamma}\times\pdenot{\Delta}});(h_1, h_2); U_{id} \mathcal{L}_{div} =
  (t_{\pdenot{\Phi}},id_{\pdenot{\Gamma}\times\pdenot{\Delta}});(h_1, h_2); f_{div}.
\]
\end{description}

\section{Proof of Backward Error Soundness}\label{sec:app_soundness}
This appendix provides a detailed proof of the main backward error soundness
theorem for \bea{} (\Cref{thm:main}).

\begin{namedtheorem}[\Cref{thm:main}]
Let $ \Phi\mid  x_1:_{r_1}\sigma_1,\cdots,x_n:_{r_n}\sigma_n = 
    \Gamma \vdash e : \sigma$ be a well-typed \bea{} term. Then for 
any well-typed substitutions 
$\bar{p} \vDash \Phi^\circ$ and $\bar{k} \vDash \Gamma^\circ$, if 
\[\Lambda(e)[\bar{p}/\dom(\Phi^\circ)][\bar{k}/\dom(\Gamma^\circ)] \stepap v\] 
for some value $v$, then the well-typed substitution $\bar{l} \vDash \Gamma^\circ$ 
exists such that \[\Lambda(e)[\bar{p}/\dom(\Phi^\circ)][\bar{l}/\dom(\Gamma^\circ)] \stepid v,\] and 
$d_{\denot{\sigma_i}}({k}_i,{l}_i) \le r_i$ for each $k_i \in \bar{k}$ 
and $l_i \in \bar{l}$.
\end{namedtheorem}

\begin{proof}
From the lens semantics (\Cref{def:interpL}) of \bea{} we have the triple 
\[ \denot{\Phi\mid \Gamma \vdash e : \sigma} = (f,\tilde{f},b) : 
   \denot{\Phi} \otimes  \denot{\Gamma} \rightarrow \denot{\sigma}. \]
Then, using the backward map $b$, we can define the tuple of vectors 
of values 
$\label{eq:bmapdef}
	(\bar{s},\bar{l}) \triangleq b((\bar{p},\bar{k}),v)
$
such that $\bar{s} \vDash \Phi^\circ$ and $\bar{l} \vDash \Gamma^\circ$. 

From the second property of backward error lenses we then have  
$$f(\bar{s},\bar{l}) = f(b((\bar{p},\bar{k}),v)) = v.$$

We can now show a backward error result, i.e.,
$\tilde{f}(\bar{p},\bar{k}) = f(\bar{s},\bar{l})$:
\begin{align}
\pdenotap{\Phi,\Gamma \vdash \Lambda(e) : \sigma} (\bar{p},\bar{k})
	&= \Uap{\denot{\Phi \mid \Gamma \vdash e : \sigma}} 
		\tag*{(\Cref{lem:pairing})}(\bar{p},\bar{k}) \\
	&=\tilde{f}(\bar{p},\bar{k})\tag*{(\Cref{def:interpL})}\\
	&= v \tag*{(\Cref{thm:soundid})} \\ 
	&= f(\bar{s},\bar{l}). \tag*{}
\end{align}

From the first property of error lenses we have 
\[
  d_{\denot{\Phi}\otimes \denot{\Gamma}}( (\bar{p},\bar{k}), 
    b((\bar{p},\bar{k}),v)) -r_{\denot{\Phi}\otimes\denot{\Gamma}}
  \le d_{\denot{\sigma}}( \tilde{f}(\bar{p},\bar{k}), v)-r_{\denot{\sigma}}
\]
so long as  
\begin{align}
d_{\denot{\sigma}}( \tilde{f}(\bar{p},\bar{k}), v) 
 =d_{\denot{\sigma}}( v, v) <\infty.
  \label{eq:error_assum}
\end{align}
If the base numeric type is interpreted as a metric
space with a standard distance function, then 
$d_{\denot{\sigma}}\left( v, v\right)$ is zero for any type 
$\sigma$, and so \Cref{eq:error_assum} is satisfied.

Unfolding definitions, and using the fact that 
$\tilde{f}(\bar{p},\bar{k}) = v$ from above, we have 
\begin{align}
\max\{d_{\denot{\Phi}}(\bar{p},\bar{s}), 
    d_{\denot{\Gamma}}({\bar{k},\bar{l}})-r_{\denot{\Gamma}}\} &\le 
    d_{\denot{\sigma}}\left( v, v\right)=0. \label{eq:error_ineq}
\end{align}
From \Cref{eq:error_ineq} we can conclude two things. First, using the
definition of the distance function on discrete metric spaces, we can conclude
$\bar{p}=\bar{s}$: the discrete variables carry no backward error.  Second, for
linear variables, by \Cref{eq:d_tensor} we can derive the required backward error bound:
\begin{align*} 
    \max\left\{ d_{\denot{\sigma_1}}(k_1,l_1)-r_1,\dots,
     d_{\denot{\sigma_n}}(k_n,l_n)-r_n\right\} &\le 0.
\end{align*}
\end{proof}

\section{Type Checking Algorithm and Proofs of Soundness and Completeness}\label{app:algorithm}
\begin{figure}
\begin{center}

\begin{mathpar}
% var
\inferrule*[right=(Var)] 
{ }
{\Phi\mid\Gamma^\bullet,x:\sigma;x\Rightarrow  \{x:_0\sigma\};\sigma}

% dvar
\inferrule*[right=(DVar)] 
{ }
{\Phi,z:\alpha\mid\Gamma^\bullet;z\Rightarrow \emptyset; \alpha}

% prod intro
\inferrule*[right=($\otimes $ I)] 
{
\Phi\mid\Gamma^\bullet;e\Rightarrow \Gamma_1; \sigma \\ 
\Phi\mid\Gamma^\bullet;f\Rightarrow \Gamma_2; \tau \\
\dom\Gamma_1\cap\dom\Gamma_2=\emptyset
}
{\Phi\mid\Gamma^\bullet; (e,f)\Rightarrow \Gamma_1,\Gamma_2; \sigma\otimes\tau}

% unit
\inferrule*[right=(Unit)] 
{ }
{\Phi\mid\Gamma^\bullet;()\Rightarrow \emptyset; \mathbf{unit}}

% prod elim
\inferrule*[right=($\otimes $ E$_\sigma$)] 
{ 
    \Phi\mid\Gamma^\bullet;e\Rightarrow \Gamma_1;\tau_1\otimes\tau_2 \\
    \Phi\mid\Gamma^\bullet,x:\tau_1, y:\tau_2; f\Rightarrow \Gamma_2 ; \sigma \\
    \dom\Gamma_1\cap\dom\Gamma_2=\emptyset \\
    r=\max\{r_1,r_2\} \text{ where }
        x:_{r_1}\tau_1,y:_{r_2}\tau_2\in\Gamma_2 \text{ else } r=0
}
{
    \Phi\mid \Gamma^\bullet; \slet {(x,y)} e f \Rightarrow 
    (r+\Gamma_1),\Gamma_2\setminus \{x,y\}; \sigma
}

% prod elim discrete
\inferrule*[right=($\otimes $ E$_\alpha$)] 
{ 
    \Phi\mid\Gamma^\bullet;e\Rightarrow  \Gamma_1;\alpha_1\otimes \alpha_2 \\
    \Phi,z_1: \alpha_1,z_2: \alpha_2\mid\Gamma^\bullet;f\Rightarrow \Gamma_2;\sigma \\
    \dom\Gamma_1\cap\dom\Gamma_2=\emptyset 
}
{\Phi\mid \Gamma^\bullet; \dlet {(z_1,z_2)} e f \Rightarrow \Gamma_1,\Gamma_2;\sigma}

% case
\inferrule*[right=($+$ E)] 
{
    \Phi\mid\Gamma^\bullet;e'\Rightarrow\Gamma_1;\sigma+\tau \\
    \Phi\mid \Gamma^\bullet,x:\sigma;e\Rightarrow \Gamma_2;\rho \\
    \Phi\mid\Gamma^\bullet,y:\tau;f\Rightarrow \Gamma_3;\rho \\
    \dom\Gamma_1\cap\dom\Gamma_2=\dom\Gamma_1\cap\dom\Gamma_3=\emptyset \\
    q=\max\{q_1,q_2\} \text{ where } x:_{q_1}\sigma\in\Gamma_2
    \text{ or } y:_{q_2}\tau\in\Gamma_3 \text{ else }q=0
}
{
    \Phi\mid\Gamma^\bullet;\mathbf{case} \ e' \ \mathbf{of} \ (\inl x.e \ | \ \inr y.f) 
    \Rightarrow (q+\Gamma_1),\max\{\Gamma_2\setminus \{x\}, \Gamma_3\setminus\{y\}\};\rho\
}

% ind sum intro
\inferrule*[right=($+$ $\text{I}_L$)] 
{\Phi\mid \Gamma^\bullet;e\Rightarrow \Gamma;\sigma}
{\Phi\mid\Gamma^\bullet; \inl_\tau \ e \Rightarrow\Gamma;\sigma + \tau}

% ind sum intro
\inferrule*[right=($+$ $\text{I}_R$)] 
{\Phi\mid \Gamma^\bullet;e\Rightarrow\Gamma;\tau}
{\Phi\mid\Gamma^\bullet; \inr_\sigma \ e \Rightarrow \sigma + \tau}

% let
\inferrule*[right=(Let)] 
{
    \Phi\mid\Gamma^\bullet;e\Rightarrow\Gamma_1;\tau \\
    \Phi\mid\Gamma^\bullet,x:\tau;f\Rightarrow\Gamma_2;\sigma \\\\
    \dom\Gamma_1\cap\dom\Gamma_2=\emptyset \\
    x:_r\sigma\in\Gamma_2 \text{ else } r=0
}
{\Phi\mid \Gamma^\bullet; \slet x e f\Rightarrow (r+\Gamma_1),\Gamma_2\setminus\{x\}; \sigma}

% disc
\inferrule*[right=(Disc)] 
{\Phi\mid\Gamma^\bullet;e\Rightarrow\Gamma;\sigma}
{\Phi\mid\Gamma^\bullet;!e\Rightarrow \Gamma;m(\sigma)}

% dlet
\inferrule*[right=(DLet)] 
{
    \Phi\mid \Gamma^\bullet;e\Rightarrow\Gamma_1;\alpha \\
    \Phi,z: \alpha\mid \Gamma^\bullet;f\Rightarrow\Gamma_2;\sigma \\\\
    \dom\Gamma_1\cap\dom\Gamma_2=\emptyset 
}
{\Phi\mid\Gamma^\bullet; \dlet{z}{e}{f}\Rightarrow \Gamma_1,\Gamma_2;\sigma}

% add, sub
\inferrule*[right=\text{(Add, Sub)}] 
{ }
{
    \Phi\mid \Gamma^\bullet, x: \num, y:\num;
    \{\mathbf{add}, \mathbf{sub}\} \ x \ y \Rightarrow 
    \{x:_{\varepsilon}\num,y:_{\varepsilon}\num\};\num
}

% mul
\inferrule*[right=(Mul)] 
{ }
{
    \Phi\mid \Gamma^\bullet, 
    x:\num, 
    y:\num; \mul  x  y \Rightarrow 
    \{x:_{\varepsilon/2}\num, y:_{\varepsilon/2}\num\};\num 
}

% div
\inferrule*[right=(Div)] 
{ }
{
    \Phi\mid \Gamma^\bullet, x: \num, y:\num;
    \{\mathbf{add}, \mathbf{sub}\} \ x \ y \Rightarrow 
    \{x:_{\varepsilon}\num,y:_{\varepsilon}\num\};\num + \unit
}

% dmul
\inferrule*[right=(DMul)] 
{ }
{
    \Phi,z: m(\num)\mid \Gamma^\bullet, 
    x:\num; \dmul  z  x \Rightarrow \{x:_\varepsilon \num\}; \num
}
\end{mathpar}

\end{center}
    \caption{Type checking algorithm for \bea{}.}
    \label{fig:algorithm}
\end{figure}

This appendix defines the type checking algorithm for \bea{} described in 
\Cref{sec:algorithm}, as well as proofs of its soundness and completeness.
First, we give the full type checking algorithm in \Cref{fig:algorithm}.
Recall that algorithm calls are written as $\Phi\mid\Gamma^\bullet;e\Rightarrow \Gamma;\sigma$
where $\Gamma^\bullet$ is a linear context skeleton, $e$ is a \bea{} program,
$\Gamma$ is, intuitively, the \emph{minimal} linear context required to type $e$ such that
$\overline{\Gamma}\sqsubseteq\Gamma^\bullet$, and $\sigma$ is the type of $e$. 
Note that we only require $\Phi$ to contain the discrete variables used in the program 
and we do nothing more; thus, it is not returned by the algorithm.
We do require that discrete and linear contexts are always disjoint, and we will denote
linear variables by $x$ and $y$ and discrete variables by $z$. 
Finally, we define the \emph{max} of two linear contexts, $\max\{\Gamma,\Delta\}$, 
to have domain $\dom\Gamma\cup\dom\Delta$ and, if $x:_q\sigma\in\Gamma$ and $x:_r\sigma\in\Delta$, 
then $x:_{\max\{q,r\}}\sigma\in\max\{\Gamma,\Delta\}$. 

Before we give proofs of \Cref{thm:algo_sound} and \Cref{thm:algo_complete},
we must prove two lemmas about type system and algorithm weakening. 
Intuitively, type system weakening says that if we can derive the type of a program from a context $\Gamma$, 
then we can also derive the same program from a larger context $\Delta$ which subsumes $\Gamma$.
\begin{lemma}[Type System Weakening]\label{lem:type_weak}
  If $\Phi\mid\Gamma\vdash e:\sigma$ and $\Gamma\sqsubseteq\Delta$, then $\Phi\mid\Delta\vdash e:\sigma$.
\end{lemma}
\begin{proof}
  Suppose $\Phi\mid\Gamma\vdash e:\sigma$. We proceed by induction on the typing derivation 
  and consider some representative cases of the final typing rule applied.
  \begin{description}
    \item[\textbf{Case (Var).}] Suppose the last rule applied was 
    \[
      \Phi\mid\Gamma,x:_r\sigma\vdash x:\sigma.
    \]
    Let $\Delta$ be a context such that $(\Gamma,x:_r\sigma)\sqsubseteq\Delta$. Thus, $x:_q\sigma\in\Delta$ 
    where $r\leq q$. By the same rule, $\Phi\mid\Delta\vdash x:\sigma$. 

    \item[\textbf{Case ($\otimes $ I).}] Suppose the last rule applied was
    \[
      \Phi\mid\Gamma,\Delta\vdash (e,f):\sigma\otimes\tau
    \]
    and thus, we also have that 
    \[
      \Phi\mid\Gamma\vdash e:\sigma\text{ and }\Phi\mid\Delta\vdash f:\tau.
    \]
    Let $\Lambda$ be a context such that $(\Gamma,\Delta)\sqsubseteq\Lambda$.
    As $\Gamma$ and $\Delta$ are disjoint, we can split $\Lambda$ into the contexts 
    $\Gamma_1$ and $\Delta_1$ such that $\Gamma\sqsubseteq\Gamma_1$ and $\Delta\sqsubseteq\Delta_1$.
    By our inductive hypothesis, it follows that 
    \[
      \Phi\mid\Gamma_1\vdash e:\sigma\text{ and }\Phi\mid\Delta_1\vdash f:\tau.
    \]
    By the same rule, we conclude that 
    \[
      \Phi\mid\Gamma_1,\Delta_1\vdash (e,f):\sigma\otimes\tau.
    \]

    \item[\textbf{Case ($\otimes $ E$_\sigma$).}] Suppose the last rule applied 
    was 
    \[
      \Phi\mid r+\Gamma,\Delta\vdash\slet{(x, y)}e f:\sigma.
    \]
    Let $\Lambda$ be a context such that $(r+\Gamma,\Delta)\sqsubseteq\Lambda$ and 
    $x,y\not\in\dom\Lambda$. As before, split $\Lambda$ into contexts $\Gamma_1$ and $\Delta_1$
    such that $(r+\Gamma)\sqsubseteq\Gamma_1$ and $\Delta\sqsubseteq\Delta_1$ but where 
    $\dom\Gamma=\dom\Gamma_1$. Now, for each $x\in\dom\Gamma_1$, we have that $x:_q\sigma\in\Gamma_1$
    where $r\leq q$. Therefore, we can define the context $-r+\Gamma_1$ which subtracts $r$ from the 
    error bound of every variable in $\Gamma_1$, and hence $\Gamma\sqsubseteq (-r+\Gamma_1)$. 
    Finally, use our inductive hypothesis to get that
    \[
      \Phi\mid(-r+\Gamma_1)\vdash e:\tau_1\otimes\tau_2\text{ and }
      \Phi\mid\Delta_1,x:_r\tau_1,y:_r\tau_2\vdash f:\sigma
    \]
    and we can apply the same rule to get our conclusion.

    \item[\textbf{Case (Add).}] Suppose the last rule applied was 
    \[
      \Phi\mid\Gamma,x:_{\varepsilon+r_1}\num,y:_{\varepsilon+r_2}\num\vdash\add{x}{y}:\num
    \]
    Let $\Delta$ be a context such that
    $(\Gamma,x:_{\varepsilon+r_1}\num,y:_{\varepsilon+r_2}\num)\sqsubseteq\Delta$.
    Hence, $x:_{q_1}\num,y:_{q_2}\num\in\Delta$ where $\varepsilon+r_1\leq q_1$ and $\varepsilon+r_2\leq q_2$.
    Rewrite $q_1=\varepsilon + (q_1-\varepsilon)$ and $q_2=\varepsilon+(q_2-\varepsilon)$ and apply the same rule.
  \end{description}
\end{proof}
Similarly, algorithmic weakening says that if we pass a context skeleton $\Gamma^\bullet$ into the algorithm and it
infers context $\Gamma$, then if we pass in a larger skeleton $\Delta^\bullet$, the algorithm will 
still infer context $\Gamma$. (Here, we extend the notion of subcontexts to context skeletons, where
$\Gamma^\bullet \sqsubseteq\Delta^\bullet$ if $\Gamma^\bullet\subseteq\Delta^\bullet$.)
This is because the algorithm discards unused variables from the context.
\begin{lemma}[Type Checking Algorithmic Weakening]\label{lem:algo_weak}
  If $\Phi\mid\Gamma^\bullet;e\Rightarrow\Gamma;\sigma$ and $\Gamma^\bullet\sqsubseteq\Delta^\bullet$, 
  then $\Phi\mid\Delta^\bullet;e\Rightarrow \Gamma;\sigma$.
\end{lemma}
\begin{proof}
  Suppose that $\Phi\mid\Gamma^\bullet;e\Rightarrow\Gamma;\sigma$. 
  We proceed by induction on the algorithmic derivation 
  and consider some representative cases of the final algorithmic rule applied.
  \begin{description}
    \item[\textbf{Case (Var).}] Suppose the last step applied was 
    \[
      \Phi\mid\Gamma^\bullet,x:\sigma;x\Rightarrow\{x:_0\sigma\};\sigma.
    \]
    Let $\Delta^\bullet$ be a context skeleton such that 
    $(\Gamma^\bullet,x:\sigma)\sqsubseteq\Delta^\bullet$. Thus, $x:\sigma\in\Delta^\bullet$
    so we can apply the same rule.

    \item[\textbf{Case ($\otimes $ E$_\sigma$).}] Suppose the last step applied 
    was 
    \[
      \Phi\mid\Gamma^\bullet;\slet{(x, y)} e f\Rightarrow 
      (r+\Gamma_1),\Gamma_2\setminus\{x,y\};\sigma.
    \]
    Let $\Delta^\bullet$ be a context skeleton such that $\Gamma^\bullet\sqsubseteq\Delta^\bullet$
    and $x,y\not\in\Delta^\bullet$. By induction, we have that 
    \[
      \Phi\mid\Delta^\bullet;e\Rightarrow\Gamma_1;\tau_1\otimes\tau_2\text{ and }
      \Phi\mid\Delta^\bullet,x:\tau_1,y:\tau_2;f\Rightarrow\Gamma_2;\sigma
    \]
    and $\dom\Gamma_1\cap\dom\Gamma_2=\emptyset$. By the same rule, we conclude that
    \[
      \Phi\mid\Delta^\bullet;\slet{(x,y)}{e}{f}\Rightarrow (r+\Gamma_1),\Gamma_2\setminus\{x,y\};\sigma.
    \]
  \end{description}
\end{proof}
Finally, we give proofs of algorithmic soundness and completeness. 
Soundness states that if the algorithm returns a linear context $\Gamma$, then we can use $\Gamma$
to derive the program using \bea{}'s type system.
\begin{namedtheorem}[\Cref{thm:algo_sound}]
  If $\Phi\mid\Gamma^\bullet;e\Rightarrow \Gamma;\sigma$, then 
  $\overline{\Gamma}\sqsubseteq\Gamma^\bullet$ and the derivation 
  $\Phi\mid\Gamma\vdash e:\sigma$ exists.
\end{namedtheorem}
\begin{proof}
  Suppose that $\Phi\mid\Gamma^\bullet;e\Rightarrow\Gamma;\sigma$. 
  We proceed by induction on the algorithmic derivation 
  and consider some representative cases of the final algorithmic rule applied.
  We use the fact that if $\Gamma,\Delta$ are disjoint,
  then $\overline{\Gamma,\Delta}=\overline\Gamma,\overline\Delta$. 
  \begin{description}
    \item[\textbf{Case (Var).}] Suppose the last step applied was 
    \[
      \Phi\mid\Gamma^\bullet,x:\sigma;x\Rightarrow\{x:_0\sigma\};\sigma.
    \]
    By the typing rule (Var$_\sigma$), we have that $\Phi\mid\{x:_0\sigma\}\vdash x:\sigma$.
    Moreover,
    $\overline{\{x:_0\sigma\}}\sqsubseteq(\Gamma^\bullet,x:\sigma)$.

    \item[\textbf{Case ($\otimes $ I).}] Suppose the last step applied was
    \[
      \Phi\mid\Gamma^\bullet;(e,f)\Rightarrow\Gamma_1,\Gamma_2;\sigma\otimes\tau
    \]
    where 
    \[
      \Phi\mid\Gamma^\bullet;e\Rightarrow\Gamma_1;\sigma\text{ and }
      \Phi\mid\Gamma^\bullet;f\Rightarrow\Gamma_2;\sigma
    \]
    and $\dom\Gamma_1\cap\dom\Gamma_2=\emptyset$. By our inductive hypothesis, 
    we have that $\Phi\mid\Gamma_1\vdash e:\sigma$ and $\Phi\mid\Gamma_2\vdash f:\tau$.
    Therefore, we can apply the typing rule ($\otimes$ I) to get that 
    \[
      \Phi\mid\Gamma_1,\Gamma_2\vdash (e,f):\sigma\otimes\tau.
    \]
    Finally, as $\overline{\Gamma_1}\sqsubseteq\Gamma^\bullet$ and 
    $\overline{\Gamma_2}\sqsubseteq\Gamma^\bullet$,
    we have that $\overline{\Gamma_1,\Gamma_2}\sqsubseteq\Gamma^\bullet$.

    \item[\textbf{Case ($\otimes $ E$_\sigma$).}] Suppose the last step applied 
    was 
    \[
      \Phi\mid\Gamma^\bullet;\slet{(x, y)} e f\Rightarrow 
      (r+\Gamma_1),\Gamma_2\setminus\{x,y\};\sigma
    \]
    By induction, we have that
    \[
      \Phi\mid\Gamma_1\vdash e:\tau_1\otimes\tau_2\text{ and }\Phi\mid\Gamma_2\vdash f:\sigma,
    \]  
    where $x,y$ may be in $\dom\Gamma_2$.
    Let $\Delta=\Gamma_2\setminus\{x,y\}$. 
    Since $r$ is defined to be the maximum of the bounds on $x,y$ if they exist in 
    $\Gamma_2$, we have that $\Gamma_2\sqsubseteq(\Delta,x:_r\tau_1,y:_r\tau_2)$. 
    From \Cref{lem:type_weak}, it follows that
    \[
      \Phi\mid\Delta,x:_r\tau_1,y:_r\tau_2\vdash f:\sigma.
    \]
    Thus, we can apply the typing rule ($\otimes$ E$_\sigma$) to conclude that 
    \[
      \Phi\mid r+\Gamma_1,\Delta\vdash \slet{(x, y)}e f:\sigma.
    \]
    Finally, as $\overline{\Gamma_1}\sqsubseteq\Gamma^\bullet$ and 
    $\overline{\Gamma_2}\sqsubseteq(\Gamma^\bullet,x:\tau_1,y:\tau_2)$, we have that
    \[
      \overline{r+\Gamma_1,\Delta}=\overline{\Gamma_1,\Gamma_2\setminus\{x,y\}}
      =\overline{\Gamma_1},\overline{\Gamma_2\setminus\{x,y\}}\sqsubseteq\Gamma^\bullet.
    \]

    \item[\textbf{Case (+ E).}] Suppose the last step applied was 
    \[
      \Phi\mid\Gamma^\bullet;\mathbf{case} \ e' \ \mathbf{of} \ (\inl x.e \ | \ \inr y.f) 
    \Rightarrow (q+\Gamma_1),\max\{\Gamma_2\setminus \{x\}, \Gamma_3\setminus\{y\}\};\rho.
    \]
    By induction, we have that 
    \[
      \Phi\mid\Gamma_1\vdash e':\sigma+\tau\text{ and }
      \Phi\mid\Gamma_2\vdash e:\rho\text{ and }\Phi\mid\Gamma_3\vdash f:\rho.
    \]
    Let $\Delta=\max\{\Gamma_2\setminus\{x\}, \Gamma_3\setminus\{y\}\}$, and we
    still have that $\dom\Gamma_1\cap\dom\Delta=\emptyset$. By \Cref{lem:type_weak}, 
    it follows that 
    \[
      \Phi\mid\Delta,x:_q\sigma\vdash e:\rho\text{ and }\Phi\mid\Delta,y:_q\tau\vdash f:\rho
    \]
    by weakening the bounds on $x$ and $y$ to $q$. Thus, we can apply typing rule (+ E) to 
    conclude that 
    \[
      \Phi\mid q+\Gamma_1,\Delta\vdash \case{e'}{\inl x.e}{\inr y.f}:\rho.
    \]
    Moreover, as $\overline{\Gamma_1}\sqsubseteq\Gamma^\bullet$ and 
    $\overline{\Gamma_2}\sqsubseteq(\Gamma^\bullet,x:\sigma)$ and 
    $\overline{\Gamma_3}\sqsubseteq(\Gamma^\bullet,y:\tau)$, we have that 
    \[
      \overline{q+\Gamma_1,\Delta}=\overline{\Gamma_1},
      \overline{\max\{\Gamma_2\setminus\{x\},\Gamma_3\setminus\{y\}\}}\sqsubseteq\Gamma^\bullet.
    \]

    \item[\textbf{Case (Add).}] Suppose the last step applied was 
    \[
      \Phi\mid \Gamma^\bullet, x: \num, y:\num;
      \mathbf{add} \ x \ y \Rightarrow 
      \{x:_{\varepsilon}\num,y:_{\varepsilon}\num\};\num.
    \]
    By the typing rule (Add) we have that 
    \[
      \Phi\mid\{x:_\varepsilon\num,y:_\varepsilon\num\}\vdash\mathbf{add}~x~y:\num.
    \]
  \end{description}
\end{proof}
Conversely, completeness says that if from $\Gamma$ we can derive the type of a program $e$, then
inputting $\overline{\Gamma}$ and $e$ into the algorithm will yield a valid output. 
\begin{namedtheorem}[\Cref{thm:algo_complete}]
  If $\Phi\mid\Gamma\vdash e:\sigma$ is a valid derivation in \bea{}, then there
  exists a context $\Delta\sqsubseteq\Gamma$ such that 
  $\Phi\mid\overline{\Gamma}; e\Rightarrow \Delta;\sigma$.
\end{namedtheorem}
\begin{proof}
  Suppose that $\Phi\mid\Gamma\vdash e:\sigma$.
  We proceed by induction on the typing derivation 
  and consider some representative cases of the final typing rule applied.
  \begin{description}
    \item[\textbf{Case (Var$_\sigma$).}] Suppose the last rule applied was 
    \[
      \Phi\mid\Gamma,x:_r\sigma\vdash x:\sigma.
    \]
    By algorithm step (Var), we have that
    \[
      \Phi\mid\overline{\Gamma},x:\sigma;x\Rightarrow\{x:_0\sigma\};\sigma
    \]
    and $\{x:_0\sigma\}\sqsubseteq(\Gamma,x:_r\sigma)$ as $0\leq r$. 

    \item[\textbf{Case ($\otimes $ I).}] Suppose the last rule applied was
    \[
      \Phi\mid\Gamma,\Delta\vdash (e,f):\sigma\otimes\tau.
    \]
    From this, we deduce that $\dom\Gamma\cap\dom\Delta=\emptyset$. By induction, there exist
    $\Gamma_1\sqsubseteq\Gamma$ and $\Delta_1\sqsubseteq\Delta$ such that 
    \[
      \Phi\mid\overline\Gamma;e\Rightarrow\Gamma_1;\sigma\text{ and }
      \Phi\mid\overline\Delta;f\Rightarrow\Delta_1;\tau.
    \]
    Moreover, $\dom\Gamma_1\cap\dom\Delta_1=\emptyset$ as well. 
    By \Cref{lem:algo_weak}, we also have that 
    \[
      \Phi\mid\overline{\Gamma,\Delta};e\Rightarrow\Gamma_1;\sigma\text{ and }
      \Phi\mid\overline{\Gamma,\Delta};f\Rightarrow\Delta_1;\tau.
    \]
    Thus, we can apply algorithm step
    ($\otimes$ I) to conclude
    \[
      \Phi\mid\overline{\Gamma,\Delta};(e,f)\Rightarrow \Gamma_1,\Delta_1;\sigma\otimes\tau
    \] 
    where we know $(\Gamma_1,\Delta_1)\sqsubseteq(\Gamma,\Delta)$. 

    \item[\textbf{Case ($\otimes $ E$_\sigma$).}] Suppose the last rule applied 
    was 
    \[
      \Phi\mid r+\Gamma,\Delta\vdash\slet{(x, y)}e f:\sigma.
    \]
    By induction, there exist $\Gamma_1\sqsubseteq\Gamma$ and 
    $\Delta_1\sqsubseteq(\Delta,x:_r\tau_1,y:_r\tau_2)$ such that 
    \[
      \Phi\mid \overline{\Gamma};e\Rightarrow \Gamma_1;\tau_1\otimes\tau_2\text{ and }
      \Phi\mid \overline{\Delta},x:\tau_1,y:\tau_2;f\Rightarrow\Delta_1;\sigma.
    \]
    If $x:_{r_1}\tau_1, y:_{r_2}\tau_2\in\Delta_1$, let $r'=\max\{r_1,r_2\}$. As 
    $\Delta_1\sqsubseteq(\Delta,x:_r\tau_1,y:_r\tau_2)$, we know $r'\leq r$. 

    Using \Cref{lem:algo_weak}, we can apply algorithm step ($\otimes$ E$_\sigma$) of
    \[
      \Phi\mid\overline{\Gamma,\Delta};\slet{(x, y)}e f \Rightarrow (r'+\Gamma_1),\Delta_1\setminus\{x,y\};\sigma.
    \]
    Moreover, $(r'+\Gamma_1)\sqsubseteq(r+\Gamma)$ and $(\Delta_1\setminus\{x,y\})\sqsubseteq\Delta$.

    \item[\textbf{Case (+ E).}] Suppose the last rule applied was 
    \[
      \Phi\mid q+\Gamma,\Delta\vdash \case{e'}{\inl x.e}{\inr y.f}:\rho.
    \]
    By induction, there exist $\Gamma_1\sqsubseteq\Gamma$, $\Delta_1\sqsubseteq(\Delta,x:_q\sigma)$, and 
    $\Delta_2\sqsubseteq(\Delta,y:_q\tau)$ such that 
    \[
      \Phi\mid\overline{\Gamma};e'\Rightarrow \Gamma_1;\sigma+\tau\text{ and }
      \Phi\mid\overline{\Delta},x:\sigma;e\Rightarrow \Delta_1;\rho\text{ and }
      \Phi\mid\overline{\Delta}, y:\tau;f\Rightarrow\Delta_2;\rho.
    \]
    If $x:_{q_1}\sigma\in\Delta_1$ and $y:_{q_2}\tau\in\Delta_2$, let $q'=\max\{q_1,q_2\}$,
    and we know $q'\leq q$. Using \Cref{lem:algo_weak}, we can apply algorithm step (+ E) of
    \[  
      \Phi\mid\overline{\Gamma,\Delta};\case{e'}{\inl x.e}{\inr y.f}\Rightarrow (q'+\Gamma_1),
      \max\{\Delta_1\setminus\{x\},\Delta_2\setminus\{y\}\};\rho.
    \]
    Furthermore, we know $(q'+\Gamma_1)\sqsubseteq(q+\Gamma)$ and 
    $\max\{\Delta_1\setminus\{x\},\Delta_2\setminus\{y\}\}\sqsubseteq \Delta$.

    \item[\textbf{Case (Add).}] Suppose the last rule applied was 
    \[
      \Phi\mid\Gamma,x:_{\varepsilon+r_1}\num,y:_{\varepsilon+r_2}\num\vdash\add{x}{y}:\num
    \]
    where $r_1,r_2\geq 0$. We can apply algorithm step (Add) of
    \[
      \Phi\mid\overline{\Gamma},x:\num,y:\num;\add{x}{y}\Rightarrow \{x:_\varepsilon\num,y:_\varepsilon\num\};\num
    \]
    and we have $\{x:_\varepsilon\num,y:_\varepsilon\num\}\sqsubseteq
    (\Gamma,x:_{\varepsilon+r_1}\num,y:_{\varepsilon+r_2}\num)$.
  \end{description}
\end{proof}

\fi

\end{document}